\documentclass[epsfig,letter,11pt]{article}

\usepackage{epsfig}
\usepackage{plain}
\usepackage{setspace}
\usepackage[outer=1in,inner=1in,top=1in,bottom=1in]{geometry}
\usepackage{titlesec}

\usepackage{mathtools}
\usepackage{amssymb}
\usepackage{amsthm}
\usepackage{amsmath}
\usepackage{mathtools}
\usepackage{bm}
\usepackage{multicol}
\usepackage{pgfplots}
\usepackage{caption}
\usepackage{chngcntr}
\usepackage{tikz-network}
\usepackage{bbm}
\usepackage{hyperref}
\usepackage{verbatim}
\usepackage[toc,page]{appendix}
\usepackage{xcolor}
\usepackage{thm-restate}
\usepackage[capitalize]{cleveref}
\usepackage{algorithm}
\usepackage[indLines=true]{algpseudocodex}
\algnewcommand{\LeftComment}[1]{\State  \textcolor{gray}{ \(\triangleright\) \emph{#1}}}

\usepackage{multicol}

\usepackage{etoolbox,refcount}

\numberwithin{equation}{section}

\newcounter{countitems}
\newcounter{nextitemizecount}
\newcommand{\setupcountitems}{
	\stepcounter{nextitemizecount}
	\setcounter{countitems}{0}
	\preto\item{\stepcounter{countitems}}
}
\makeatletter
\newcommand{\computecountitems}{
	\edef\@currentlabel{\number\c@countitems}
	\label{countitems@\number\numexpr\value{nextitemizecount}-1\relax}
}
\newcommand{\nextitemizecount}{
	\getrefnumber{countitems@\number\c@nextitemizecount}
}
\newcommand{\previtemizecount}{
	\getrefnumber{countitems@\number\numexpr\value{nextitemizecount}-1\relax}
}
\makeatother

{\makeatletter
	\gdef\xxxmark{
		\expandafter\ifx\csname @mpargs\endcsname\relax
		\expandafter\ifx\csname @captype\endcsname\relax
		\marginpar{xxx}
		\else
		xxx 
		\fi
		\else
		xxx 
		\fi}
	\gdef\xxx{\@ifnextchar[\xxx@lab\xxx@nolab}
	\long\gdef\xxx@lab[#1]#2{{\bf \color{blue} [\xxxmark #2 ---{\sc #1}]}}
	\long\gdef\xxx@nolab#1{{\bf \color{blue} [\xxxmark #1]}}
}

\newcommand{\eps}{\epsilon}
\DeclareMathOperator{\ex}{\mathbb{E}}
\DeclareMathOperator{\unif}{UNIF}
\DeclareMathOperator{\vol}{vol}
\newcommand\vt[1]{\mathbf{#1}}
\newcommand{\nm}[1]{\left\|#1\right\|}

\newcommand{\otil}{\tilde{O}}

\newcommand{\unbor}[2]{\operatorname{bnd}_{#2} (#1)}
\newcommand{\wbor}[2]{\operatorname{bnd}_{#2}^w (#1)}

\newcommand{\unvolc}[2]{\operatorname{vol}_{#2}^\circ (#1)}
\newcommand{\wvolc}[2]{\operatorname{vol}_{#2}^{\circ \,w} (#1)}

\newcommand{\unspc}[2]{\Phi_{#2}^\circ (#1)}
\newcommand{\wspc}[2]{\Phi_{#2}^{\circ \,w} (#1)}

\newcommand{\wvol}[1]{\operatorname{vol}^{w} (#1)}

\newcommand{\uncut}[2]{{\partial}_{#2} #1}
\newcommand{\uncutg}[1]{{\partial} #1}
\newcommand{\wcut}[2]{\partial^w_{#2} #1}
\newcommand{\wcutg}[1]{\partial^w #1}

\newcommand{\iuncut}[2]{{\partial}_{#2}^{i} #1}
\newcommand{\iwcut}[2]{{\partial}_{#2}^{i \,w} #1}

\newcommand{\bsc}{\textsc{BalSparseCut}_{\alpha,\lambda}}
\newcommand{\ed}{\textsc{Decompose}}
\newcommand{\trim}{\textsc{Trim}}

\DeclareMathOperator{\ind}{\mathbbm{1}}

\renewcommand{\hat}{\widehat}
\renewcommand{\tilde}{\widetilde}

\newcommand{\card}[1]{|#1|}

\DeclareMathOperator*{\argmin}{argmin}

\DeclareMathOperator*{\poly}{poly}
\DeclareMathOperator*{\pylog}{polylog}

\DeclareMathOperator*{\ber}{Ber}

\DeclareMathOperator*{\low}{low}

\newcommand*{\supp}{\mathrm{supp}}

\DeclareMathOperator*{\good}{g}

\newcommand{\e}{\mathrm{e}}

\DeclareMathOperator*{\sw}{\textsc{bsc}}

\DeclareMathOperator*{\er}{ER}

\DeclareMathOperator*{\recover}{\textsc{recover}}

\newcommand*{\remove}{\mathrm{R}}

\newcommand{\I}{\mathcal{I}}
\newcommand{\T}{\mathcal{T}}

\newcommand*{\entropy}{\textup{\textsf{H}}}
\newcommand*{\info}{\textup{\textsf{I}}}

\newtheorem{mainthsimple}{Theorem}
\newtheorem{mainthactual}{Theorem}
\newtheorem{theorem}{Theorem}[section]
\newtheorem{definition}[theorem]{Definition}
\newtheorem{lemma}[theorem]{Lemma}
\newtheorem*{lemma*}{Lemma}
\newtheorem*{definition*}{Definition}

\newtheorem{claim}[theorem]{Claim}
\newtheorem{remark}{Remark}
\newtheorem{fact}[theorem]{Fact}

\newtheorem{obs}[theorem]{Observation}
\newtheorem{assump}{Assumption}
\newtheorem{reduct}[theorem]{Reduction}
\newtheorem*{reduct*}{Reduction}

\title{On the Streaming Complexity of Expander Decomposition}
\author{
	Yu Chen\\
	EPFL\\
	\and
	Michael Kapralov\\
	EPFL\\
	\and
	Mikhail Makarov\\
	EPFL\\
	\and
	Davide Mazzali\\
	EPFL\\}
\date{}

\begin{document}

\maketitle
\thispagestyle{empty}
\begin{abstract}
\noindent
In this paper we study the problem of finding $(\epsilon, \phi)$-expander decompositions of a graph in the streaming model, in particular for dynamic streams of edge insertions and deletions. The goal is to partition the vertex set so that every component induces a $\phi$-expander, while the number of inter-cluster edges is only an $\epsilon$ fraction of the total volume. It was recently shown that there exists a simple algorithm to construct a $(O(\phi \log n), \phi)$-expander decomposition of an $n$-vertex graph using \smash{$\widetilde{O}(n/\phi^2)$} bits of space [Filtser, Kapralov, Makarov, ITCS'23]. This result calls for understanding the extent to which a dependence in space on the sparsity parameter $\phi$ is inherent. We move towards answering this question on two fronts.

We prove that a $(O(\phi \log n), \phi)$-expander decomposition can be found using~\smash{$\widetilde{O}(n)$} space, for every~$\phi$.  At the core of our result is the first streaming algorithm for computing boundary-linked expander decompositions, a recently introduced strengthening of the classical notion [Goranci et al., SODA'21]. The key advantage is that a classical sparsifier [Fung et al., STOC'11], with size independent of $\phi$, preserves the cuts inside the clusters of a boundary-linked expander decomposition within a multiplicative error.

Notable algorithmic applications use sequences of expander decompositions, in particular one often repeatedly computes a decomposition of the subgraph induced by the inter-cluster edges (e.g., the seminal work of Spielman and Teng on spectral sparsifiers [Spielman, Teng, SIAM Journal of Computing 40(4)], or the recent maximum flow breakthrough [Chen et al., FOCS'22], among others). We prove that any streaming algorithm that computes a sequence of $(O(\phi \log n), \phi)$-expander decompositions requires \smash{${\widetilde{\Omega}}(n/\phi)$} bits of space, even in insertion only streams.
\end{abstract}

\pagestyle{empty}
\newpage
\tableofcontents
\setcounter{page}{0}
\newpage

\pagestyle{plain}

\section{Introduction}
Expander graphs are known to represent a class of easier instances for many problems. Therefore, breaking down the input into disjoint expanders can allow to conveniently solve the task on each of them separately, before combining the partial results into a global solution. This approach is enabled by $(\epsilon, \phi)$-expander decompositions (for short, $(\epsilon, \phi)$-ED). For an undirected graph $G=(V,E)$, this is a partition $\mathcal{U}$ of the vertex set $V$, such that there are at most $\epsilon |E|$ inter-cluster edges, while every cluster $U \in \mathcal{U}$ induces a $\phi$-expander. The list of successful applications of this framework is long, including Laplacian system solvers~\cite{stlapl}, deterministic algorithms for minimum cut~\cite{mincutjli}, graph and hyper-graph sparsification~\cite{stspectral, chugao, hgsparse,svapprox}, dynamic algorithms for cut and connectivity problems~\cite{boundlink,sublinearconn}, fast max flow algorithms~\cite{maxflow}, distributed triangle enumeration~\cite{triang}, polynomial time algorithms for semirandom planted CSPs~\cite{csps}, and many more.

We refer to $\epsilon$ and $\phi$ as the \emph{sparsity} parameters: the former controls how sparsely connected the clusters need to be, and the latter determines how expanding (i.e. non-sparse) the cuts within clusters are. The reason for using the same term for both is that they are in fact very closely related. One can show that any $n$-vertex graph has an $(\epsilon,\phi)$-ED with $\epsilon = O(\phi \log n)$. To see this, consider the following constructive argument: if the graph has no $\phi$-sparse cut then $\{V\}$ is a valid ED of~$G$, otherwise recurse on the two sides $(S,V \setminus S)$ of a $\phi$-sparse cut and union the results to get an ED for $G$. Every cluster in this decomposition is then an expander, and a charging argument allows to bound the number of inter-cluster edges by $O(\phi |E| \log n)$~\cite{offlineexpdec}. One can also observe that no better asymptotic trade-off between $\epsilon$ and $\phi$ is possible in general~\cite{clusteringhypercube}. Therefore, this sets the benchmark for algorithmic constructions of EDs.

At a high level, many ED algorithms follow the approach suggested by the existential argument. As a naive implementation would take exponential time, the crux often lies in efficient algorithms that either certify that a large portion of the input is an expander, or find a balanced sparse cut. This would result in a small depth recursion, thanks to balancedness, where each level requires little computational resources. There are several successful examples of this approach. In the sequential setting, a recent algorithm constructs a $(O(\phi \log^3 n), \phi)$-ED in $\otil(|E|)$ time~\cite{liexpdec}, based on the previous best algorithm which runs in $\otil(|E|/\phi)$ time~\cite{offlineexpdec}. There is also a deterministic counterpart, which outputs a $(\phi \cdot n^{o(1)}, \phi)$-ED in almost linear time~\cite{determexpdec}. For the CONGEST model of distributed computing, it is possible to obtain, for instance, a $(\phi^{1/\sqrt{\log n}} \cdot n^{o(1)}, \phi)$-ED in $n^{o(1)} / \phi$ rounds~\cite{triang}. It is also possible to maintain a $(\phi \cdot n^{o(1)},\phi)$-ED for a graph undergoing edge updates in $n^{o(1)}/\phi^2$ amortized update time~\cite{boundlink,dynexpdec}.

\subsection{Previous work}
In the streaming setting, the problem of finding EDs was open until the recent work of~\cite{streamexpdec}. They obtain a dynamic stream algorithm which outputs a $(O(\phi \log n),\phi)$-ED and takes $\otil(n/\phi^2)$ space. While being optimal in the quality of the decomposition, decoding the sketch to actually output an ED takes exponential time. The authors also give a polynomial time version: one can produce a $(\phi \cdot n^{o(1)}, \phi)$-ED using space $\otil(n/\phi^2)+n^{1+o(1)}/\phi^{1-o(1)}$, with a post-processing that takes $\poly(n)$ time (where the $o(1)$'s can be tuned, allowing for a quality-space trade-off).

These streaming algorithms also adopt the recursive approach based on finding balanced $\phi$-sparse cuts, but the streaming model poses a challenge that we now illustrate. Sparsification for graph streams has been extensively studied~\cite{ahnguha,linmeas,sparsdynamicstreams,focsspars,spectralspars}, so a natural attempt would consist of maintaining a cut sparsifier as the stream comes and later run the recursive partitioning on it. However, it must be noted that these cuts are to be found in the subgraphs induced by the two sides of a previously made cut. This is not a problem in a classical computational setting, but it actually constitutes the main obstacle for the streaming model: at sketching time, the algorithm does not know which subgraphs it will need to access. Unfortunately, it is impossible to preserve cut sizes in arbitrary subgraphs with multiplicative precision. The natural work-around is to introduce an additive error term: in~\cite{streamexpdec}, the authors introduce the concept of {power-cut} sparsifiers. For any given partition $\mathcal{U}$ of $V$, these sparsifiers preserve the cuts $|E(S,U\setminus S)|$ of each cluster $U \in \mathcal{U}$ in the partition up to an error of $\delta \cdot |E(S,U\setminus S)| + \psi \cdot \vol(S)$. They further show that maintaining $\otil(n / \delta \psi)$ random linear measurements of the incidence matrix of the input graph is enough to obtain such sparsifiers. One can see that setting $\delta \ll 1$ and~$\psi \approx \phi$ preserves the sparsity of cuts to within an additive error of roughly $\phi$, using $\otil(n/\phi)$ linear measurements. This is enough to find a balanced $\phi$-sparse cut in every subgraph induced by a given partition of the vertex set. There is another caveat, though: power-cut sparsifiers give a high probability guarantee for any fixed partition, but we cannot expect them to work for all partitions simultaneously. Therefore, in order not to use the sparsifiers adaptively, we need one power-cut sparsifier for every recursion level. The authors show that the depth of the procedure cannot exceed $\otil(1/\epsilon) = \otil(1/\phi)$, thus obtaining the space complexity stated before. A few more details are involved in the polynomial time algorithm, but the underlying framework is the same.

\subsection{Our contribution}
The work of~\cite{streamexpdec} initiated the study of expander decompositions in the streaming setting, and consequently raised the question of whether a dependence in space on the sparsity parameter~$\phi$ is inherent. In this paper, we move towards settling the streaming complexity of expander decompositions by attacking the problem on two fronts: (1) we give a nearly optimal algorithm for ``one-level'' expander decomposition that avoids the sparsity dependence, and (2) we show that computing a ``repeated'' expander decomposition, commonly used in applications, cannot avoid such dependence.
\subsubsection*{Upper bound}
We give an $\otil(n)$ space algorithm for computing a $(O(\phi \log n), \phi)$-ED in dynamic streams. Specifically, we show that a ``universal'' sketch consisting of $\otil(n)$ random linear measurements of the incidence matrix can be decoded into a $(O(\phi \log n), \phi)$-ED for any $\phi$: the sketch is independent of the sparsity~$\phi$.

\begin{mainthsimple}[ED algorithm -- exponential time decoding]
	\label{cor:main}
	Let $G=(V,E)$ be a graph given in a dynamic stream. Then, there is an algorithm that maintains a linear sketch of $G$ in $\otil(n)$ space. For any $\phi \in (0,1)$, the algorithm decodes the sketch to compute a $(O(\phi \log n), \phi)$-ED of $G$ with high probability, in $\otil(n)$ space and $2^{O(n)}$ time.
\end{mainthsimple}

\noindent
We note that at least $\Omega(n\log n)$ space is needed for any small enough $\phi$: for example, if the input graph is a matching of size $n/10$, say, its ED gives a $1-O(\phi \log n)$ fraction of the matching edges. 

The decoding time of our sketch can be made polynomial, at the expense of some loss in the quality of the expander decomposition (similarly to~\cite{streamexpdec}), but keeping the space independent of the sparsity $\phi$.

\begin{mainthsimple}[ED algorithm -- polynomial time decoding]
	\label{cor:mainpoly}
	Let $G=(V,E)$ be a graph given in a dynamic stream. Then, there is an algorithm that maintains a linear sketch of $G$ in $n^{1+o(1)}$ space. For any $\phi \in (0,1)$, the algorithm decodes the sketch to compute a $(\phi \cdot n^{o(1)}, \phi)$-ED of $G$ with high probability, in $n^{1+o(1)}$ space and $\poly(n)$ time.
\end{mainthsimple}

\noindent
In this case we are off by subpolynomial factors in both quality and space complexity as compared to the optimal ones. The actual theorem that we prove allows one to trade-off the loss in quality and increase in space.

\subsubsection*{Lower bound}
Most algorithmic applications of EDs, including the ones mentioned above, do not use just one ED of the input graph. Rather, they use an ED sequence obtained by repeatedly computing an ED of the inter-cluster edges from the previous level. This can be done in two natural ways: by contracting the clusters of an ED (we call this variant CED, for “contraction”), or by removing the intra-cluster edges without changing the vertex set (we call this variant RED, for “removal”). These approaches lead to different results and are used for different applications (e.g.,~\cite{boundlink,mincutjli} contract the clusters, and~\cite{stspectral, maxflow} recurse on inter-cluster edges). In the sequential setting, both CEDs and REDs can be obtained straightforwardly given an ED algorithm. However, this is not so obvious in the streaming model.

One the one hand, one should be able get a sparsity-independent algorithm for computing CEDs in dynamic streams via \cref{cor:main} or \cref{cor:mainpoly}: observe that contracting the vertices of a sparsifier of the input graph $G$  gives a sparsifier of the graph obtained by contracting vertices in~$G$, so the idea would be to maintain an independent copy of our algorithm for each of the $O(\log n)$ levels and contracting vertices in the sparsifier based on the decomposition of the previous level. On the other hand, we show a space lower bound for computing REDs, even in insertion only streams, showing that a dependence on $1/\phi$ is necessary.

\begin{mainthsimple}[RED lower bound]
	\label{th:lbsimple}
	Let $\epsilon,\phi \in (0,1)$ such that $1/n \ll \phi \ll 1/\pylog n$ and $ \epsilon =\otil( \phi)$. Any streaming algorithm that with constant probability computes at least two levels of an $(\epsilon,\phi)$-RED requires \smash{$\widetilde{\Omega}(n/\phi)$} bits of space.
\end{mainthsimple}

\noindent
The result seems to challenge the intuition that EDs become weaker as $\phi$ and $\epsilon$ decrease (note that when $\phi$ is, say, $1/n^2$, an ED can simply consist of connected components). The questions of~(1)~whether this bound can be improved,~(2)~how it scales with the number of levels of RED we compute, and~(3)~whether there are algorithms matching such bounds, remain open.

\subsection{Basic notation}
\label{subsec:basic}

\paragraph{Graph streaming.} In this paper, we will be mostly working in the dynamic graph streaming model, where we know the vertex set $V=[n]$, and we receive a stream of insertions and deletions for undirected edges over $V$. In insertion-only graph streams, the only difference is that previously inserted edges cannot be deleted. At the end of the stream, the graph $G=(V,E)$ consists of the edges that have been inserted and not deleted, and we say that $G$ is given in a (dynamic) stream. When we consider a graph given in a (dynamic) stream, we are implicitly assuming it to have $n$ vertices, without introducing the parameter $n$ explicitly. Also, when we refer to $G$ and $n$ without reintroducing them we are implicitly considering the graph resulting from the input stream and its number of vertices.

A powerful tool for dynamic graph streams is linear sketching, introduced in the seminal work of Ahn, Guha, and McGregor~\cite{linmeas}. The idea is to left-multiply the $\binom{n}{2} \times n$ incidence matrix of the graph by a random $k \times \binom{n}{2}$ matrix for $k \ll \binom{n}{2}$. Since the sketch consists of linear measurements, it automatically handles the case of dynamic streams. We will not be using sketching techniques directly, but rather employ existing algorithms that do. In this paper we restrain ourselves to unweighted edge streams. One may study the same problem in general turnstile graph streams~\cite{weightedlinearspars}, but we do not do this here.

\paragraph{Cuts, volumes, expanders.} Given an unweighted graph $G=(V, E)$, possibly with multiple self-loops on the vertices, we will operate with weighted graphs that approximate $G$ in an appropriate sense. We write $G'=(V',E', w)$ to denote a weighted graph, possibly with multiple weighted self-loops on the vertices. We describe next some notation for such weighted graph $G'$. The same notation carries over to the unweighted graph $G$ by implicitly setting $w$ to assign a weight of $1$ to all edges and self-loops.

For any $A,B \subseteq V'$ we denote by $E'(A,B)$ the edges in $E'$ with one endpoint in $A$ and one in $B$, and by $w(A,B)$ the total weight of edges in $E'(A,B)$. The volume of a cut $S \subseteq V'$ is the sum of the (weighted) degrees, including self-loops, of its vertices. We denote it by $\vol_{G'}(S)$. The sparsity of a cut $\emptyset \neq S \subsetneq V'$ is defined as
\begin{equation*}
	\Phi_{G'}(S) = \frac{w(S,V' \setminus S)}{\min\{\vol_{G'}(S), \vol_{G'}(V' \setminus S)\}} \, .
\end{equation*}
For $\psi \in (0,1)$, we make a distinction between cuts having sparsity less than $\psi$, which we call $\psi$-sparse, and cuts having sparsity at least $\psi$, which we call $\psi$-expanding. The sparsity of $G'$ is the defined as
\begin{equation*}
	\Phi_{G'} = \min_{\emptyset \neq S \subsetneq V'} \Phi_{G'}(S) \, ,
\end{equation*}
and we call $G'$ a $\psi$-expander if all its cuts are $\psi$-expanding, i.e. $\Phi_{G'} \ge \psi$.

For a pair of vertices $e=\{u,v\} \in \binom{V'}{2}$, we denote by $\lambda_{e}(G')$ the edge connectivity of $e$ in $G'$, i.e.
\begin{equation*}
	\lambda_{e}(G') = \min_{\emptyset \neq S \subsetneq V'\, : u \in S, v \in V' \setminus S} w(S,V'\setminus S) \, .
\end{equation*}

\paragraph{Expander decomposition.} As we will treat expander decompositions for the input graph $G$ only, we conveniently use the following additional notation. For a cluster $U \subseteq V$, i.e. a subset of the vertices, and a cut $S \subseteq U$, which is also a subset of the vertices, we denote the number of edges crossing $S$ in $U$ by $\uncut{S}{U}$, i.e. $\uncut{S}{U}=|E(S,U \setminus S)|$. This is the \emph{local} cut of $S$ in $U$. If $U = V$, we use a shorthand notation $\uncutg{S} \coloneqq \uncut{S}{V}$. We call such a cut the \emph{global} cut of $S$, since for $S \subseteq U$ we will be interested in both $\uncut{S}{U}$ and $\uncutg{S}$. Moreover, we drop the subscript from the volume and simply write $\vol(\cdot)$ instead of $\vol_G(\cdot)$. Then, EDs can be defined as follow.

\begin{definition}[Expander decomposition]
	\label{def:expdecclassic}
	Let $G=(V,E)$ and let $\epsilon, \phi \in (0,1)$. A partition $\mathcal{U}$ of~$V$ is an $(\epsilon, \phi)$-expander decomposition (for short, $(\epsilon, \phi)$-ED) of $G$ if
	\begin{enumerate}
		\item \label{property:crossingclassic} $\frac{1}{2}\sum_{U \in \mathcal{U}} \uncutg{U} \le \epsilon |E|$, and
		\item \label{property:expanderclassic} for every $U \in \mathcal{U}$, one has that $G[U]$ is a $\phi$-expander.
	\end{enumerate}
\end{definition}

\section{Technical overview}
In this section, we outline our techniques, and give intuition for how to prove our results.
\subsection{Sparsity-independent one-level expander decomposition}
\label{sec:techoverviewub}
Given a graph $G=(V,E)$ in a  dynamic stream, and a parameter $\phi \in (0,1)$, we consider the problem of computing an \smash{$(\epsilon, \phi)$-ED} of $G$ for \smash{$\epsilon = O(\phi   \log n)$}. We show that one can do this in \smash{$\otil(n)$} bits of space, without any dependence on $\phi$.
In this section, we sketch our approach.

Let us begin by recalling the standard recursive framework for constructing expander decompositions~\cite{kvv,Trevisan08,stspectral}, concisely summarized in \cref{alg:metaed}. For any parameter $\phi \le 10^{-1}/\log n$, this algorithm produces a $(O(\phi \log n), \phi)$-ED by recursively partitioning $G$ along $\phi$-sparse cuts until no more such cuts are found.

\begin{algorithm}[H]
	\caption{\ed: recursive procedure for computing a $(O(\phi \log n), \phi)$-ED of $G$}
	\label{alg:metaed}
	
	\begin{algorithmic}[1]
		\LeftComment{$\phi \in (0,1)$ is the sparsity parameter}
		\Procedure{$\ed(G)$}{} \Comment{$G=(V,E)$ is the input graph}
		
		\If{$G$ is a $\phi$-expander} \Return $\{V\}$
		\Else
		\State $S \gets $ a $\phi$-sparse cut of $G$ \label{eq:sparsecut}
		\State \Return $\ed(G[S]) \cup \ed(G[V \setminus S])$
		\EndIf
		
		\EndProcedure
	\end{algorithmic}
\end{algorithm}
\noindent
Many algorithmic constructions of EDs are essentially efficient implementations of \cref{alg:metaed}. Adapting \cref{alg:metaed} to dynamic streams comes with its own set of challenges.
When the input is given in a dynamic stream, one can only afford to store a limited amount of information about the input graph. Since \cref{alg:metaed} only needs to measure the sparsity of cuts, it seems enough to have access to cut sizes and volumes. Both these quantities are preserved by cut sparsifiers. According to the classical definition~\cite{benczurkarger}, a $\delta$-cut sparsifier is a weighted subgraph $H=(V,E',w)$ of the input $G=(V,E)$, where for every cut $S \subseteq V$ one has
\begin{equation*}
	(1-\delta)\cdot \uncutg{S} \le \wcutg{S}\le (1+\delta) \cdot \uncutg{S} \, .
\end{equation*}
It is known that such a sparsifier can be constructed in dynamic streams using $\otil(n/\delta^2)$ bits of space~\cite{sparsdynamicstreams}. With the same space requirement we can also measure the sparsity of cuts up to a $(1\pm \delta)$ multiplicative error. Having this in mind, it is natural to consider the following algorithm: first, read the stream and construct a cut sparsifier, then run \cref{alg:metaed} on it and output the resulting clustering as the expander decomposition. However, this approach does not work as is. As it turns out, more information is needed about the graph than what is captured by the regular notion of a cut sparsifier.

The problem with this approach becomes immediately apparent when one considers the second recursion level. As the process recurses on the two sides of a sparse cut $S$, it will repeat the procedure on the subgraphs $G[S]$ and $G[V \setminus S]$. Unfortunately, the cut preservation property of the sparsifier does not carry over to those subgraphs, making it inadequate for estimating the sparsity of the cuts in those graphs. 
In fact, the notion of expander decomposition itself already operates on the subgraphs, as it is required that each cut in each cluster of the decomposition is expanding.

\subsubsection*{Testing expansion of subgraphs}
The reasoning from above gives rise to the following sketching problem: produce a sparsifier that can be used to check that any subgraph is a $\phi$-expander. To solve it, the authors of~\cite{streamexpdec} introduced the concept of a $(\delta,\psi)$-power-cut sparsifier. Its property is that, for any cluster $U \subseteq V$, with high probability all cuts $S \subseteq U$ are preserved within an additive-multiplicative error.
Recall that $\uncut{S}{U}$ is equal to the size of the cut $S$ inside the induced subgraph $G[U]$.
When talking about sparsifiers in this section, we will slightly abuse notation by assuming that there is some instance $H=(V,E',w)$ of it and denote by $\wcut{S}{U}$ the size of the cut $S$ in the subgraph $H[U]$ of this sparsifier~$H$. Then we can write the guarantee of a $(\delta, \psi)$-power-cut sparsifier as follows\footnote{A similar sparsifier construction was proposed by~\cite{sublinhierarchical}. Their construction has an additive error of $\psi |S|$, so the dependence is on the number of vertices instead of the volume}:
\begin{equation*}
	\forall \, S \subseteq U, \quad  (1-\delta)\cdot \uncut{S}{U}-\psi \cdot \vol(S) \le \wcut{S}{U} \le (1+\delta)\cdot \uncut{S}{U}+\psi \cdot \vol(S) \, .
\end{equation*}
The authors also give a dynamic stream construction, which uses $\otil(n/\delta\psi)$ bits of space by sampling edges proportionally to the degrees of their endpoints.

To check the $\phi$-expansion of subgraphs up to a small constant multiplicative error, it is enough to set $\delta$ to be a small constant. However, the multiplicative error parameter $\psi$ must be $\lesssim \phi$. 
This is a significant downside of this construction, as it was shown in \cite{streamexpdec} that a $(\delta, \phi)$-power-cut sparsifier must have at least $\Omega(n /\phi)$ edges. In fact, their lower bound is more general, and holds for general subgraphs\footnote{The lower bound instance is a $1/\phi$-regular graph. There, each edge by itself forms a $\phi$-expander, while any pair of vertices without an edge between them is not. Being able to test the $\phi$-expansion of the majority of those small subgraphs would imply being able to recover the majority of edges in the graph.}.
In other words, {$\Omega(n/\phi)$} space is necessary in order to test {$\phi$}-expansion of general subgraphs. Consequently, new tools must be used to have any hope of getting an algorithm independent~of~$1 / \phi$.

\paragraph*{Our contribution: sparsification of boundary-linked subgraphs}

As it turns out, solving the expansion testing problem, as it was stated, is not necessary. Recently, in the breakthrough work of \cite{boundlink}, it was shown that one could demand additional properties from the expander decomposition, and it will still exist at the price of increasing the number of inter-cluster edges by a small multiplicative factor.

More formally, for a set of vertices $U \subseteq V$ and $\tau > 0$, we denote by $G[U]^\tau$ an induced subgraph of $G$ on $U$, where additionally for each edge $e=\{u,v\}$ that connects a vertex $u$ inside $U$ with a vertex $v$ outside of $U$, $\tau$ self-loops are added to $u$. The graph $G[U]^\tau$ is called a $\tau$-boundary-linked subgraph. Note that if $G[U]^\tau$ is a $\phi$-expander, then so necessarily is $G[U]$, but not the other way around.

Then it is possible, for a given $\phi$, and $\eps = \otil(\phi), \, \tau \approx 1/\phi$, to construct an $(\eps, \phi)$-expander decomposition $\mathcal{U}$ where for each cluster $U \in \mathcal{U}$, $G[U]^\tau$ is a $ \widetilde{\Omega}(\phi)$-expander~\cite{boundlink}. Such a decomposition is called a boundary-linked expander decomposition.

A key observation is that if in our algorithm we were aiming to construct a boundary-linked expander decomposition, the testing problem would only involve checking that any given boundary-linked subgraph is $\phi$-expanding. Indeed, this problem is much easier than the original one and can be solved in space \smash{$\otil(n)$}. We will show how to do it in two steps: first, we will discuss how to strengthen the power-cut sparsifier, and then prove that this strengthening is enough to resolve the problem.

\paragraph*{Achieving additive error in the global cut}
As was noted in \cite{streamexpdec}, the idea behind constructing a power-cut sparsifier was to reanalyse the guarantee given by the construction of~\cite{stspectral} for sparsifying expanders. In other words, a power-cut sparsifier results from a more rigorous analysis of an existing sparsifier. The problem with this approach is that the sparsifier of~\cite{stspectral} is relatively weak to begin with: it only preserves cuts in expanders, while other constructions can preserve them in all graphs~\cite{benczurkarger,fungspars}.
To strengthen the guarantee, we give the same treatment to the construction of \cite{fungspars} in \cref{sec:spars}. For the sake of simplicity and to gain an intuition for why this kind of sparsification is at all possible, we discuss here how to do that with the classical construction of \cite{karger1994random} by closely following the original proof.

We show that, given a graph $G=(V, E)$ with a minimum cut of size $k$, it is possible to construct a sparsifier $H$ of $G$ such that every cut inside any given subgraph $G[U]$ of $G$ is preserved with high probability with the following guarantee:
\begin{equation}
	\label{eq:to-cut-guar}
	\forall \, S \subseteq U, \quad \uncut{S}{U} - \delta \cdot \uncutg{S} \le \wcut{S}{U}  \le \uncut{S}{U} + \delta \cdot \uncutg{S}\,.
\end{equation}
In this paper, a sparsifier with property~\eqref{eq:to-cut-guar} is called a cluster sparsifier (see \cref{def:spars}).
To achieve~\eqref{eq:to-cut-guar}, consider using the same process as in \cite{karger1994random}: sample each edge with the same probability $p \approx {\delta^{-2}/ k}$.

To see why this works, fix a subgraph $U \subseteq V$, and consider any cut $S$ in $U$. 
We wish to show that the size of the cut $S$ inside $U$ concentrates well after sampling. In the original proof, this is done by simply applying a Chernoff bound. In our case, this bound would look like this:
\begin{equation*}
	\Pr[|\wcut{S}{U} - \uncut{S}{U}| \geq \delta \cdot \uncut{S}{U}] \leq \exp\left(-\frac{1}{3} \delta^2 p \cdot \uncut{S}{U} \right) \, .
\end{equation*}
\noindent
However, this is insufficient, as the probability depends on the cut size inside $G[U]$. As we have no lower bound on its size, unlike with the sizes of global cuts, the second part of the argument of \cite{karger1994random} cannot be applied. Instead, we apply an additive-multiplicative version of the Chernoff bound (see for example \cite{streamexpdec}), that allows us to compare the approximation error with a bigger value than its expectation. This gives us
\begin{equation*}
	\Pr[|\wcut{S}{U} - \uncut{S}{U}| \geq \delta \cdot \uncutg{S}] \leq 2\exp\left(-\frac{1}{100} \delta^2 p \cdot \uncutg{S} \right) \, .
\end{equation*}
Expressing the probability in terms of the global cut allows us to use the cut counting lemma~\cite{karger1993global}, which bounds the number of global cuts of size at most $\alpha k$ by $n^{2\alpha}$,  for $\alpha \geq 1$. The proof is concluded by associating each global cut with a local cut in $G[U]$ and taking the union bound over them.

In order to get the guarantee of~\eqref{eq:to-cut-guar} for all graphs, not only those with a minimum cut of size $k$, one can sample edges proportionally to the inverse of their connectivity (as in the work of \cite{fungspars}) as opposed to uniformly. Moreover, one can implement such sampling scheme in dynamic streams, using the approach of Ahn, Guha, and McGregor~\cite{sparsdynamicstreams}. We defer the details of these reanalyses to \cref{sec:spars}, where we show how to construct cluster sparsifiers with property~\eqref{eq:to-cut-guar} in dynamic streams.

\paragraph*{Benefits of boundary-linked graphs}
To see why the cluster sparsifier is enough to solve the boundary-linked $\phi$-expansion testing problem, consider the following reasoning. Set the self-loop parameter $\tau$ equal to $b / \phi$, for some $b \gg \phi$ and $b \ll 1$. Fix a cluster $U$, for which $G[U]^{b / \phi}$ is an $\widetilde{\Omega}(\phi)$-expander. 
The crucial fact is that the size of any cut inside $G[U]^{b / \phi}$ is lower bounded by its size in the global graph up to a small polylogarithmic factor in the following way:
\begin{equation}
	\label{eq:to_cut_lb}
	\uncut{S}{U} \geq \widetilde{\Omega}(b) \cdot \uncutg{S} \, .
\end{equation}
We will now explain the derivation of the above equation in detail.
For a cut $\emptyset \neq S \subsetneq U$, we denote by $\unbor{S}{U} = \uncutg{S} - \uncut{S}{U}$ the number of edges going from $S$ to $V \setminus U$. Note that by the definition of~$G[U]^{b / \phi}$, the volume of any cut $S$ inside of it is equal to \smash{$\vol_{G[U]^{b / \phi}}(S) = \vol(S) + (\frac{b}{\phi} - 1)\unbor{S}{U}$}.
First, because we assume that $G[U]^{b / \phi}$ is an $\widetilde{\Omega}(\phi)$-expander, we have
\begin{equation*}
	\uncut{S}{U} \geq \widetilde{\Omega}(\phi)\cdot \vol_{G[U]^{b / \phi}}(S)\,.
\end{equation*}
Then, applying the aforementioned formula for $\vol_{G[U]^{b / \phi}}(S)$, we have
\begin{equation*}
	\uncut{S}{U} \geq \widetilde{\Omega}(\phi)\vol(S) + \widetilde{\Omega}(\phi)\left(\frac{b}{\phi} -1 \right)\unbor{S}{U} \,.
\end{equation*}
Since $b \gg \phi$ and dropping the first summand, the above simplifies to 
\begin{equation*}
	\uncut{S}{U} \geq  \widetilde{\Omega}(b)\unbor{S}{U} \,.
\end{equation*}
Finally, applying $\uncut{S}{U} + \unbor{S}{U} = \uncutg{S}$ and $b \ll 1$, we arrive back at equation~\eqref{eq:to_cut_lb}.

A consequence of equation~\eqref{eq:to_cut_lb} is that setting $\delta \approx 1/ b$ in the cluster sparsifier produces the following guarantee for $ \widetilde{\Omega}(\phi)$-expander subgraphs $G[U]^{b / \phi}$: 
\begin{equation*}
	\forall \, S \subseteq U, \quad \uncut{S}{U} - O(1) \cdot \uncut{S}{U} \le \wcut{S}{U}  \le \uncut{S}{U} + O(1) \cdot \uncut{S}{U}\,,
\end{equation*}
where the $O(1)$ can be made arbitrarily small. In other words, we achieve a multiplicative approximation guarantee on subgraphs of interest, which is enough to solve the testing problem. Since $b$ can be set to be $1/\pylog n$, the final sparsifier size does not depend on $1/\phi$. In this sense, such sparsifier can be thought of as a ``universal'' sketch of the graph that allows to solve the testing problem for all $\phi$ simultaneously.

Even though we now know how to solve the testing problem, an implementation of \cref{alg:metaed} would need to actually find a sparse cut when the test fails (and in particular, it needs to find a \textit{balanced} sparse cut, see the next section). In \cref{sec:balsparse}, we show that cluster sparsifiers enable to solve this harder task too: consider an offline algorithm that either correctly determines a graph to be a boundary-linked expander or finds a (balanced) sparse cut; we show that one can run such algorithm on a subgraph of the sparsifier as a black box and obtain essentially the same result as if the algorithm was run on the corresponding subgraph of the original graph. In this sense, cluster sparsifiers serve as small error proxies to the original graph for expander-vs-sparse-cut type of queries on vertex-induced subgraphs.

\subsubsection*{Low depth recursion}
Another problem that arises when using the recursive approach is that the subgraph sparsification guarantee of both power-cut sparsifiers and cluster sparsifiers is only probabilistic, and it can be shown that it cannot be made deterministic \cite{streamexpdec}. This means that after finding a sparse cut~$S$ in a subgraph $G[U]^{b / \phi}$ of a sparsifier, we cannot claim that the same sparsifier would preserve the cuts inside the new graph $G[S]^{b / \phi}$ with high probability, as it is dependent on another cut that we have already found inside the sparsifier. This means we cannot use the same sparsifier for two different calls to \cref{alg:metaed} inside the same execution path.
However, we can share a sparsifier among all the calls at the same recursion level since these will operate on independent portions of~$G$. Therefore, we want to have a separate sparsifier for every recursion level. This means that in order to minimize space requirements, it is crucial to have a small recursion depth.

To have small recursion depth, the algorithmic approach of \cite{streamexpdec,triang,ChangPSZ21} enforces the sparse cut $S$ from line~\eqref{eq:sparsecut} to be balanced, i.e., none of $S$ and $V\setminus S$ is much larger than the other. In particular, they only recurse on the two sides of a cut if $\vol(S) \gtrsim \epsilon \vol(V \setminus S)$. 
When there is no balanced sparse cut and yet the input is not an expander, a lemma of Spielman and Teng~\cite{stspectral} suggests there should be an $\Omega(\phi)$-expander $G[S']$ that accounts for a $(1-O(\epsilon))$-fraction of the total volume. An algorithmic version of this structural result allows us to iteratively trim off a small piece of the graph until such $S'$ is found. At this point, the algorithm of~\cite{streamexpdec} can simply return $\{S'\} \cup (\cup_{u \in V \setminus S'} \{u\})$ as an ED, with at most $O(\epsilon|E|)$ inter-cluster edges between singletons. As the volume of the cluster multiplicatively decreases by $1 - O(\eps)$ after each call, this gives recursion depth at most $\otil(1/\epsilon) \approx \otil(1/\phi)$.

\paragraph{Our contribution: adaptation of trimming.} We show that a simple refinement of this approach allows us to adapt the framework of \cite{offlineexpdec}, which leads to an algorithm with recursion depth independent of  $\phi$. We run the same algorithm, but instead of separating each vertex in $V \setminus S'$ into its own singleton cluster, we recurse with \cref{alg:metaed} on the whole set $V \setminus S'$. Conceptually, this implements an analogue of the trimming step of \cite{offlineexpdec}.

This means that at the end, fewer edges in $V \setminus S$ become inter-cluster edges. Because of that,
we can strengthen the balancedness requirement: we recurse on the two sides of a sparse cut only if \smash{$\vol(S) \gtrsim \frac{1}{C} \vol(V \setminus S)$} for a large constant $C$. This allows us to trim more vertices each time, resulting in $\otil(1)$ depth of the trimming step.  On the other hand, because in each call to \cref{alg:metaed} the volume of clusters passed to recursive calls is decreased by at least a constant factor, the total recursion depth becomes at most $\otil(1)$.

Other than the refinement discussed above, our space efficient implementation of \cref{alg:metaed}, as well as  the iterative procedure to find the large expander $S'$, are almost the same as the ones of~\cite{streamexpdec} (which in turn are  inspired by the one of~\cite{triang}). The details of the algorithms are given in \cref{sec:algexpdec}.

\subsubsection*{Putting it all together}
Combining the ideas illustrated in the two sections above, neither the sparsifier's size nor the recursion depth depend on $\phi$, thus giving a sparsity-independent space algorithm for boundary-linked expander decomposition  (BLD for short). This result is stated in terms of parameters $b,\epsilon,\phi \in (0,1), \, \gamma \ge 1$: a $(b, \epsilon, \phi, \gamma)$-BLD is a partition $\mathcal{U}$ with  at most an $\epsilon$ fraction of crossing edges and every $U \in \mathcal{U}$ induces a $\phi/\gamma$-expander $G[U]^{b/\phi}$ (see \cref{sec:prelims} and \cref{def:expdec}). Using this terminology, we obtain the following result.

\begin{restatable}[Exponential time decoding BLD]{theorem}{mainbld}
	\label{th:main}
	Let $G=(V, E)$ be a graph given in a dynamic stream, and let $b \in (0,1)$ be a parameter such that $b \le 1/\log^2 n$. Then, there is an algorithm that maintains a linear sketch of $G$ in $\otil(n/b^3)$ space. For any $\epsilon \in [n^{-2}, b \log n]$, the algorithm decodes the sketch to compute, with high probability and in $\otil(n/b^3)$ space and $2^{O(n)}$ time, a $(b, \epsilon, \phi, \gamma)$-BLD of $G$ for
	\begin{equation*}
		\phi = \Omega\left(\frac{\epsilon}{\log n}\right) \quad \text{ and } \quad \gamma = O(1) \, .
	\end{equation*}
\end{restatable}
\noindent
From this, one can easily conclude our main result, restated here for convenience of the reader.

\begin{mainthactual}[ED algorithm -- exponential time decoding]
	\label{cor:maingeneral}
	Let $G=(V, E)$ be a graph given in a dynamic stream. Then, there is an algorithm that maintains a linear sketch of $G$ in $\otil(n)$ space. For any  \smash{$\phi \in (0, 1)$} such that $ \phi \le c/\log^2 n$ for a small enough constant $c>0$, the algorithm decodes the sketch to compute a $(O(\phi \log n), \phi)$-ED of $G$ with high probability, in $\otil(n)$ space and $2^{O(n)}$ time.
\end{mainthactual}
\noindent
We remark that for \smash{$\phi \ge \Omega(1/\log^2n)$}, one can use the algorithm of~\cite{streamexpdec} to still have an $\otil(n)$ space construction of a $(O(\phi \log n), \phi)$-ED.
\begin{proof}
	First note that without loss of generality we can assume $\phi \ge 1/n^2$, otherwise an ED can simply consists of the connected components of $G$ (which can be computed in dynamic streams in~$\otil(n)$ space~\cite{linmeas}). Then, we note that since $c$ is small enough and $\phi \le c / \log^2 n$, one can always define $\epsilon = C \cdot \phi \cdot \log n$ for an appropriate constant $C>0$ while ensuring $1/n^2 \le \epsilon \le 1/\log n$. We can thus prove the theorem by equivalently showing that there is an algorithm that maintains a linear sketch of $G$ in $\otil(n)$ space, and that for all \smash{$\epsilon \in [{1}/{n^2},{1}/{\log n}]$} decodes the sketch to compute, with high probability, an $(\epsilon, \Omega(\epsilon / \log n)))$-ED of $G$ in $\otil(n)$ space and $2^{O(n)}$ time.
	
	Let then \smash{$\epsilon \in [{1}/{n^2},{1}/{\log n}]$}. We use the algorithm from \cref{th:main} with parameter $\epsilon$ and a parameter $b$ of our choice. We need to meet two preconditions: $b \le 1/\log^2 n$ and $\epsilon \le b \log n$. Since we assume $\epsilon \le 1/\log n$, we can set $b=1/\log^2 n$, and all the prerequisites are fulfilled. Then the algorithm from \cref{th:main} runs in $2^{O(n)}$ time and takes $\otil(n/b^3)=\otil(n)$ bits of space. The output $\mathcal{U}$ is a $(b, \epsilon, \phi, \gamma)$-BLD of $G$ with high probability, where $\phi = \Omega(\epsilon/\log n)$ and $\gamma =O(1)$. Since a $(b, \epsilon, \phi, \gamma)$-BLD of $G$ is an $(\epsilon,\phi/\gamma)$-ED of $G$, we have obtained an $(\epsilon, \Omega(\epsilon/\log n))$-ED of $G$ with high probability.
\end{proof}
\noindent
The exponential time in the decoding is due to the subtask of finding a balanced sparse cut. As we show, one can make the decoding time polynomial by resorting to known offline approximation algorithms~\cite{offlineexpdec,liexpdec}. However, we only have $\log^{\Omega(1)} n$-approximations for finding a balanced sparse cut, and in particular, we do not expect (under NP-hardness and the Unique Games Conjecture) there to be a polynomial time $O(1)$-approximation~\cite{hardsparsecut}. Such super-constant factor error incurs some loss in the quality of decomposition and space requirement, which, nevertheless, remains independent of the sparsity.

\begin{restatable}[Polynomial time decoding BLD]{theorem}{mainbldpoly}
	\label{th:mainpoly}
	Let $G=(V, E)$ be a graph given in a dynamic stream, and let $b \in (0,1)$ be a parameter such that $b \le 1/\log^5 n$. Then, there is an algorithm that maintains a linear sketch of $G$ in \smash{$n/b^3 \cdot \log^{O(\log n / \log\frac{1}{b})} n $} space. For any $\epsilon \in [n^{-2}, b \log n]$, the algorithm decodes the sketch to compute, with high probability and in \smash{$n/b^3 \cdot \log^{O(\log n / \log\frac{1}{b})} n $} space and $\poly(n)$ time, a $(b, \epsilon, \phi, \gamma)$-BLD of $G$ for
	\begin{equation*}
		\phi = \Omega\left(\frac{\epsilon}{\log^4 n}\right) \quad \text{ and } \quad \gamma = \log^{O\left(\frac{\log n}{\log 1/b}\right)} n\, .
	\end{equation*}
\end{restatable}
\noindent
\cref{cor:mainpolygeneral}, restated here for convenience, then follows from \cref{th:mainpoly}.

\begin{mainthactual}[ED algorithm -- polynomial time decoding]
	\label{cor:mainpolygeneral}
	Let $G=(V, E)$ be a graph given in a dynamic stream, and let $b \in (0,1)$ be a parameter such that $b \le 1/\log^5 n$. Then, there is an algorithm that maintains a linear sketch of $G$ in \smash{$n/b^3 \cdot \log^{O(\log n / \log\frac{1}{b})} n $} space. For any {$\phi \in (0,1)$} such that \smash{$\phi \le b/\log^{C \cdot \log n / \log \frac{1}{b}}n$} for a large enough constant $C>0$, the algorithm decodes the sketch to compute a \smash{$(\phi \cdot \log^{O(\log n / \log\frac{1}{b})} n, \phi)$}-ED of $G$  with high probability, in \smash{$n/b^3 \cdot \log^{O(\log n / \log\frac{1}{b})} n $} space and~\smash{$\poly(n)$} time.
\end{mainthactual}
\noindent
We remark that setting, say, $b = 2^{-\sqrt{\log n}}$  in the above result gives a $n^{1+o(1)}$ space algorithm for computing a $(\phi \cdot n^{o(1)}, \phi)$-ED for any \smash{$\phi \le 2^{-2C\log\log n \sqrt{\log n}}$}. For larger values of $\phi$, one can use the polynomial time algorithm of ~\cite{streamexpdec} to still get a $n^{1+o(1)}$ space construction for a  $(\phi \cdot n^{o(1)}, \phi)$-ED.
\begin{proof}
	As in the proof of \cref{cor:maingeneral}, we can assume $\phi \ge 1/n^2$. Also, by virtue of $C$ being a large enough constant and  \smash{$\phi \le b/\log^{C \cdot \log n / \log \frac{1}{b}}n$}, one can always define \smash{$\epsilon = \phi \cdot \log^{C \cdot \log n / \log \frac{1}{b}}n$} while ensuring $n^{-2} \le \epsilon \le b \log n$. Then, we equivalently prove that for any $b \in (0,1)$ with $b \le 1/\log^5 n$ there is an algorithm that maintains a linear sketch of~$G$ in \smash{$n/b^3 \cdot \log^{O(\log n / \log\frac{1}{b})} n $} space, and that for any $\epsilon \in (0,1)$ such that $n^{-2} \le \epsilon \le b \log n$ decodes the sketch to compute, with high probability, an \smash{$(\epsilon, \epsilon/\log^{O(\log n / \log\frac{1}{b})} n)$}-ED of~$G$ in \smash{$n/b^3 \cdot \log^{O(\log n / \log\frac{1}{b})} n $} space and \smash{$\poly(n)$} time.
	
	Let then $b,\epsilon \in (0,1)$ with $b \le 1/\log^5 n$ and $n^{-2} \le \epsilon \le b \log n$. We use the algorithm from \cref{th:mainpoly} with the same parameters $b$ and $\epsilon$, since every admissible pair of parameters $b$ and $\epsilon$ fulfils the conditions of \cref{th:mainpoly}. The space complexity is also the same, and the running time is $\poly(n)$. Again, observe that a $(b, \epsilon, \phi, \gamma)$-BLD of $G$ is an $(\epsilon,\phi/\gamma)$-ED of $G$. Hence the claim, since \cref{th:mainpoly} gives
	\begin{equation*}
		\frac{\phi}{\gamma} = \Omega\left(\frac{\epsilon}{\log^4 n}\right) \cdot \frac{1}{\log^{O\left(\frac{\log n}{\log 1/b}\right)} n} = \frac{\epsilon}{\log^{O\left(\frac{\log n}{\log 1/b}\right)} n} \, .
	\end{equation*}
\end{proof}

\noindent
The proofs of \cref{th:main} and \cref{th:mainpoly} can be found in \cref{sec:algexpdec}.

\subsection{Two-level expander decomposition incurs a sparsity~dependence}
\label{sec:techoverviewlb}
Given a graph $G=(V,E)$ in a stream and parameters $\epsilon,\phi \in (0,1)$, we consider the problem of computing a two-level $(\epsilon,\phi)$-RED of $G$. In other words, we study the problem of computing an $(\epsilon,\phi)$-ED $\mathcal{U}$ of $G$ and an $(\epsilon,\phi)$-ED $\mathcal{U'}$ of the graph $G'=(V,E\setminus \mathcal{U})$, where $E\setminus \mathcal{U}$ denotes the set of inter-cluster edges of $\mathcal{U}$, i.e. the edges of $E$ that are not entirely contained in a cluster $U \in \mathcal{U}$ (see \cref{sec:prelims} and \cref{def:expdecrec}). We remark that we wish to do so in a single pass over the stream.

\paragraph*{A natural algorithmic approach and why it fails.} A naive attempt to solve this problem would be that of sketching the graph twice, for example by using our algorithm from \cref{cor:main}. One can use the first sketch to construct the first level ED $\mathcal{U}$. After, the hope is that one can send updates to the second sketch so as to remove the intra-cluster edges and then decode this sketch into an ED of $G'=(V,E\setminus \mathcal{U})$ using again \cref{cor:main}. However, this hope is readily dashed. Indeed, sketching algorithms break down if we send a removal update for an edge that was not there in the first place, and we do not have knowledge of which of the pairs~$\binom{U}{2}$ are in $E$ and which are not.

\paragraph*{Our contribution: space lower bound.} We show that, in sharp contrast to our algorithm for constructing a one-level expander decomposition, this problem requires $\widetilde{\Omega}(n/\phi)$ space, i.e. a dependence on $1/\phi$ is unavoidable. Formally, we obtain the following result.

\begin{restatable}[RED lower bound]{mainthactual}{lb}
	\label{th:lb}
	Let $\ell \ge 2$ and let $\epsilon, \phi \in (0,1)$ such that $ \epsilon = 1-\Omega(1)$, $\phi \le \epsilon$, and $\phi \ge C \cdot\max\{\epsilon^2 , \, 1/n\}$ for a large enough constant $C > 0$. Any  streaming algorithm that with probability at least $9/10$ computes  an $\ell$-level $(\epsilon,\phi)$-RED requires {$\Omega(n/\epsilon)$} bits of space.
\end{restatable}
\noindent
The above theorem gives an $\widetilde{\Omega}(n/\phi)$ space lower bound for algorithms that compute a RED with near-optimal parameters, i.e. algorithms that achieve $\epsilon =\otil(\phi)$ for any $1/n \ll \phi \ll 1/\log n$.

\paragraph*{Setup and hard instances.}
Throughout this section, the symbols $\ll$ and $\gg$ mean smaller or larger by a large constant factor. Let us fix the RED parameters $\epsilon, \phi\in (0,1)$, and let us restrain ourselves to the regime $\phi \gg \eps^2$ (and of course $\phi \ll \epsilon$). We prove the lower bound by giving a distribution over hard instances $G=(V, E)$. This distribution is parametrised by integers $d$ and $m$ such that $1 \ll d \ll m \ll n$, $m \gg 1/\phi$, $ m \ll 1/\eps^2$, and $d \ll \frac{1}{\eps}$. With these parameters fixed, our hard distribution $\mathcal{G}$ is defined below in \cref{def:harddistrsimple}. An illustration is given in \Cref{fig:lb1}.

\begin{definition}[Distributions $\mathcal{G}$ and $\mathcal{G}'$ -- Informal, see \cref{def:harddistr}]
	\label{def:harddistrsimple}
 We partition $V$ arbitrarily  into two sets $S$ and $T$ with $n/2$ vertices each, and further partition $S$ into $n/m$ sets $S_1,\dots,S_{n/m}$ with $m/2$ vertices each. The edge set of the graph $G=(V,E) \sim \mathcal{G}$ is defined as follows.

\begin{enumerate}
    \item  For each $i~\in~[n/m]$, the induced subgraph $G[S_i]$ is an Erd\H{o}s-R\'{e}nyi random graph with $m/2$ vertices and degree $\approx d$. We denote by $\mathcal{G}'$ the distribution of the subgraph $G[S]$.

\item The induced subgraph $G[T]$ is a {\em fixed} $d$-regular $ \Omega(1)$-expander.

\item We fix for convenience an arbitrary labelling $s_{i,1},\dots,s_{i,m/2}$ of the vertices in each $S_i$, and we sample an index $K$ uniformly from $[m/2]$. Then, for every $i\in [n/m]$, we add $dm/2$ edges from $s_{i,K}$ to $T$ so that each $t \in T$ has $d$ incident edges connecting to $S$.
\end{enumerate}

\end{definition}
\noindent
Roughly speaking, our hard instances should be composed of $n/m$ regular expanders that are densely connected to $T$, which is also an expander, through a selection of ``special'' vertices. We will show that the hardness arises from recovering information about certain important vertices and edges, defined below and also illustrated in \Cref{fig:lb1}.

\begin{figure}[h]
	\begin{minipage}[center]{\textwidth}
		\centering
		\includegraphics[scale=0.93]{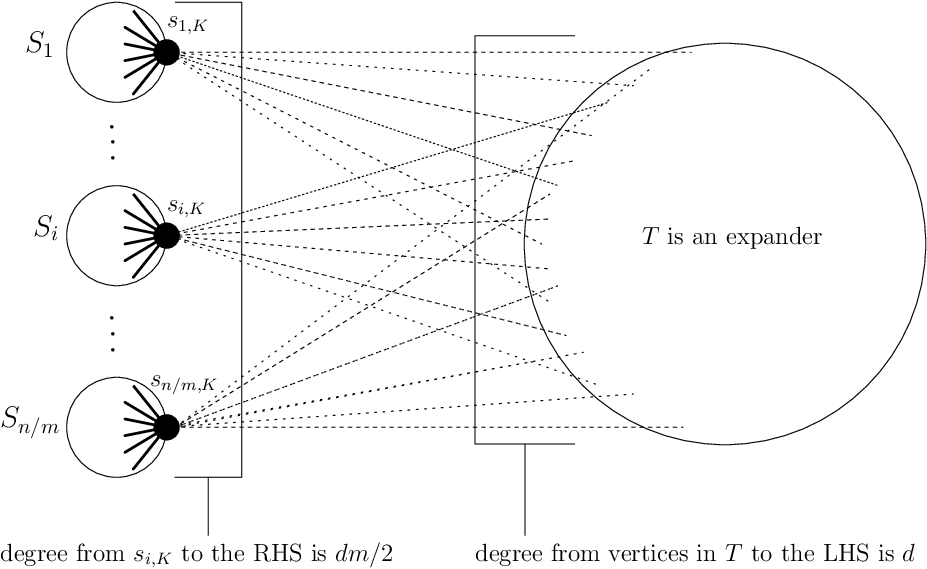}
	\end{minipage}
	\caption{Illustration of the graph we use for proving the lower bound. Thick bullets represent important vertices, thick lines represent important edges, dotted lines represent edges connecting the important vertices to $T$.}
	\label{fig:lb1}
\end{figure}

\begin{definition}[Important vertices and edges -- Informal, see \cref{def:impedges}]
	\label{def:impedgessimple}
Let $G=(V,E) \sim \mathcal{G}$.  We define the set of important vertices $V^*=\{s_{i,K}:i \in [n/m]\}$ to be the set of vertices of $S$ that are connected to $T$, and define the set of important edges $E^*=\{\{s_{i,K},v\} : i \in [n/m],\, v \in S \}$ to be the set of edges in the induced subgraph $G[S]$ that are incident on $V^*$.
\end{definition}
\noindent
The lower bound proof has two steps: we first prove that a two-level RED leaks a non-trivial amount of information about the graph $G \sim \mathcal{G}$; then we prove that, in order to obtain such amount of information, the algorithm must use a lot of space.

\subsubsection*{Two-level expander decomposition of the hard instance}
In this section, we show that any valid two-level RED reveals a lot of informations about the important edges. The following lemma shows that a non-trivial amount of important edges are inter-cluster edges in the first level decomposition. An ideal decomposition is  illustrated in \Cref{fig:decomp1}.

\begin{lemma}[Informal, see \cref{lem:special-edge}] \label{lem:special-edge-ov}
	Let $G=(V,E) \in \supp(\mathcal{G})$. Then, any $(\epsilon,\phi)$-ED~$\mathcal{U}$ of $G$ satisfies
	\begin{equation*}
		\left| E^* \setminus \mathcal{U}\right| \ge \frac{4}{5} \cdot \left| E^*\right| \, .
	\end{equation*}
\end{lemma}

\begin{proof} [Proof sketch]
    By definition of the graph, there are $\Theta(dn)$ edges in the graph. Since there is at most an $\eps$ fraction of crossing edges, there are only $O(\eps dn)$ crossing edges. Note that there are~$\Theta(dn)$ edges in $G[T]$, so only an $O(\eps)$ fraction of the edges in $G[T]$ are crossing edges. Furthermore, $G[T]$ is a regular expander: this implies that there is a large cluster $U^* \in \mathcal{U}$ comprising a $1-O(\eps)$ fraction of $T$, together with a $1-O(\eps)$ fraction of important vertices. The latter is true since the edges between $S$ and $T$ make up a constant fraction of the total volume and only a small fraction of the edges can be crossing. We refer the reader to \Cref{claim:giantcluster} for more details.

    The edges in the subgraph $G[S]$ also account for a constant fraction of the total volume. Together with the fact that each $S_i$ induces a regular expander, for many of the $S_i$'s we will have a cluster in $\mathcal{U}$ that contains most of $S_i$. Now consider a set $S_i$ such that the important vertex of~$S_i$ is in~$U^*$ and most of the vertices in $S_i$ are all in the same cluster. If most of the vertices in $S_i$ are also inside~$U^*$, then consider the cut from $(U^* \cap S_i )\setminus\{s_{i,K}\} $ to $(U^* \setminus S_i) \cup \{s_{i,K}\}$. The cut size is at most the number of important edges in $S_i$, which is $O(d)$. On the other hand, the volume of the cut is $\Theta(dm)$ since most of the vertices of $S_i$ are in the cluster. Recalling that $m \gg \frac{1}{\phi}$, one concludes that the cut is sparse. Therefore, we ruled out the possibility of having many vertices of~$S_i$ in~$U^*$. See \Cref{claim:separateclusters} for a detailed discussion. 
    
    In summary, as illustrated in \Cref{fig:decomp1}, in any valid expander decomposition, most of the vertices in $T$ and most of the important vertices are inside a giant cluster $U^*$, and for most of the $S_i$'s, there is a cluster other than $U^*$ that contains most of the vertices in $S_i$. For any such $S_i$, most of the important edges inside it are then crossing edges.
\end{proof}
\noindent
The second level expander decomposition, i.e. an expander decomposition of the inter-cluster edges from the first level, is also quite structured, as illustrated in the ideal RED of \Cref{fig:decomp1}. 

\begin{lemma}[Informal, see \cref{claim:pkrecover}] \label{lem:sec-round-ov}
    Let $G=(V,E) \in \supp(\mathcal{G})$, and let $\mathcal{U}_1,\mathcal{U}_2$ be any $2$-level $(\epsilon,\phi)$-RED  of $G$. Then, there are at most $n/10$ vertices in $S$ that are non-isolated vertices\footnote{We call a vertex $v$ non-isolated in a decomposition $\mathcal{U}$ if $\mathcal{U}$ puts $v$ in a cluster with other vertices, i.e. $v$ does not constitute a singleton cluster in $\mathcal{U}$.} in~$\mathcal{U}_2$. Moreover, at least a $2/3$ fraction of important edges are not in $E\setminus \mathcal{U}_2$, i.e. a $2/3$ fraction of important edges are inside clusters of $ \mathcal{U}_2$.
\end{lemma}

\begin{proof}[Proof sketch]
    The number of crossing edges in the first level decomposition is $O(\eps d n)$, which is much less than $n$ since $d \ll \frac{1}{\eps}$. This means that most of the vertices are isolated vertices in the second level. Moreover, by \Cref{lem:special-edge-ov}, most of the important edges are crossing edges in the first level decomposition. Among these $\Theta(dn/m)$ edges, at most $O(\eps^2 dn)$ edges can be crossing edges in the second level decomposition. Recalling that $m \ll \frac{1}{\eps^2}$, we see that most of the important edges are not crossing edges in the second level decomposition.
\end{proof}

\begin{figure}[h]
	\begin{minipage}[center]{\textwidth}
		\centering
		\includegraphics[scale=0.83]{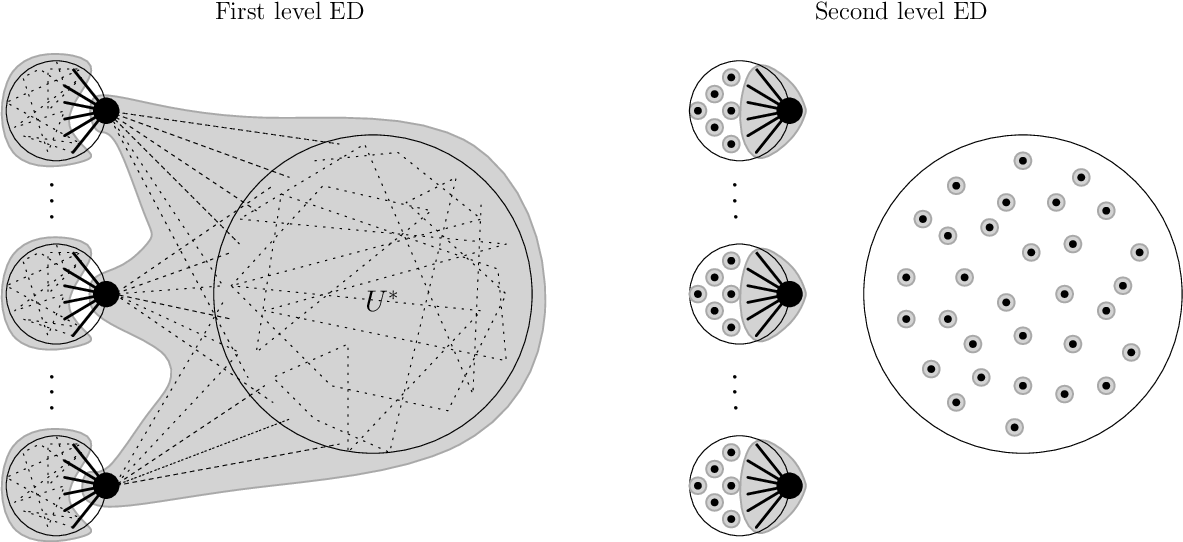}
	\end{minipage}
	\caption{Illustration of the ideal expander decomposition of the graph. Thick bullets represent important vertices, smaller bullets represent ordinary vertices, thick lines represent important edges, dotted lines represent the rest of the edges. Grey areas represent the clusters in the decomposition.}
	\label{fig:decomp1}
\end{figure}

\subsubsection*{Lower bound via communication complexity}
Our streaming lower bound will be proven in the two-player one-way communication model. In this setting, Alice gets the edges in $G[S]$ and $G[T]$, and Bob gets the edges between $S$ and $T$. We prove that in order to give a two-level RED, Alice needs to send $\Omega(dn)$ bits of information to Bob.

The high level idea is the following. Note that the identity of the important vertices can be only revealed by edges given to Bob. Thus, given Alice's input, every vertex in $S$ has the same probability to be an important vertex, which means that every edge in $S$ has the same probability to be an important edge. Therefore, in order to make sure that Bob recovers most of the important edges (which is morally equivalent to computing a two-level RED, as suggested by \Cref{fig:decomp1}), Alice needs to send most of the edges in $S$ to Bob, which is $\Omega(dn)$. 

To make the above idea concrete, we consider a communication problem where Alice is given a graph, and Bob is asked to output a not too large set of pairs that contains a good fraction of the edges of Alice's input. We first reduce this new problem to the two-level RED problem. Then, we will prove a communication complexity lower bound for this problem.
\begin{definition}[Informal, see \cref{def:recover}]
	\label{def:recoversimple}
    In the communication problem $\recover$, Alice's input is a graph $G'=(S,E')$ where $\card{S}=n/2$, and Bob's output is a set of pairs of vertices $F~\subseteq~\binom{S}{2}$ that must satisfy $|F| \le nm/10$ and $|F \cap E'| \ge \Omega(|E'|)$, i.e. at least a constant fraction of the edges in $G'$ are in $F$.
\end{definition}

\noindent
Now, the idea is to plant Alice's input for $\recover$ into our instance from \cref{def:harddistrsimple}. The input distribution for $\recover$ is sampled from  the distribution $\mathcal{G}'$. In other words, the distribution of Alice's input is the same as the left-hand side part of $G \sim \mathcal{G}$ (i.e. the subgraph $G[S]$). Then, in the reduction, Alice gets the edges of $G[S]$ and $G[T]$ while Bob gets the edges between $S$ and $T$. By virtue of the discussion in the previous section, we expect a RED of $G$ to allow Bob to recover many important edges in $G[S]$. Hence, Bob could simulate the RED algorithm for all the $m/2$ possible choices of the random index $K \sim [m/2]$ that defines the important edges (see \cref{def:harddistrsimple}): in this way, the $k$-th RED should reveal information about Alice's edges that are incident on the vertices $\{s_{i,k}\}_i$. Therefore, by varying $k$ over $[m/2]$, Bob should obtain information about all the edges in Alice's graph. More precisely, the reduction is the following.

\begin{reduct*}[Informal, see \cref{red:recover}]
	Let $\mathcal{A}$ be a deterministic streaming algorithm for computing a $2$-level $(\epsilon,\phi)$-RED. Alice, given her input graph $G'=(S,E')$ generated by $\mathcal{G}'$, feeds her edges $E'$ to $\mathcal{A}$, together with the fixed edges of $G[T]$. Then, she sends the memory state of $\mathcal{A}$ to Bob. Upon receiving the message, Bob makes $m/2$ copies  of $\mathcal{A}$ and initialises them to the memory state he received from Alice. Call these copies $\mathcal{A}_1, \dots, \mathcal{A}_{m/2}$. Next, for each $k \in [m/2]$, call $G_k = (V,E_k)$ the graph we obtain from~$\mathcal{G}$ when $K=k$ and the left-hand side $G_k[S]$ is exactly Alice's input $G'$ (so that the vertices $\{s_{i,k}\}_i$ are the important vertices in $G_k$, see \cref{def:harddistr}). Then, Bob feeds the edges $E_k(S,T)$ to $\mathcal{A}_k$. Let then $\mathcal{U}^k_1,\mathcal{U}^k_2$ be the RED output by $\mathcal{A}_k$. Bob finally constructs his output set $F$ as follows: for each $k \in [m/2]$, add the pair $\{s,s_{i,k}\}$ to $F$ for every $i \in [n/m]$ and every $s \in S_i$ that is not an isolated vertex in $\mathcal{U}^k_2$.
\end{reduct*}
\noindent
By \cref{lem:sec-round-ov}, the number of non-isolated vertices in the second-level decomposition is at most $n/10$. Hence, the total number of pairs added to $F$ is at most $nm/10$, thus satisfying the first requirement of \cref{def:recoversimple}. Moreover, by \cref{lem:sec-round-ov}, at least a $2/3$ fraction of the important edges are not crossing edges in the second level. This means that for each $k$, at least a $2/3$ fraction of the edges that are incident on $s_{i,k}$ is added to $F$. In turn, this implies that $F$ contains at least a constant fraction of the edges in~$E'$, thus satisfying the second requirement of \cref{def:recoversimple}. More precisely, one can prove the following.

\begin{lemma}[Informal, see \cref{lem:ind-recover}]
	\label{lem:informalreduction}
	If there is a deterministic $L$-bit space streaming algorithm $\mathcal{A}$ that computes a $2$-level $(\epsilon,\phi)$-RED with constant probability over inputs $G \sim \mathcal{G}$, then there is a deterministic protocol $\mathcal{R}$ that solves $\recover$ with constant probability over inputs $G' \sim \mathcal{G}'$. The communication complexity of $\mathcal{R}$ is at most $L$.
\end{lemma}

\noindent
The final component of the proof is the communication complexity lower bound for $\recover$.

\begin{lemma}[Informal, see \cref{lem:lb-recover}]
	\label{lem:informalrecover}
    The one-way communication complexity of solving $\recover$ with constant probability over inputs sampled from $\mathcal{G}'$ is $\Omega(dn)$.
\end{lemma}

\begin{proof}[Proof sketch]
    Roughly speaking, we show that the posterior distribution of the input conditioned on the output is shifted away from its prior distribution.
    
    Recall that when the input $G'=(S,E')$ is sampled from $\mathcal{G}'$, the graph is a disjoint union of $n/m$ random graphs with $m/2$ vertices each and degree~$\approx~d$. There are roughly 
    $$\left(\binom{m/2}{d}^{m/2}\right)^{n/m} \approx \left(\frac{m}{2d}\right)^{dn/2}$$ 
    possible inputs in total and the information complexity is $\Omega(dn)$. Conditioning on the the output~$F$ of  a correct protocol for $\recover$, the number of possible inputs is greatly decreased. In particular, to determine the input $E'$, we need to select a constant fraction, say $2/3$ for example, of the pairs from $F$, and select the rest of the edges (a $1/3$ fraction, in our example)  arbitrarily. Since $|F|$ is at most $nm/10$, $|E'| = \Theta(dn)$, and \smash{$\binom{S}{2} \approx nm/2$}, the total number of possible inputs is then roughly 
    $$\binom{nm/10}{2/3 \cdot dn} \cdot \binom{nm/2}{1/3 \cdot dn} \approx \left(\frac{3m}{20d}\right)^{1/3 \cdot dn} \cdot \left(\frac{3m}{2d}\right)^{1/6 \cdot dn} < \left(\frac{m}{3d}\right)^{ dn/2}.$$
    This means the information complexity of the input is decreased by a constant factor, which means that the protocol needs to communicate $\Omega(dn)$ bits of information.
\end{proof}

\noindent
Finally, one can conclude the main result (\Cref{th:lb}) combining \cref{lem:informalreduction} and \cref{lem:informalrecover}. The formal proofs and definitions are deferred to \cref{sec:lb}.

\section{Preliminaries}
\label{sec:prelims}

In this section, we introduce definitions and notation that we will use in the remainder of the paper.

\paragraph{Boundary-linked expander decomposition.}
In this paper we work heavily with boundary-linked EDs~\cite{boundlink}. A boundary-linked ED is the same as a classical ED except that Property~\eqref{property:expanderclassic} of \cref{def:expdecclassic} is strengthened. 

For a cut $S$ in a cluster $U$, the boundary $\unbor{S}{U}$ of $S$ with respect to $U$ is the number of edges that go from $S$ to the outside of $U$, i.e. $\unbor{S}{U}=|E(S,V \setminus U)|$. For $U \subseteq V$ and $\tau \geq 0$, the $\tau$-boundary linked subgraph of $G$ on $U$, denoted by $G[U]^{\tau}$, is the subgraph of $G$ induced by $U$ with additional \smash{$\tau \cdot \unbor{\{u\}}{U}$} self-loops attached to every $u \in U$.
	For cuts $S \subseteq U$, we adopt the following shorthand notation:
	\begin{itemize}
		\item the volume $\vol_{G[U]^{\tau}}(S)$ is denoted by \smash{$\vol_U^{\tau}(S)$};
		\item the sparsity $\Phi_{G[U]^{\tau}}(S)$ in $G[U]^{\tau}$ is denoted by $\Phi_U^{\tau}(S)$, which is equal to
		\begin{equation*}
			\Phi_U^{\tau}(S) = \frac{\uncut{S}{U}}{\min\{\vol_U^{\tau}(S), \vol_U^{\tau}(U \setminus S)\}} \, .
		\end{equation*}
	\end{itemize}
Let $\phi \in (0,1), b \in [\phi, 1)$ be some parameters. Most of the time in this paper, the value of~$\tau$ for all instances of boundary linked subgraphs will be the same and equal to the ratio $b/\phi$. Hence we adopt additional shorthand notation: we use \smash{$G[U]^{\circ}$}, \smash{$\vol_U^\circ(S)$}, \smash{$\Phi_U^{\circ}(S)$}, \smash{$\Phi_U^{\circ}$} for referring to \smash{$G[U]^{{b/\phi}}$, $\vol_U^{b/\phi}(S)$, $\Phi_U^{b/\phi}(S)$, $\Phi_U^{b/\phi}$} respectively.
One can observe that
\begin{equation*}
	\vol_U^\circ(S) = \vol_U^{b/\phi}(S) = \vol(S)+\left(\frac{b}{\phi}-1\right)\unbor{S}{U} \, .
\end{equation*}
\noindent
Then a boundary-linked ED can be defined as follows.

\begin{definition}[Boundary-linked expander decomposition~\cite{boundlink}]
	\label{def:expdec}
	Let $G=(V,E)$, let $b, \epsilon, \phi \in \, (0,1)$ be parameters such that $b \ge \phi$, and let $\gamma \ge 1$ be an error parameter. A partition $\mathcal{U}$ of $V$ is a $(b, \epsilon, \phi, \gamma)$-boundary-linked expander decomposition (for short, $(b, \epsilon, \phi, \gamma)$-BLD) of $G$ if
	\begin{enumerate}
		\item \label{property:crossing} $\frac{1}{2}\sum_{U \in \mathcal{U}} \uncutg{U} \le \epsilon |E|$, and
		\item \label{property:expander} for every $U \in \mathcal{U}$, $G[U]^{\circ}$ is a $ \phi/\gamma$-expander.
	\end{enumerate}
\end{definition}

\paragraph{Expander decomposition sequence.} For a partition $\mathcal{U}$ of $V$ (think of $\mathcal{U}$ as an ED of $G$), we denote by $E \setminus \mathcal{U}$ the set of inter-cluster (or crossing) edges with respect to $\mathcal{U}$, i.e.
\begin{equation*}
	E \setminus \mathcal{U} = E \setminus \bigcup_{U \in \mathcal{U}} \binom{U}{2} \, .
\end{equation*}
Analogously, we let $G\setminus \mathcal{U} = (V,E\setminus \mathcal{U})$ be the subgraph of $G$ obtained by removing the intra-cluster edges in $\mathcal{U}$.  For a sequence of partitions $\mathcal{U}_1, \dots, \mathcal{U}_\ell$ of $V$ and $i \in [\ell]$, we define $G^{\remove}_1=G$ and denote by $G^{\remove}_{i+1} = G^{\remove}_{i} \setminus \mathcal{U}_i$ the subgraph of $G$ obtained by removing the intra-cluster edges of the first $i$ partitions.

\begin{definition}[Removal-based ED sequence]
	\label{def:expdecrec}
	Let $G=(V,E)$, let $\epsilon, \phi \in (0,1)$, let $\ell \ge 1$, and let $\mathcal{U}_1, \dots, \mathcal{U}_\ell$ be a sequence of partitions of $V$. The sequence $\mathcal{U}_1, \dots, \mathcal{U}_\ell$ is an $\ell$-level removal-based $(\epsilon, \phi)$-ED sequence (for short, $\ell$-level $(\epsilon, \phi)$-RED or $(\epsilon,\phi,\ell)$-RED) of $G$ if, for all $i \in [\ell]$, $\mathcal{U}_i$ is an $(\epsilon, \phi)$-ED of the graph~$G^{\remove}_i$.
\end{definition}

\paragraph{Sparsifier-specific notation.}
Since our algorithms will be working on sparsifiers of the input graph, it is convenient to have short-hand notation for the above quantities in the sparsifier. For a weighted subgraph $H=(V,E',w)$ of $G$, a cluster $U \subseteq V$ and a cut $S \subseteq U$, we have ``sparsified'' counterparts $\wcutg{S}$, $\wvol{S}$, $\wcut{S}{U},\wbor{S}{U}$ for $\uncutg{S}$, $\vol(S)$, $\uncut{S}{U},\unbor{S}{U}$. More precisely, $\wcutg{S}=w(S,V\setminus S)$, $\wcut{S}{U}=w(S,U\setminus S)$, $\wbor{S}{U}=w(S,V\setminus U)$, and $\wvol{S}$ is the sum of weighted degrees of vertices of $S$ in $H$. Then, $H[U]^\circ$ denotes the subgraph of $H$ induced by $U$ (retaining edge weights from $w$), where every vertex is attached  \smash{$b/\phi \cdot \wbor{\{x\}}{U}$} self-loops. We then also introduce notation for the ``sparsified'' version of $\unvolc{S}{U}$ and $\unspc{S}{U}$: we denote $\vol_{H[U]^{b/\phi}}(S)$ and $\Phi_{H[U]^{b/\phi}}(S)$ by $\wvolc{S}{U}$ and $\wspc{S}{U}$ respectively. For clarity, we note that \smash{$\wvolc{S}{U} = \wvol{S}+\frac{b-\phi}{\phi}\wbor{S}{U}$}, reflecting the fact that \smash{$\vol_U^\circ(S) = \vol(S)+\frac{b-\phi}{\phi}\unbor{S}{U}$}.
\section{Sparsification for vertex-induced subgraphs}
\label{sec:spars}
In order to give a space efficient implementation of \cref{alg:metaed}, we will take a subroutine computing an approximate most-balanced sparse cut of a graph, and run it as a black box on a vertex-induced subgraph of a sparsifier. From such sparsifier, we demand that the local cuts are approximated to within an additive error proportional the corresponding global cut. Formally, we rely on the following cut preserving property for vertex-induced subgraphs.
\begin{definition}[Cluster sparsifier]
	\label{def:spars}
	Let $G=(V,E)$, let $\delta \in (0,1)$, let $U \subseteq V$ be a cluster, and let $H=(V,E',w)$ be a weighted subgraph of $G$ on the vertex set $V$. Then, $H$ is a cluster sparsifier for $U$ in $G$ with error~$\delta$, denoted by $H[U] \approx_\delta G[U]$, if one has
	\begin{equation*}
		\forall S \subseteq U,\,\uncut{S}{U} - \delta \cdot \uncutg{S} \le \wcut{S}{U}  \le \uncut{S}{U} + \delta \cdot \uncutg{S} \quad \text{ and }  \quad \forall S \subseteq V,\, (1-\delta)\cdot \uncutg{S}\le \wcutg{S} \le (1+\delta)\cdot \uncutg{S}\, .
	\end{equation*}
\end{definition}
\noindent
In \cref{def:spars}, the second property is the same as the first one when $U=V$. Therefore, if one is able to sample a sparsifier where the first property holds with high probability for any cluster $U\subseteq V$, then we get a sparsifier that with high probability embodies \cref{def:spars}. Perhaps surprisingly, we show that classical constructions of cut sparsifiers give the first property of \cref{def:spars}. Specifically, we use the approach of Fung et al.~\cite{fungspars}, based on sampling edges with probability proportional to their edge connectivity, together with its dynamic stream implementation by Ahn, Guha, and McGregor~\cite{sparsdynamicstreams}. Hence, we obtain the following sparsification lemma.

\begin{restatable}[Cluster sparsifiers in dynamic streams]{lemma}{maintechlemadd}
	\label{lem:distrcutstreamadd}
	Let $G=(V,E)$ be an $n$-vertex graph and let $\delta \in (0,1)$. Then, there is a distribution $\mathcal{D}_\delta$ such that, for any cluster $U \subseteq V$ and a sample $H=(V,E',w)  \sim \mathcal{D}_\delta$ one has $H[U] \approx_\delta G[U]$ with high probability, and $|E'| = \otil(n/\delta^2)$ edges. When $G$ is given in a dynamic stream, there is an algorithm that maintains a linear sketch of $G$ in $\otil(n/\delta^2)$ space and decodes it to output a weighted subgraph $H=(V,E',w)$  in $\otil(n/\delta^2)$ space and $\poly(n)$ time such that $H \sim  \mathcal{D}_\delta$.
\end{restatable}
\noindent
The distribution \smash{$\mathcal{D}_\delta$} we use is a cosmetic modification of the algorithm of Ahn, Guha, and McGregor~\cite{sparsdynamicstreams}, hereafter referred to as the AGM algorithm. It provides a dynamic stream implementation of the sampling scheme of Fung et al.~\cite{fungspars} for constructing cut sparsifiers. Even though we deviate only slightly from the original analysis of Ahn, Guha, and McGregor and Fung et al., we give a complete proof for the sake completeness. In \cref{subsec:primer}, we give a few preliminaries on sparsification and dynamic streams that are needed to present and discuss the algorithm. In \cref{subsec:properties}, we outline the algorithm and prove the technical lemmas that characterise its correctness. In \cref{subsec:cutcounting}, we use standard cut counting arguments to combine such technical lemmas into proving \cref{lem:distrcutstreamadd}.

\subsection{Sparsification primer}
\label{subsec:primer}
We give a few preliminaries needed to present the AGM algorithm. We start with the following result: for a graph given in a dynamic stream, we can output a subgraph that preserves the distinction between $k$-edge connected and not $k$-edge connected cuts.

\begin{theorem}[Connectivity witness \cite{linmeas}]
	\label{th:connwit}
	Let $G=(V,E)$ be a graph given in a dynamic stream, and let $k \ge 1$ be an integer parameter. Then there is an algorithm $\textsc{ConnWit}(k)$ that maintains a linear sketch of $G$ in $\otil(kn)$ space. With high probability, the algorithm decodes the sketch to output a subgraph $G'=(V,E')$ in $\otil(kn)$ space and $\poly(n)$ time such that:
	\begin{enumerate}
		\item for all $\emptyset \neq S \subsetneq V$, if $|E(S,V\setminus S)| < k$ then $E(S, V \setminus S) \subseteq E'$;
		\item for all $\emptyset \neq S \subsetneq V$, if $|E(S,V\setminus S)| \ge k$ then $|E(S, V \setminus S) \cap E'|\ge k$;
		\item $|E'|=O(kn)$.
	\end{enumerate}
\end{theorem}

\begin{proof}
	The algorithm works by maintaining $k$ independent linear sketches for dynamic spanning forest~\cite{linmeas}, each taking $\otil(n)$ space to maintain and $\otil(n)$ space and $\poly(n)$ time to decode. We use the first sketch to recover a spanning forest of $G$. Then, we feed edge removals for every edge in the spanning forest to the second to $k$-th sketch. We proceed in this fashion until we have recovered $k$ edge-disjoint maximal spanning forest $F_1, \dots, F_k \subseteq E$ of $G$. Then we let $E'$ be the union of these spanning forests. This takes $k \cdot \otil(n)$ space and one has $|E'| = O(kn)$.
	
	Now consider a cut $\emptyset \neq S \subsetneq V$. Clearly $E'(S,V \setminus S) \subseteq E(S,V \setminus S)$. Now suppose $E(S,V \setminus S) < k$. For the sake of a contradiction, let us say $E(S,V \setminus S) \setminus E' \neq \emptyset$. By disjointness of the $F_i$'s, there must be a $j \in [k]$ such that the residual graph $G \setminus (\cup_{i=1}^{j-1} F_i)$ still contains an edge from $E(S,V \setminus S)$. Hence, there are two components of the spanning forest $F_j$ which could be joint by this edge, which contradicts its maximality.
	
	Consider the case of $E(S,V \setminus S) \ge k$, and suppose for the sake of a contradiction that $E'(S,V \setminus S) < k$. By disjointness of the $F_i$'s there must be a forest $F_j$ such that $F_j \cap E(S,V \setminus S) = \emptyset$ (as otherwise $E'(S,V \setminus S) \ge k$). As before, this contradicts the maximality of $F_j$.
\end{proof}
\noindent
Hereafter let $C>0$ be a large enough constant, and define for every $e=\{u,v\} \in \binom{V}{2}$ the ideal sampling probabilities
\begin{equation*}
	p_e = \min \left\{1, \frac{1}{\lambda_e}\cdot\frac{C \log^3 n}{\delta^2} \right\} \, ,
\end{equation*}
where $\lambda_e$ is the edge connectivity of $e$, i.e. the minimum number of edges crossing a $uv$-cut in $G$. The original AGM algorithm estimates edge connectivities (and hence sampling probabilities) via the sparsifier that is being constructed, in some sense. Here, we simplify the analysis and instead use a separate sparsifier for that purpose. For instance, one can use a spectral sparsifier, defined as follows.
\begin{definition}[Spectral sparsifier]
	Let $G=(V,E)$ and let $\xi \in (0,1)$. A $(1\pm \xi)$-spectral sparsifier of $G$ is a weighted subgraph $\tilde{G} = (V,\tilde{E},\tilde{w})$ such that
	\begin{equation*}
		\forall \vt{x} \in \mathbb{R}^V , \quad (1-\xi)\vt{x}^\top L \vt{x} \le \vt{x}^\top \tilde{L} \vt{x} \le (1+\xi)\vt{x}^\top L \vt{x} \, ,
	\end{equation*}
	where $L$ and  $\tilde{L}$ are the Laplacian matrices of $G$ and $\tilde{G}$, respectively.
	The Laplacian matrix of a graph is $L=D-A$, where $D$ denotes the degree diagonal matrix of the graph and $A$ denotes its (weighted) adjacency matrix
\end{definition}\noindent
It is known how to construct such sparsifiers in dynamic streams.
\begin{theorem}[\cite{spectralspars}]
	\label{th:spectral}
	Let $G=(V,E)$ be a graph given in a dynamic stream, and let $\xi \in (0,1)$. Then, there is an algorithm $\textsc{Spectral}(\xi)$ that maintains a linear sketch of $G$ in $\otil(n/\xi^2)$ space. With high probability, the algorithm decodes the sketch to output in $\otil(n/\epsilon^2)$ time and space a weighted subgraph  $\tilde{G} = (V,\tilde{E},\tilde{w})$ with $|\tilde{E}| = \otil(n/\epsilon^2)$ that is a $(1\pm \xi)$-spectral sparsifier of $G$.
\end{theorem}

\subsection{Sampling algorithm for dynamic streams}
\label{subsec:properties}
With all the preliminaries in place, we can describe (our version of) the AGM algorithm. The high level approach consists in sampling edges at geometric rates as they arrive in the stream, and then use the $i$-th connectivity witness to recover the edges $e$ that are sampled at rate $p_e \approx 2^i$. Algorithm~\ref{alg:dynsampling} outlines this process, where we denote by $\tilde{\lambda}_e$ the edge connectivity of \smash{$e \in \binom{V}{2}$} in \smash{$\tilde{G}$}. As one can see from the description of the algorithm, we actually let $w$ be a function $E \rightarrow \mathbb{R}_{\ge 0}$ where $w(e)=0$ for all $e \in E \setminus E'$. This facilitates the discussion.

\begin{algorithm}
	\caption{\textsc{agm}: sparsifier construction in dynamic streams}
	\label{alg:dynsampling}
	
	\begin{algorithmic}[1]
		\LeftComment{$\delta \in (0,1)$ is the error parameter}
		\LeftComment{$C >0$ is a large enough constant}
		\LeftComment{$k = 16 C \cdot \delta^{-2}\cdot \log^3 n$}
		\Statex
		\Procedure{\textsc{PreProcessing}}{}
		\State $\textsc{Spectral}\gets $ instance of $\textsc{Spectral}(1/2)$ from \cref{th:spectral}
		\For{$i=1,\dots,{2 \log n}$}
		\State $h_i \sim \unif\left(\{0,1\}^{\binom{V}{2}}\right)$ \Comment{sample $2 \log n$ independent uniform hash functions}
		\State $\textsc{ConnWit}_i(k) \gets $ $i$-th independent instance of $\textsc{ConnWit}(k)$ from \cref{th:connwit} \label{lineagm:feedconnwit}
		\EndFor
		\EndProcedure
		\Statex
		\Procedure{\textsc{OnUpdate}$(e$)}{} \Comment{$e \in \binom{V}{2}$}
		\State feed the update for $e$ to $\textsc{Spectral}$
		\For{$i = 0, \dots, 2\log n$}
		\If{$\prod_{j = 1}^{i} h_j(e)=1$}
		\State feed the update for $e$ to $\textsc{ConnWit}_i(k)$
		\EndIf
		\EndFor
		\EndProcedure
		\Statex
		\Procedure{\textsc{PostProcessing}}{}
		\State $\tilde{G}=(V,\tilde{E},\tilde{w}) \gets $ result of $\textsc{Spectral}$
		\For{$i=0,\dots,{2 \log n}$}
		\State $G_i= (V,E_i) \gets (V,\{e \in E: \prod_{j = 1}^{i} h_j(e)=1\})$ \Comment{used for convenience of analysis only}
		\State $G_i'=(V,E'_i) \gets $ result of $\textsc{ConnWit}_i(k)$ \Comment{used by the algorithm}
		\EndFor
		\For{$e \in \binom{V}{2}$}
		\State $\tilde{p}_e \gets \min \left\{1, \frac{2}{\tilde{\lambda}_e}\cdot\frac{C \log^3 n}{\delta^2} \right\} $
		\State $j_e \gets  \left\lfloor \log \frac{1}{\min\{1, \tilde{p}_e\}}\right\rfloor $.
		\If{$e \in E'_{j_e}$}
		\State $w(e) \gets 2^{j_e}$
		\Else
		\State $w(e)\gets 0$
		\EndIf
		\EndFor
		\State $E' \gets \supp(w)$
		\State \Return $H=(V,E',w)$.
		\EndProcedure
	\end{algorithmic}
\end{algorithm}

We remark that Algorithm~\ref{alg:dynsampling} assumes access to $\poly(n)$ bits of randomness. It is known that one can lift this assumption by using Nisan's pseudorandom generator~\cite{nisansprg,lpindyk} with only a $\pylog(n)$ space blow-up factor (and a $\poly(n)$ time blow-up). We remark that we can do this because the random bits are only used during the stream processing, and we only compute linear sketches in this phase, which in particular are oblivious to the order of updates in the stream. Then we can already conclude a bound on the space complexity, as well as on the size of the output graph.
\begin{lemma}
	\label{lem:spacesampling}
	Algorithm~\ref{alg:dynsampling} takes $\otil(n/\delta^2)$ bits of space and $H$ has at most $\otil(n/\delta^2)$ edges.
\end{lemma}
\begin{proof}
	The edges of $H$ can only come from the output of one of the instances $\textsc{ConnWit}$. In other words, $E' \subseteq \cup_i E_i'$. From \cref{th:connwit} we also know that each instance of $\textsc{ConnWit}$ takes $\otil(kn)$ bits of space and that each $E_i'$ contains at most $\otil(kn)$ edges. Hence, $|E'| = \otil(kn)=\otil(n/\delta^2)$ and Algorithm~\ref{alg:dynsampling} takes \smash{$\otil(kn) = \otil(n/\delta^2)$} bits of space.
\end{proof}

\noindent
We continue our analysis by bounding the probability that a certain bad event happens, which corresponds to the sampling or sketching primitives failing to preserve a needed value. To be more specific about this, first define for every $i =0, \dots, 2\log n$ the event $\mathcal{B}_i$ that instance $i$ of Theorem~\ref{th:connwit} $\textsc{ConnWit}_i(k)$ fails to correctly construct the graph $G_i'$. In addition, for every $e \in \binom{V}{2}$, let \smash{$\mathcal{B}^{(1)}_e$} be the event that \smash{$\tilde{\lambda}_e < 1/2 \cdot \lambda_e$ or $\tilde{\lambda}_e > 3/2 \cdot \lambda_e$}. Finally, define for every $e \in E$ the quantity
\begin{equation*}
	\tau_e = \left\lfloor \log \frac{1}{\min\{1, p_e\}}\right\rfloor \, ,
\end{equation*} and let \smash{$\mathcal{B}^{(2)}_e$} be the event that there exists $j  \in \mathbb{N}$ with $\tau_e-2 \le j \le \tau_e$ such that $\lambda_e(G_{j}) \ge k$, where $k$ is defined in \cref{alg:dynsampling} and $\lambda_e(G_i)$ is the edge connectivity of $e$ in $G_i$ for $i=0,\dots,2\log n$. With this notation, the bad event we want to avert is
\begin{equation*}
	\mathcal{B} = \Big(\vee_{i=0}^{2 \log n} \mathcal{B}_i\Big) \vee \Big(\vee_{e \in \binom{V}{2}} \mathcal{B}^{(1)}_e\Big) \vee \Big(\vee_{e \in E} \mathcal{B}^{(2)}_e\Big) \, .
\end{equation*}
Intuitively, this is a bad event because the algorithm cannot trust the connectivity witnesses nor the edge connectivities computed via the spectral sparsifiers. On the contrary, when $\mathcal{B}$ does not happen, we know that all the graphs $G_i'$ are constructed correctly (thanks to $\neg  \mathcal{B}_i$), all the edges have $j_e$ that is not too large (thanks to \smash{$\neg  \mathcal{B}_e^{(1)}$}), and we can recover the sampled edges via the graphs~$G_i'$ (thanks to \smash{$\neg  \mathcal{B}_e^{(2)}$}).

\begin{lemma}
	\label{lem:failsampling}
	When running Algorithm~\ref{alg:dynsampling}, the bad event $\mathcal{B}$ does not happen with high probability.
\end{lemma}
\begin{proof}
	By \cref{th:spectral}, we know that $\vee_{e} \mathcal{B}^{(1)}_e$ does not happen with high probability for any $e \in \binom{V}{2}$. Now let $i=0,\dots, 2 \log n$, and note that the randomness used to define the input $G_i$ to $\textsc{ConnWit}_i(k)$ and the randomness used by $\textsc{ConnWit}_i(k)$ itself are independent. Hence, by \cref{th:connwit} and by union bound over $i$, we get that $\vee_{i}\mathcal{B}_i$ does not happen with high probability.  Lastly, we deal with the events $\mathcal{B}^{(2)}_e$.
	
	For ease of notation, define for any $\emptyset \neq S \subseteq V$ the size of its cut projected onto the edges of $G_i$ as $\partial^i S=|E_i \cap E(S,V \setminus S)|$. Then consider any $e=\{s,t\} \in E$ and any $j  \in \mathbb{N}$ with $\tau_e-2 \le j \le \tau_e$, and let $S_e \subseteq V$ be an $st$-cut in $G$ of minimum size, i.e. $\lambda_e = \uncutg{S_e}$. Observe that by the way the algorithm defines the sub-streams, each $f \in E$ is in $E_i$ with probability $2^{-i}$ for all $i=0, \dots, 2 \log n$, so in expectation the edge connectivity of $e$ has dropped below $4 C \cdot \delta^{-2}\cdot \log^3 n$ in the graph $G_{j}$. More precisely,
	\begin{align*}
		\ex[\partial^{j} S_e] = \uncutg{S_e} \cdot 2^{-j} & \le 4\cdot \uncutg{S_e} \cdot 2^{-\tau_e} \\
		&  \le 8\lambda_e \min\{1,  p_e\} \\
		& \le 8\lambda_e p_e  \\
		& \le 8 \lambda_e \frac{C \cdot \delta^{-2}\cdot \log n}{\lambda_e} = 8 C \cdot \delta^{-2}\cdot \log^3 n \, .
	\end{align*}
	Applying the Chernoff bound (since edges are sampled independently of each other)  we get that
	\begin{align*}
		\Pr\left[\partial^{j} S_e > 12 C \cdot \delta^{-2}\cdot \log^3 n \right] & \le \Pr\left[\partial^{j} S_e > \left(1+\frac{1}{2}\right) \ex[\partial^{j} S_e] \right] \\
		& \le \exp\left(-\frac{\ex[\partial^{j} S_e]}{12}\right) \\
		& \le \exp\left(-\frac{\ex[\partial^{\tau_e} S_e]}{12}\right)\, .
	\end{align*}
	In order to make sure that the exponential above is in fact small enough, we now have two cases to consider. First consider the possibility that $\tau_e= 0$, or equivalently $p_e > \frac{1}{2}$. This means two things: $E_{\tau_e}=E$ (see \cref{alg:dynsampling}), and $\lambda_e < 2 C \cdot \delta^{-2}\cdot \log^3 n$. Then, $\partial^{\tau_e} S_e=\lambda_e \le 12 C \cdot \delta^{-2}\cdot \log^3 n$ with probability $1$, so we need not worry about this case. For the other possibility that $ p_e \le \frac{1}{2}$, one can deduce that $p_e = C \cdot \lambda_{e}^{-1} \delta^{-2}\cdot \log^3 n$ and thus lower bound the expectation of $\partial^{\tau_e} S_e$ as
	\begin{equation*}
		\ex[\partial^{\tau_e} S_e] = \partial S_e \cdot 2^{-\tau_e} \ge \lambda_e p_e = C \cdot \delta^{-2}\cdot \log^3 n \, ,
	\end{equation*}
	Hence, the Chernoff bound above gives $\partial^{\tau_e} S_e \le 12 C \cdot \delta^{-2}\cdot \log^3 n < k$ with high probability. Thus, the edge connectivity of $e$ in $G_{j}$ must be at most $\partial^{j} S_e < k$. Taking a union bound gives that $\vee_{e} \mathcal{B}^{(2)}_e$ does not happen with high probability.
\end{proof}

\noindent
Finally, we show properties of the weight function produced by the algorithm. We first focus on the special case of edges with very low connectivity, i.e. with $p_e=1$, and show that these are indeed copied exactly in $E'$, as one would expect.

\begin{lemma}
	\label{lem:oneprobsampling}
	Let $H=(V,E',w)$ be the output of Algorithm \ref{alg:dynsampling}. If $\mathcal{B}$ does not happen, then for any edge $e \in E$ with $p_e =1$ we have $e \in E'$ and $w(e)=1$ .
\end{lemma}
\begin{proof}
	Assume $\mathcal{B}$ does not happen. This in particular means that $G_0'$ is correctly constructed, so by Theorem~\ref{th:connwit} any cut of size less than $16C \cdot \delta^{-2}\cdot \log^3$ is exactly preserved in $G_0'$. On the other hand, if an edge $e \in E$ has $p_e = 1$ then its edge connectivity $\lambda_e$ is at most $C \cdot \delta^{-2}\cdot \log^3 n$. Since $\tilde{\lambda}_e \le 2 \lambda_e$, it then follows that $\tilde{p}_e = 1 $, so $j_e=0$ and hence the claim.
\end{proof}

\noindent
For edges which we do not know to have $p_e=1$, but for which we at least know a lower bound $p$ on their sampling probability, we wish to show that they are preserved with ``very high'' probability. We do this by using a Chernoff-style bound of Ahn, Guha, and McGregor.

\begin{restatable}{lemma}{probsampling}
	\label{lem:probsampling}
	Let $H=(V,E',w)$ be the output of Algorithm \ref{alg:dynsampling}. Then for any edge set $X \subseteq E$, any $p \le \min_{e \in X} p_e$, and any $x \ge |X|$ one has
	\begin{equation*}
		\Pr\left[\neg \mathcal{B} \text{ and } \left|\sum_{e \in X} w(e) - |X|\right| > \delta \cdot x \right] \le 2 \exp\left(-\frac{1}{48} \cdot \delta^2 \cdot  p x\right).
	\end{equation*}
\end{restatable}
\begin{proof}
	For ease of notation, let $\tilde{\lambda}$ be the vector of edge connectivities $(\tilde{\lambda}_e)_{e \in E}$ from \cref{alg:dynsampling}, and $\lambda$ is the vector of edge connectivities $(\lambda_e)_{e\in E}$ of $G$. We will condition on the values $\hat{\lambda}$ that $\tilde{\lambda}$ can take, and define
	\begin{equation*}
		\hat{p}_e = \min \left\{1, \frac{2}{\hat{\lambda}_e}\cdot\frac{C \log^3 n}{\delta^2} \right\} \quad \text{ and } \quad  \hat{j}_e = \left\lfloor \log \frac{1}{\min\{1, \hat{p}_e\}}\right\rfloor  \, .
	\end{equation*}
	Then,
	\begin{align*}
		& \Pr\left[\neg \mathcal{B} \text{ and } \left|\sum_{e \in X} w(e) - |X|\right| > \delta \cdot x \right] \\
		= & \sum_{\hat{\lambda} \in \mathbb{N}^E} 	\Pr\left[\neg \mathcal{B} \text{ and } \left|\sum_{e \in X} w(e) - |X|\right| > \delta \cdot x  \text{ and } \tilde{\lambda}=\hat{\lambda} \right] \\
		= & \sum_{\hat{\lambda} \in \mathbb{N}^E \,:\, \frac{1}{2} \lambda \le {\hat{\lambda}} \le \frac{3}{2} \lambda	} \Pr\left[\neg \mathcal{B} \text{ and } \left|\sum_{e \in X} w(e) - |X|\right| > \delta \cdot x  \text{ and } \tilde{\lambda}=\hat{\lambda} \right] \, .
	\end{align*}
	By conditioning we get
	\begin{align*}
		& \Pr\left[\neg \mathcal{B} \text{ and } \left|\sum_{e \in X} w(e) - |X|\right| > \delta \cdot x \right] \\
		= & \sum_{\hat{\lambda} \in \mathbb{N}^E \,:\, \frac{1}{2} \lambda \le {\hat{\lambda}} \le \frac{3}{2} \lambda	} \Pr\left[\neg \mathcal{B}\text{ and } \left|\sum_{e \in X} w(e) - |X|\right| > \delta \cdot x  \middle| \tilde{\lambda}=\hat{\lambda} \right] \cdot \Pr\left[\tilde{\lambda}=\hat{\lambda}\right]
	\end{align*}
	We now analyse the first factor in each term of the sum above. We use three facts to do this: (1) whenever we have $\neg \mathcal{B}$, for all $e \in X$ one has that $e \in E_{j_e}'$ if and only if $e \in E_{j_e}$ (this is the case because \smash{$\frac{1}{2} \lambda \le {\tilde{\lambda}} \le \frac{3}{2} \lambda$}, so \smash{$\max\{0, \tau_e-2\} \le j_e \le \tau_e$, and the event $\mathcal{B}^{(2)}_e$} does not happen); (2) recall that, for $i=0,\dots,2\log n$, $e \in E_{i}$ if and only if $\prod_{t=1}^{i}h_t(e)=1$ (see \cref{alg:dynsampling}); (3) also recall that \cref{alg:dynsampling} assigns weight $2^{j_e}$ to an edge recovered from $E'_{j_e}$ (see \cref{alg:dynsampling}).
	Using these three facts, we get
	\begin{equation*}
		 \Pr\left[\neg \mathcal{B} \text{ and } \left|\sum_{e \in X} w(e) - |X|\right| > \delta \cdot x  \middle| \tilde{\lambda}=\hat{\lambda} \right]
		= \Pr\left[\neg \mathcal{B} \text{ and } \left|\sum_{e \in X} 2^{j_e}\prod_{t=1}^{j_e}h_t(e) - |X|\right| > \delta \cdot x  \middle| \tilde{\lambda}=\hat{\lambda} \right] \, .
	\end{equation*}
	Next, by the condition $\tilde{\lambda}=\hat{\lambda}$, we have $\hat{j}_e=j_e$ for all $e \in X$, so we get
	\begin{align*}
	\Pr\left[\neg \mathcal{B} \text{ and } \left|\sum_{e \in X} w(e) - |X|\right| > \delta \cdot x  \middle| \tilde{\lambda}=\hat{\lambda} \right]
		& =  \Pr\left[\neg \mathcal{B} \text{ and } \left|\sum_{e \in X} 2^{\hat{j}_e}\prod_{t=1}^{\hat{j}_e}h_t(e) - |X|\right| > \delta \cdot x  \middle| \tilde{\lambda}=\hat{\lambda} \right] \\
		& \le \Pr\left[\left|\sum_{e \in X} 2^{\hat{j}_e}\prod_{t=1}^{\hat{j}_e}h_t(e) - |X|\right| > \delta \cdot x  \middle| \tilde{\lambda}=\hat{\lambda} \right] \\
		& = \Pr\left[\left|\sum_{e \in X} 2^{\hat{j}_e}\prod_{t=1}^{\hat{j}_e}h_t(e) - |X|\right| > \delta \cdot x \right]  \, .
	\end{align*}
	Now we apply the following form of the Chernoff bound.
	\begin{fact}[\cite{streamexpdec}]
		Let $X_1,\dots,X_N$ be independent random variables distributed in $[0,a]$. Let $X= \sum_{r \in [N]} X_r$ and $\mu = \ex[X]$. Then, for $\zeta \in (0,1)$ and $\alpha \ge 0$ one has
		\begin{equation*}
			\Pr\left[\left| X- \mu \right| > \zeta \mu+\alpha \right] \le 2 \exp\left(-\frac{\zeta \alpha}{3a}\right)
		\end{equation*}
	\end{fact}
	\noindent
	Before applying this bound, we note that the random variables $2^{\hat{j}_e}\prod_{t=1}^{\hat{j}_e}h_t(e)$ are independent across $e \in X$ and have each mean $1$. Moreover, each of this random variables have value at most \smash{$2^{\hat{j}_e} \le 2^{\lfloor \log \frac{1}{\min\{1,p\}}\rfloor}$}, so we set $a = 1/p$. We also set $\zeta = \delta/2$ and $\alpha = \delta\cdot (x-|X|)+\delta|X|/2$. Then we get
	\begin{equation*}
		\Pr\left[\left|\sum_{e \in X} 2^{\hat{j}_e}\prod_{t=1}^{\hat{j}_e}h_t(e) - |X|\right| > \delta \cdot x \right] \le 2 \exp\left(-\frac{\delta^2}{48}px\right) \, .
	\end{equation*}
	This concludes the proof.
\end{proof}

\subsection{Cut counting and proving the sparsification lemma}
\label{subsec:cutcounting}
We hereafter fix a cluster $U\subseteq V$, and we partition the edge set $E$ into $F_i = \{e \in E: \, 2^i \le \lambda_e \le 2^{i+1} -1 \, \}$ for $i =0,\dots,\log n$. Note that these $F_i$'s are pairwise disjoint and their union gives~$E$. We recall that $E(A,B)$ denotes the set of edges in $E$ with one endpoint in $A$ and one in $B$, for $A,B \subseteq V$. Then, we write $\iuncut{S}{U}=|F_i \cap E(S,U \setminus S)|$ to denote the size of the projection of the local cut of $S$ onto the edge set $F_i$, for all $i=0,\dots, \log n$ and $\emptyset \neq S \subsetneq U$. Analogously, we also write $\iwcut{S}{U}=w(F_i \cap E(S,U \setminus S))$ to denote the total weight of the projection of the local cut of $S$ onto~$F_i$, where $w$ is the weight function of $H$.

We show that $H$ preserves the cuts of a cluster $U \subseteq V$ within each $F_i$ individually, with high probability. This follows by two of the technical lemmas we proved in the previous section, namely \cref{lem:oneprobsampling} and \cref{lem:probsampling}.

\begin{lemma} \label{lem:levelbound}
	Let $H=(V,E',w)$ be the output of \cref{alg:dynsampling}, and let $i =0,\dots,\log n$. Then
	\begin{align*}
		\Pr\left[ \neg B\,  \text{ and }\, \exists \, \emptyset \neq S \subsetneq U \, : \,  \left|\iwcut{S}{U} - \iuncut{S}{U}\right| > \delta \cdot \frac{\uncutg{S}}{1+\log n} \right] \le \frac{1}{\poly(n)}
		\, .
	\end{align*}
\end{lemma}
\begin{proof}
We proceed by first showing a probability bound for each cut $S$, and then combine them using a union bound to get the final result. Specifically, for all $j \in [2\log n]$ we define the set
\begin{equation*}
	\mathcal{Y}_{ij}=\{E(S,V \setminus S) : \, \emptyset \neq S \subsetneq U \text{ and } 2^{i+j-1} \le \uncutg{S} \le 2^{i+j}-1 \} \, ,
\end{equation*}
and we will show a probability bound for all $j \in [2\log n]$ and $Y \in \mathcal{Y}_{ij}$. Taking a careful union bound over all such choices will be enough to get the lemma statement. It shall be useful to define $\Pi_U^i(Y)=Y \cap F_i \cap \binom{U}{2}$ to be the projection onto $F_i$ and onto the cluster $U$ of the edges in $Y$, for any $Y \subseteq E$. Moreover, we denote by $\mathcal{F}_{Y}$ the failure event that
\begin{equation*}
	\left|w(\Pi_U^{i}(Y)) - |\Pi_U^{i}(Y)|\right| > \delta \frac{|Y|}{1+\log n} \, ,
\end{equation*}
For ease of writing we shall denote $\Pi_U^i(E(S,V \setminus S))$ and $\mathcal{F}_{E(S,V \setminus S)}$ by $\Pi_U^i(\uncutg{S})$ and $\mathcal{F}_{\uncutg{S}}$ respectively. To appreciate their purpose, we remark that for any $\emptyset \neq S \subsetneq U$,
\begin{align*}
	\mathcal{F}_{\uncutg{S}} \quad \iff \quad \left|\iwcut{S}{U} - \iuncut{S}{U}\right| > \delta \frac{\uncutg{S}}{1+\log n} \, 
\end{align*}
since $\Pi_U^{i}(\uncutg{S})=\Pi_U^{i}(E(S,V \setminus S)) = E(S,V \setminus S) \cap F_i \cap \binom{U}{2} = F_i \cap E(S,U\setminus S)$.

\begin{claim}
	\label{claim:boundprobsingleset}
	Let $j \in [2\log n]$ and $Y \in \mathcal{Y}_{ij}$. Then
	\begin{equation*}
		\Pr[\neg \mathcal{B} \text{ and } \mathcal{F}_Y] \le 2n^{-\frac{C}{768}2^j} \, .
	\end{equation*}
\end{claim}

\begin{proof}
	We start by observing that we can consider only the set $X = \{e \in \Pi_U^i(Y) \, : \, p_e < 1\}$, since edges with $p_e=1$ are always added to $E'$ with weight set to $1$ (when $ \mathcal{B}$ does not occur) by Lemma~\ref{lem:oneprobsampling}, so
	\begin{equation*}
		\left|w(\Pi_U^{i}(Y)) - |\Pi_U^{i}(Y)|\right| = \left|w(X) - |X|\right| \, .
	\end{equation*}
	If $X = \emptyset$, the probability of $\mathcal{F}_Y$ is zero, so let us consider the case that $X \neq \emptyset$, for which we want to upper bound the probability
	\begin{equation*}
		\Pr\left[\neg \mathcal{B} \text{ and } \left|w(X) - |X|\right| > \delta \frac{|Y|}{1+\log n}\right] \, .
	\end{equation*}
	\noindent
	To do so, we use Lemma~\ref{lem:probsampling} and set $p=\min \{1, C \cdot 2^{-(i+1)} \delta^{-2} \log^3 n\}$, which is a lower bound on all the probabilities $\{p_e\}_{e \in X}$. This is the case because every $e \in X$ has $p_e < 1$ and $\lambda_e \le 2^{i+1}$ (since $X \subseteq F_i$). We also note that $|Y| \ge |X|$, thus we can safely apply Lemma~\ref{lem:probsampling} and get
	\begin{align*}
		\Pr\left[\neg \mathcal{B} \text{ and } \left|w(X) - |X|\right| > \delta \frac{|Y|}{1+\log n}\right] & \le 2 \exp\left(-\frac{1}{48}\cdot \delta^2 \cdot p \cdot \frac{|Y|}{(1+\log n)^2} \right) \, .
	\end{align*}
	If we recall $Y \in \mathcal{Y}_{ij}$, and in particular $|Y| \ge 2^{i+j-1}$, we can then conclude the claim by plugging the value of $p$ into the exponent.
\end{proof}
\noindent
Now that we have bounded the probability that a single cut from $\mathcal{Y}_{ij}$ is not well preserved once projected onto $U$ and $F_i$, we show that the union over $j$ of the collections $\mathcal{Y}_{ij}$ is representative of all the local cuts of $U$.
\begin{claim}
	\label{claim:calyisenough}
	Let $\emptyset \neq S \subsetneq U$ such that $\Pi_U^i(\uncutg{S})\neq~\emptyset$. Then there exists $j \in [2\log n]$ so that $ E(S,V \setminus~S)~\in~\mathcal{Y}_{ij}$.
\end{claim}
\begin{proof}
	Since $\Pi_U^i(\uncutg{S}) \neq~\emptyset$, there exists $e=\{u,v\} \in E(S,V\setminus S)$ that is also in $F_i$. In particular this means that $E(S,V\setminus S)$ is a $uv$-cut in $G$, and $e$ has edge connectivity at least $2^i$, thus $\uncutg{S} \ge 2^{i}$. At the same time, $\uncutg{S} \le 2^{i+j}-1$ for some $j \in [2\log n]$ since $\uncutg{S} \le |E| \le n^2-1$.
\end{proof}
\noindent
Now, for all $j \in [2 \log n]$ and all $Z \subseteq F_i$, we define the worst-case edge set $Y^*_j(Z) \subseteq E$ of $Z$ as
\begin{equation*}
	Y^*_j(Z) \in \argmin_{Y \in \mathcal{Y}_{ij}: \Pi_U^i(Y) = Z}  |Y|\, .
\end{equation*}
The reason why we call $Y^*_j(Z)$ worst case for $Z$ is the following. Among all the global cuts in $\mathcal{Y}_{ij}$ whose projection onto $F_i$ and $U$ is $Z$,  the cut \smash{$Y^*_j(Z)$} is the one which allows for the smallest error term, which intuitively makes $\mathcal{F}_Y$ more likely to happen. Formally we have the following claim, which follows from the definition.

\begin{claim}
	\label{claim:worstcaseproj}
	If there exist $j \in [2\log n]$ and $Y \in \mathcal{Y}_{ij}$ such that the event $\mathcal{F}_{Y}$ happens, then there exist $j \in [2\log n]$ and $Z \in \{\Pi_U^i(Y'): Y' \in \mathcal{Y}_{ij}\}$ such that $\mathcal{F}_{Y_j^*(Z)}$ happens.
\end{claim}
\noindent
Claim~\ref{claim:calyisenough} and~\ref{claim:worstcaseproj} together suggest that we can restrain ourselves to union bound over $\{\Pi_U^i(Y): Y \in \mathcal{Y}_{ij}\}$, which is exactly what we are going to do. Before that, we make sure that this set is not too large.
\begin{claim}
	\label{claim:smallunionbound}
	For every $j \in [2\log n]$ we have $|\{\Pi_U^i(Y): Y \in \mathcal{Y}_{ij}\}| \le n^{2 \cdot 2^{j}}$.
\end{claim}
\begin{proof}
	We will employ the following cut counting result, where $E^k$ denotes the set of edges $e \in E$ with edge connectivity $\lambda_e \ge k$, for any integer $k$.
	\begin{theorem}[\cite{fungspars} Theorem 1.6]\label{th:counting}
		For any $k, s \ge 1$ one has
		\begin{equation*}
			|\{E(S,V\setminus S) \cap E^k \, : \, \emptyset \neq S \subsetneq V \text{ and } \uncutg{S} \le s\cdot k\}| \le n^{2s} \, ,
		\end{equation*}
		i.e. the number of distinct subsets of edges of connectivity at least $k$ from cuts of size at most $s \cdot k$ is at most $n^{2s}$.
	\end{theorem}
	\noindent
	Fix $j \in [2\log n]$. We recall that every $e \in F_i$ has $\lambda_e \ge 2^i$, which means $F_i \subseteq E^{2^{i}}$. Also, for any $Y \in \mathcal{Y}_{ij}$ there is by definition a cut $\emptyset \neq S \subsetneq V$ such that $\uncutg{S} \le 2^{i+j}-1$ and $Y=E(S,V\setminus S)$. Hence
	\begin{align*}
		|\{\Pi_U^i(Y) : \, Y \in \mathcal{Y}_{ij}\}|
		& \le |\{E^{2^i} \cap E(S,V\setminus S) : \emptyset \neq S \subsetneq V \text{ and } \uncutg{S} \le 2^{i+j}\}| \, ,
	\end{align*}
	and applying Theorem~\ref{th:counting} with $k \coloneqq 2^i$ and $s \coloneqq 2^j$ we conclude the claim.
\end{proof}
\noindent
Using the above claims we finally have
\begin{align*}
	& \Pr[\neg\mathcal{B} \text{ and } \exists \, \emptyset \neq S \subsetneq U \text{ s.t. } \mathcal{F}_{\uncutg{S}}] \\
	\text{by Claim~\ref{claim:calyisenough}} \quad  \quad \le & \Pr[\exists j \in [2\log n] \text{ s.t. } \exists Y \in \mathcal{Y}_{ij} \text{ s.t. } \neg\mathcal{B} \text{ and } \mathcal{F}_{Y}] \\
	\text{by Claim~\ref{claim:worstcaseproj}} \quad  \quad \le & \Pr[\exists j \in [2\log n] \text{ s.t. } \exists Z \in \{\Pi_U^i(Y): Y \in \mathcal{Y}_{ij}\} \text{ s.t. }  \neg\mathcal{B} \text{ and } \mathcal{F}_{Y_j^*(Z)}] \\
	\text{by Claim~\ref{claim:boundprobsingleset}} \quad  \quad \le & \sum_{j \in [2\log n]} \quad \sum_{ Z \in \{\Pi_U^i(Y) : \, Y \in \mathcal{Y}_{ij}\}} 2n^{-\frac{C}{768}2^j} \\
	\text{by Claim~\ref{claim:smallunionbound}} \quad  \quad \le & \sum_{j \in [2\log n]} n^{2 \cdot 2^{j}}\, 2n^{-\frac{C}{768}2^j} \\
	\le & \frac{1}{\poly(n)}
\end{align*}
\end{proof}
\noindent
Finally, we can conclude the main sparsification lemma \cref{lem:distrcutstreamadd} by a simple combination of \cref{lem:failsampling}, \cref{lem:levelbound}, and \cref{lem:spacesampling}.

\begin{proof}[Proof of \cref{lem:distrcutstreamadd}]
Recall that the $F_i$'s partition and cover $E$ entirely, and that $\neg\mathcal{B}$ happens with high probability by \cref{lem:failsampling}. Then, one can see that conditioning on $\neg \mathcal{B}$, summing over $i=0,\dots, \log n$, and applying \cref{lem:levelbound} gives \cref{lem:distrcutstreamadd}.
\end{proof}
\section{Testing expansion and finding sparse cuts in sparsifiers}
\label{sec:balsparse}
The key ingredient to our ED algorithms is that sparsifiers as per \cref{def:spars} serve as good proxies to the subgraphs of $G$ for any subroutine that either asserts expansion or outputs an approximate most-balanced sparse cut. We will in fact show that the result of running such subroutine on a boundary-linked subgraph of $G$ is almost the same as if we run it on the corresponding subgraph of a sparsifier.

Before we proceed further, we need to define what we demand exactly from the approximate most-balanced sparse cut subroutine. Loosely speaking, it is a procedure that takes as input a graph and a sparsity parameter, and either asserts that the graph is an expander, or outputs a sparse cut. In particular, the returned cut should be not too small compared to the largest cut among those certifying that the input is not an expander. For technical reasons, we ask that the algorithm also returns an approximation to the volume of the returned cut. Inspired by the definition of~\cite{streamexpdec}, we have the following definition, which captures what the subroutine should return.

\begin{definition}[Balanced sparse cut witness]
	\label{def:bscw}
	Let $\psi, \xi \in [0,1]$ and $\alpha,\lambda \ge 1$, and let $H=(U,E',w)$ be a weighted graph with self-loops. A $(\psi,\xi,\alpha,\lambda)$-balanced sparse cut witness of $H$ (for short, $(\psi,\alpha,\lambda,\xi)$-BSCW) is an element $\omega$ of $\{\bot\} \cup (\{R\,:\,\emptyset \neq R \subsetneq U\} \times \mathbb{R}_{>0})$ such that:
	\begin{enumerate}
		\item \label{item:expr} if $\omega = \bot$, then $H$ is a $\psi$-expander;
		\item if $\omega = (R,\nu)$, then:
		\begin{enumerate}
			\item \label{item:sparser} $\Phi_{H}(R) < \alpha \cdot \psi $,
			\item \label{item:balr} for every other cut $\emptyset \neq T \subsetneq U$ with $\Phi_H(T) <  \psi$ and $\vol_H(T) \le \vol_H(U \setminus T)$ one has $\vol_H(T) \le \lambda\cdot \vol_H(R)$,
			\item \label{item:smallsider} $\vol_H(R) \le (1+\xi) \vol_H(U \setminus R)$,
			\item \label{item:goodestimr} $ (1-\xi)  \cdot \unvolc{R}{U} \le \nu \le  (1+\xi)  \cdot \unvolc{R}{U}$.
		\end{enumerate}
	\end{enumerate}
\end{definition}
\noindent
A balanced sparse cut subroutine is then simply an algorithm returning a balanced sparse cut witness. For convenience of later discussions, we also have the following definition.

\begin{definition}[Balanced sparse cut algorithm]
	\label{def:bsca}
	Let $\alpha,\lambda \ge 1$. An algorithm $\bsc$ is an $(\alpha, \lambda)$-balanced sparse cut algorithm (for short, $(\alpha,\lambda)$-BSCA) if, taken as input a weighted graph with self-loops $H=(U,E',w)$ and a parameter $\psi \in (0,1)$, $\bsc(H,\psi)$ outputs a $(\psi,\alpha,\lambda,0)$-BSCW of $H$.
\end{definition}
\noindent
Ultimately, we will use either an exponential time brute force BSCA, or a fast BSCA with worse parameters, leading to our exponential and polynomial time ED algorithms, respectively.

Using these definitions, the main result of this section is the following: one can almost seamlessly run a BSCA on the sparsifier, without heavily deteriorating the quality of the BSCW hence obtained.
\begin{restatable}[Proxying lemma]{lemma}{reduction}
	\label{lem:reduction}
	Let $b, \phi, c,\delta \in (0,1)$ such that $c \leq 1/30$ is a constant, $\phi < b$, $b \le c$, $\delta \le c^2b/ \log n$, and  let $\bsc$ be an $(\alpha,\lambda)$-BSCA with $ \alpha \le \frac{1}{2b}$. Fix a cluster $U \subseteq V$ and also let $H=(V,E',w)$ such that $H[U] \approx_\delta G[U]$. Let $\omega$ be the output of $\bsc(H[U]^\circ, (1+\frac{1}{2\log n}) \cdot \phi)$. Then, $\omega$ is a  $(\phi,(1+\frac{1}{\log n})\alpha,(1+c)\lambda,c)$-BSCW of~$G[U]^\circ$.
\end{restatable}
\noindent
In \cref{subsec:boundlink} and \cref{subsec:preservingcuts}  we show properties that will let us compare cuts and volumes in $H[U]^\circ$ and $G[U]^\circ$. In \cref{subsec:proofproxying} we use these tools to prove \cref{lem:reduction}.

\subsection{Characterizing boundary-linked cuts}
\label{subsec:boundlink}
Structurally, we will divide the cuts of a subgraph $G[U]^\circ$ into two groups: those those that are boundary-linked and those that are not. This concept\footnote{The notion of boundary-linked cuts was already used in previous works, but not in the context of sparsification and streaming~\cite{boundlink,mincutjli}}., which we introduce next, is key to our goal: it captures the cuts that we are able to sparsify to within multiplicative error. Loosely speaking, a cut of $G[U]^\circ$ is boundary-linked if its local cut is a decent fraction of its global cut.

\begin{definition}[Boundary-linked cuts]
	\label{def:boundlinkcuts}
	Let $\beta \in (0,1)$ be a parameter, let $U \subseteq V$ be a cluster, and let $S \subseteq U$ be a cut. We say that $S$ is $\beta$-boundary-linked in $U$ if
	\begin{equation*}
		\uncut{S}{U} \ge \beta \cdot \min\{\uncutg{S}, \uncutg{(U \setminus S)}\} \, .
	\end{equation*}
\end{definition}
\noindent
We shall then group the cuts of $G[U]^\circ$ into boundary-linked ones and those that are not. As shown by the two lemmas below, each of these groups has a useful properties for our goal: for a correct set of parameters, boundary-linked cuts are very well approximated by a sparsifier, while those that are not boundary-linked must be sparse.

\begin{lemma}[Characterization of boundary-linked cuts]
	\label{lem:blinkcuts}
	Let $b, \phi, c,\delta \in (0,1)$ such that $c \leq 1/2$ is a constant, $\phi \le b$, $b \le c$, and $\delta \le c^2b/ \log n$.  Fix a cluster $U \subseteq V$ and also let $H=(V,E',w)$ such that $H[U] \approx_\delta G[U]$. Then, the following hold:
	\begin{enumerate}
		\item \label{charact:nonblink} for all non $b/2$-boundary-linked cuts $\emptyset \neq S \subsetneq U$ we have $\unspc{S}{U} < \phi$, and
		\item \label{charact:blink} for all $b/2$-boundary-linked cuts $S \subseteq U$ we have $(1-\frac{c}{\log n})\cdot \uncut{S}{U} \le \wcut{S}{U} \le 	(1+\frac{c}{\log n})\cdot \uncut{S}{U}$.
	\end{enumerate}
\end{lemma}
\begin{proof}
	To show~\eqref{charact:nonblink}, let $\emptyset \neq S \subsetneq U$ be a cut with $\unspc{S}{U} \ge \phi$. By definition this means that $\uncut{S}{U} \ge \phi \min \{\unvolc{S}{U},\unvolc{U \setminus S}{S}\}$. Since $\uncut{S}{U}=\uncut{(U\setminus S)}{U}$, we can in fact assume without loss of generality that $\uncut{S}{U} \ge \phi \unvolc{S}{U}$. Recalling how these volumes are defined, we have that every vertex in $u \in S$ has a contribution of at least $b\phi \cdot \unbor{\{u\}}{U}$ to $\unvolc{S}{U}$, so $\unvolc{S}{U} \ge b/\phi \cdot \unbor{S}{U}$. Hence, $\uncut{S}{U} \ge \phi \unvolc{S}{U}$ yields $\uncut{S}{U} \ge b \cdot \unbor{S}{U}$. Finally, we observe that the global cut of $S$, i.e.~$\uncutg{S}$, can be decomposed as the sum of $\uncut{S}{U}$ and $\unbor{S}{U}$. By virtue of this decomposition, from $\uncut{S}{U} \ge b \cdot \unbor{S}{U}$ we get $\uncut{S}{U} \ge b \cdot (\uncutg{S}-\uncut{S}{U})$. Rearranging the terms, and since we assume that~$b$ is bounded by $c \le 1/2$, we conclude $\uncut{S}{U} \ge b/2 \cdot \uncutg{S}$, i.e. $S$ is $b/2$-boundary-linked.
	
	To show~\eqref{charact:blink}, from definition of boundary-linked we know $\uncut{S}{U} \ge b/2 \cdot \min\{\uncutg{S}, \uncutg{(U \setminus S)}\} $, while from the assumption that $H[U] \approx_\delta G[U]$ we know $|\uncut{S}{U} - \wcut{S}{U} | \le \delta \cdot \uncutg{S}$. Again, without loss of generality we can assume $\uncut{S}{U} \ge b/2 \cdot \uncutg{S}$, since $\uncut{S}{U}=\uncut{(U\setminus S)}{U}$ and $\wcut{S}{U}=\wcut{(U\setminus S)}{U}$. Therefore, the sparsification error $|\uncut{S}{U} - \wcut{S}{U} | $ can actually be bounded by $\delta \cdot \uncutg{S} \le c^2b/ \log n \cdot 2/b \cdot \uncut{S}{U} \le  c/\log n \cdot \uncut{S}{U}$ (since $c\le 1/2$).
\end{proof}

\subsection{Preserving volumes and sparsities}
\label{subsec:preservingcuts}
 \cref{lem:blinkcuts} says that for any cut $S$ in $U$, either $S$ is sparse, or the sparsifier gives a multiplicative approximation to its local cut. Intuitively, for the latter cuts, one would expect the sparsifier to preserve their sparsity as well. However, this is not trivially true, because the sparsity depends on both the local cut and the volume in the boundary-linked subgraph: while we have a multiplicative preservation of boundary-linked cuts, we may not have the same guarantee for their volume. To circumvent this issue, we consider the case of sparse and expanding cuts separately: the former have in fact their volumes preserved within a relative error, while for the latter we get that their volume in the sparsifier is still not too large. We combine these observations by treating separately boundary-linked and non boundary-linked cuts, and conclude that the distinction between sparse and expanding is very accurately preserved in a sparsifier. We begin with analysing the volume of sparse cuts in the sparsifier.

\begin{lemma}[Multiplicative error for volumes of sparse cuts]
	\label{lem:volumesparsecuts}
	Let $b, \phi,\hat{\phi}, c,\delta \in (0,1)$ such that $c \leq 1/5$ is a constant, $\phi < b$, $\delta \le c^2 / \log n$, and \smash{$\hat{\phi} \le \phi/b$}. Fix a cluster $U \subseteq V$ and also let $H=(V,E',w)$ such that $H[U] \approx_\delta G[U]$. Then, for all cuts $\emptyset \neq S \subsetneq U$ with \smash{$\unspc{S}{U} < \hat{\phi}$} we have
	\begin{equation*}
		\left(1-\frac{c}{\log n}\right) \cdot \unvolc{S}{U} \le \wvolc{S}{U} \le \left(1+\frac{c}{\log n}\right) \cdot \unvolc{S}{U} \, .
	\end{equation*}
\end{lemma}
\begin{proof}
	By definition of $\wvolc{S}{U}$, the error in estimating the volume can be written as
	\begin{equation}
		\left|\wvolc{S}{U}- \unvolc{S}{U}\right| \le \left|\wvol{S}- \vol_U(S)\right| + \frac{b-\phi}{\phi}\left|\wbor{S}{U} - \unbor{S}{U}\right| \, . \label{eq:triangineqselfloops}
	\end{equation}
	The first term is easily upper-bounded by $\delta \cdot \vol(S)$, because of the global cut preservation property of $H[U] \approx_\delta G[U]$ (see \cref{def:spars}) applied to singleton cuts. For the second term, we can further split it as
	\begin{equation}
		\left|\wbor{S}{U} - \unbor{S}{U}\right| \le \left|\wcutg{S}-\uncutg{S}\right|+\left|\wcut{S}{U}-\uncut{S}{U}\right| \, , \label{eq:trianginequality}
	\end{equation}
We now upper-bound the right hand side using again the cut preservation properties from \cref{def:spars}: we employ the multiplicative error for the global cut $\uncutg{S}$, and the additive error for $\uncut{S}{U}$. Combining these with the assumption that $S$ is sparse, one gets
	\begin{align}
		\Big|\wcutg{S} - \uncutg{S}\Big| + \Big|\wcut{S}{U}-\uncut{S}{U}\Big| 
		& \le  2\delta\cdot \uncutg{S} \\
		& =   2\delta \cdot\big(\uncut{S}{U} + \unbor{S}{U}\big) \\
		  \text{ since $\unspc{S}{U} < \hat{\phi}$ } \quad & <  2\delta \cdot \hat{\phi} \cdot \vol_U^\circ(S) + 2\delta \cdot \unbor{S}{U}\\
	 \text{ since $\hat{\phi} \le \phi/b$ and $\vol_U^{\circ}(S) \ge \frac{\phi}{b}\unbor{S}{U}$ } \quad & \le 2\delta \cdot \frac{\phi}{b} \cdot \vol_U^\circ(S) + 2\delta \cdot \frac{\phi}{b} \cdot \vol_U^{\circ}(S)\\
		& \le 4\delta\cdot\frac{\phi}{b-\phi} \cdot \vol_U^{\circ}(S) \, , \label{eq:sumgloballocalerror}
	\end{align}
We now backtrack and plug~\eqref{eq:sumgloballocalerror} into~\eqref{eq:trianginequality}, and then~\eqref{eq:trianginequality} into~\eqref{eq:triangineqselfloops} to conclude the claim:
	\begin{align*}
		\left|\wvolc{S}{U}- \unvolc{S}{U}\right| & \le \left|\wvol{S}- \vol_U(S)\right| + \frac{b-\phi}{\phi}\left|\wbor{S}{U} - \unbor{S}{U}\right| \\
		& \le \delta \cdot \vol(S) + 4\delta \cdot \vol_U^{\circ}(S) \\
		& \le 5\delta \cdot \vol_U^{\circ}(S) \\
		& \le \frac{5c^2}{\log n}\cdot \unvolc{S}{U} \\
		& \le \frac{c}{\log n}\cdot \unvolc{S}{U} \, ,
	\end{align*}
since $c \le 1/5$.
\end{proof}
\noindent
Using this fact together with the cut preservation properties of a cluster sparsifier, we can already conclude that there cannot be cuts that look ``very'' expanding in $H[U]^\circ$ but are actually sparse in~$G[U]^\circ$. In particular, if the sparsity of a cut in the sparsifier is above $\phi$ by a factor that dominates the error of the sparsifier, then its sparsity must be at least $\phi$ in the original graph. We show this in \cref{lem:iffsparsity}.
\begin{lemma}[Very expanding cuts in a sparsifier are indeed expanding]
	\label{lem:iffsparsity}
	Let $b, \phi,\hat{\phi}, c,\delta \in (0,1)$ such that $c \leq 1/30$ is a constant, $\phi < b$, $b \le c$, $\delta \le c^2 b/ \log n$, and \smash{$\hat{\phi} \ge (1+\frac{1}{3\log n})\cdot \phi$}.  Fix a cluster $U \subseteq V$ and also let $H=(V,E',w)$ such that $H[U] \approx_\delta G[U]$. Then, for all cuts $\emptyset \neq S \subsetneq U$ with \smash{$\wspc{S}{U} \ge \hat{\phi}$} we have $\unspc{S}{U} \ge \phi$.
\end{lemma}
\begin{proof}
	For the sake of a contradiction, let $\emptyset \neq S \subsetneq U$ be a cut with \smash{$\wspc{S}{U} \ge \hat{\phi}$} such that $\unspc{S}{U} < \phi$.
	
	Consider first the case that $S$ is $b/2$-boundary-linked. Then we know from $H[U] \approx_\delta G[U]$ (see \cref{def:spars}) that $\wcut{S}{U} = (1 \pm c/\log n) \uncut{S}{U}$. Also, as we are assuming $\unspc{S}{U} < \phi$, from \cref{lem:volumesparsecuts} (our parameters verify stronger conditions than what is demanded by the lemma, so it applies) we know $\wvolc{S}{U}= (1 \pm c/\log n)\unvolc{S}{U}$. These two together give $\wspc{S}{U} < (1+10c/\log n)\phi < \hat{\phi}$ for small enough $c$, a contradiction.
	
	Next consider the case that $S$ is not $b/2$-boundary-linked. Then we do not have a multiplicative approximation to $\wcut{S}{U}$. However, $S$ is $\phi$-sparse in $G[U]^\circ$ so we still have $\wvolc{S}{U}= (1 \pm c/\log n)\unvolc{S}{U}$ and $\wvolc{U \setminus S}{U}= (1 \pm c/\log n)\unvolc{U \setminus S}{U}$ by Lemma~\ref{lem:volumesparsecuts}. Therefore, because $\hat{\phi} \ge (1+\frac{1}{3\log n}) \phi$ we get
	\begin{align*}
		\wcut{S}{U} & \ge \left(1-\frac{c}{\log n}\right) \left(1+\frac{1}{3\log n}\right)\cdot \phi \cdot \min\{\unvolc{S}{U}, \unvolc{U \setminus S}{U}\} \\
		& \ge \left(1+\frac{1}{10\log n}\right)\cdot \phi \cdot \min\{\unvolc{S}{U}, \unvolc{U \setminus S}{U}\} \, ,
	\end{align*}
	and at the same time the additive error approximation from \cref{def:spars} is enough to yield
	\begin{equation*}
		\wcut{S}{U} \le \uncut{S}{U} + \delta \cdot \uncutg{S} < \phi \cdot \min\{\unvolc{S}{U}, \unvolc{U \setminus S}{U}\}  + \delta\cdot \uncutg{S} \, .
	\end{equation*}
	Therefore,
	\begin{align*}
		\uncutg{S} & \ge \frac{1}{\delta \cdot10 \log n} \cdot \phi \cdot \min\{\unvolc{S}{U}, \unvolc{U \setminus S}{U}\} \\
		\text{ since $\delta \le c^2 b/ \log n$ } \quad & \ge \frac{1}{10c^2} \cdot \min \{\unbor{S}{U}, \unbor{U \setminus S}{U}\}\, .
	\end{align*}
	Now let us assume $\unbor{S}{U} \le \unbor{U \setminus S}{U}$ without loss of generality. Then the above is essentially saying that $\uncutg{S}$ is much larger than $\unbor{S}{U}$.	Recalling now that we can decompose $\uncutg{S}=\uncut{S}{U}+\unbor{S}{U}$, the above lower bound is actually saying that much of the advantage that~$\uncutg{S}$ has over~$\unbor{S}{U}$ must come from $\uncut{S}{U}$. More precisely, since $c \le 1/30$ we get
	\begin{equation*}
		\uncut{S}{U} \ge\frac{1}{c} \cdot \uncutg{S} > \frac{b}{2} \cdot \uncutg{S} \, ,
	\end{equation*}
	which contradicts the assumption that $S$ is not $b/2$-boundary-linked.
\end{proof}
\noindent
Next, we handle the case of expanding cuts, and show that their volume does not get overshot too much. Even though this is a much weaker guarantee than the one we obtained in Lemma~\ref{lem:volumesparsecuts} for sparse cuts, it will  still be enough for our purposes: this will later be useful in concluding that their sparsity has not dropped by a lot.

\begin{lemma}[Upper bound for volumes of expanding cuts]
	\label{lem:testexpandingcuts}
	Let $b, \phi,\hat{\phi}, c,\delta \in (0,1)$ such that $c \leq 1/2$ is a constant, $\phi \le b$, $\delta \le c^2 / \log n$, and \smash{$\hat{\phi} \le \phi/b$}. Fix a cluster $U \subseteq V$ and also let $H=(V,E',w)$ such that $H[U] \approx_\delta G[U]$. Then, for all cuts $\emptyset \neq S \subsetneq U$ with \smash{$\Phi_U^\circ(S) \ge \hat{\phi}$} we have
	\begin{equation*}
		\min \{\wvolc{S}{U},\wvolc{U \setminus S}{U}\} \le \frac{1+\frac{c}{\log n}}{\hat{\phi}} \cdot \uncut{S}{U}\, .
	\end{equation*}
\end{lemma}
\begin{proof}
	Let $S^* \in \{S, U \setminus S\}$ be the side of the cut achieving the minimum $\min\{\unvolc{S}{U}, \unvolc{U \setminus S}{U}\}$. Even though $\wbor{S^*}{U}$ may not be a very good approximation of $\unbor{S^*}{U}$, we can upper-bound it by decomposing  $\wcutg{S^*}=\wcut{S}{U}+\wbor{S}{U}$ and using the local cut preservation property of $H[U] \approx_\delta G[U]$ (see \cref{def:spars}): we get
	\begin{align*}
		\wbor{S^*}{U} & = \wcutg{S^*} - \wcut{S^*}{U} \le  \unbor{S^*}{U}+\delta \cdot \uncutg{S^*} \, .
	\end{align*}
Using the above bound, together with the global approximation from \cref{def:spars} applied to singleton cuts, we can upper-bound $\wvolc{S^*}{U}$ as
	\begin{align*}
		\wvolc{S^*}{U} & = \wvol{S}+\frac{b-\phi}{\phi}\wbor{S}{U} \\
		& \le  (1+\delta) \cdot \vol(S^*)+\frac{b-\phi}{\phi} \unbor{S^*}{U} + \frac{b-\phi}{\phi}  \cdot \delta \cdot \uncutg{S^*}\\
		\text{ since $\uncutg{S^*}=\uncut{S}{U}+\unbor{S}{U}$} \quad & \le (1+\delta) \cdot \vol(S^*)+ (1+\delta) \cdot\frac{b-\phi}{\phi}  \cdot \unbor{S^*}{U}+ \frac{b-\phi}{\phi}  \cdot \delta \cdot \uncut{S^*}{U}\\
		\text{ by definition of $\unvolc{S}{U}$} \quad & \le (1+\delta) \cdot \vol_U^{\circ}(S^*)+ \frac{b-\phi}{\phi}  \cdot \delta \cdot \uncut{S^*}{U}\\
		\text{ since $\hat{\phi} \le \frac{\phi}{b}$} \quad & \le (1+\delta) \cdot  \unvolc{S^*}{U} + \delta \cdot  \frac{1}{\hat{\phi}} \cdot \uncut{S^*}{U} \, .
	\end{align*}
By our assumption that $\Phi_U^{b/\phi}(S) \ge \hat{\phi}$, we also know $\vol_U^{b/\phi}(S^*) \le \frac{1}{\hat{\phi}} \uncut{S}{U}$. We then conclude
\begin{align*}
	\wvolc{S^*}{U}
	& \le (1+2\delta) \cdot \frac{1}{\hat{\phi}} \uncut{S}{U} \le \left(1+\frac{c}{\log n}\right) \frac{1}{\hat{\phi}} \uncut{S}{U}\, ,
\end{align*}
since $c \le 1/2$.
\end{proof}
\noindent
Having a bound on the volumes of expanding cuts, allows to bound the sparsity of those cuts in the sparsifier (since all such cuts are boundary-linked, and hence the local cut is well preserved). This is almost the reverse of \cref{lem:iffsparsity}.

\begin{lemma}[Expanding cuts are almost expanding in the sparsifier]
	\label{lem:boundlinkexpsparse}
Let $b, \phi,\hat{\phi}, c,\delta \in (0,1)$ such that $c \leq 1/6$ is a constant, $\phi \le b$, $b \le c$, $\delta \le c^2 b/ \log n$, and \smash{$\phi \le \hat{\phi} \le \phi/b$}.  Fix a cluster $U \subseteq V$ and also let $H=(V,E',w)$ such that $H[U] \approx_\delta G[U]$. Then, for all cuts $\emptyset \neq S \subsetneq U$ with \smash{$\unspc{S}{U} \ge \hat{\phi}$} we have $\wspc{S}{U} \geq (1 - \frac{1}{3\log n}) \hat{\phi}$.
\end{lemma}
\begin{proof}
	First note that our parameter regime is strictly stronger than that of \cref{lem:blinkcuts} and \cref{lem:testexpandingcuts}. Hence, by \cref{lem:blinkcuts}, we know that $S$ is $b/2$-boundary-linked, since it is $\hat{\phi}$-expanding and $\hat{\phi} \ge \phi$. \cref{lem:blinkcuts} then also guarantees that $\wcut{S}{U} \ge (1-\frac{c}{\log n}) \uncut{S}{U}$. On the other hand, \cref{lem:testexpandingcuts} gives \smash{$	\min \{\wvolc{S}{U},\wvolc{U \setminus S}{U}\} \le (1+\frac{c}{\log n})/{\hat{\phi}} \cdot \uncut{S}{U}$}. Hence,
	\begin{equation*}
		\min \{\wvolc{S}{U},\wvolc{U \setminus S}{U}\} \le \frac{1+\frac{c}{\log n}}{1-\frac{c}{\log n}}\frac{1}{\hat{\phi}} \cdot \wcut{S}{U} \leq  \frac{1}{1-\frac{1}{3\log n}}\cdot \frac{1}{\hat{\phi}} \cdot \wcut{S}{U}  \, ,
	\end{equation*}
since $c \le 1/6$ is small enough.
\end{proof}

\subsection{Proving the proxying lemma}
\label{subsec:proofproxying}
From the previous section, we have that the distinction between sparse and expanding cuts of $G[U]^\circ$ is very accurately preserved by $H[U]^\circ$, and also the volume of sparse cuts of $G[U]^\circ$ is approximated to within a small relative error in $H[U]^\circ$. As these are exactly the quantities that a BSCA is interested in, we can then prove the main lemma of the section, restated here for convenience of the reader.

\reduction*

\begin{proof}[Proof]
	By definition of $\bsc$, we know that $\omega$ is a $((1+\frac{1}{2\log n}) \cdot \phi, \alpha,\lambda,0)$-BSCW of~$H[U]^\circ$, i.e.
	\begin{enumerate}
		\item \label{itdef:exp2} if $\omega=\bot$, then $H[U]^\circ$ is a $(1+\frac{1}{2\log n}) \cdot \phi$-expander;
		\item if $\omega=(R,\nu)$, then:
		\begin{enumerate}
			\item \label{itdef:sparse2} $\wspc{R}{U} < (1+\frac{1}{2\log n}) \cdot \alpha \cdot \phi $,
			\item \label{itdef:smallside2} $\wvolc{R}{U} \le \wvolc{U \setminus R}{U} $,
			\item \label{itdef:bal2} for every other cut $\emptyset \neq T \subsetneq U$ with $\wspc{R}{U} <(1+\frac{1}{2\log n}) \cdot \phi$ and $\wvolc{T}{U} \le \wvolc{U \setminus T}{U}$ one has $\wvolc{T}{U} \le B \wvolc{R}{U} $:
			\item \label{itdef:goodestim2} $\nu =\wvolc{R}{U}$.
		\end{enumerate}
	\end{enumerate}
Our goal is to translate these properties to $G[U]^\circ$ by showing that $\omega$ is a  $(\phi,(1+\frac{1}{\log n})\alpha,(1+c)\lambda,c)$-BSCW of~$G[U]^\circ$, i.e.
\begin{enumerate}
	\item \label{itdef:exp3} if $\omega=\bot$, then $G[U]^\circ$ is a $ \phi$-expander;
	\item if $\omega=(R,\nu)$, then:
	\begin{enumerate}
		\item \label{itdef:sparse3} $\unspc{R}{U} < (1+\frac{1}{\log n}) \cdot\alpha \cdot  \phi $,
		\item \label{itdef:smallside3} $\unvolc{R}{U} \le (1+c)\unvolc{U \setminus R}{U} $,
		\item \label{itdef:bal3} for every other cut $\emptyset \neq T \subsetneq U$ with $\unspc{R}{U} <\phi$ and $\unvolc{T}{U} \le \unvolc{U \setminus T}{U}$ one has $\unvolc{T}{U} \le (1+c)B \unvolc{R}{U} $:
		\item \label{itdef:goodestim3} $(1-c)\cdot\unvolc{R}{U}\le \nu \le (1+c)\cdot\unvolc{R}{U}$.
	\end{enumerate}
\end{enumerate}

	\paragraph{Correctness of the case $\omega = \bot$.} The idea is that if the cluster $U$ is non-trivially more than a $\phi$-expander in the sparsifier, then it should be at least a $\phi$-expander in $G$. Let us say that $\omega=\bot$, i.e. $\wspc{S}{U} \ge  (1+\frac{1}{2\log n})\cdot \phi $ for all $\emptyset \neq S \subsetneq U$. Let \smash{$\hat{\phi}= (1+\frac{1}{2\log n}) \cdot \phi$}, so that we are in the parameter regime of \cref{lem:iffsparsity}. Then, this lemma ensures that $\unspc{S}{U} \ge \phi$ for all $\emptyset \neq S \subsetneq U$. In other words, $G[U]^\circ$ is a $\phi$-expander.
	
	\paragraph{Correctness of the case $\omega=(R,\nu)$.}
Assume $\omega=(R,\nu)$. We need to show the four properties of a BSCW for~$G[U]^\circ$.
	\begin{itemize}
		\item \textit{Property~\eqref{itdef:sparse3}}. Intuitively, $R$ cannot be much more expanding in $G$ than what it looks like in $H$, so if it is sparse in $H$ it is basically just as sparse in $G$. To prove this formally, it is convenient to distinguish the case of this cut being boundary-linked or not.
		
		If $R$ is not $b/2$-boundary-linked, then \cref{lem:blinkcuts} implies $\unspc{R}{U} < \phi$. The property is  then already proved in this case. 
		
		If it is $b/2$-boundary-linked, suppose for the sake of a contradiction that $R$ has {$\unspc{R}{U} \ge (1+{1}/{\log n})\alpha \cdot \phi$}.  Let \smash{$\hat{\phi}=(1+\frac{1}{\log n})\alpha \cdot \phi$}. This value of \smash{$\hat{\phi}$} fulfils the condition that \smash{$\hat{\phi} \le 1/b$} from \cref{lem:boundlinkexpsparse}, since $\alpha\le 1/(2b)$ by assumption. We then meet the requirements to apply \cref{lem:boundlinkexpsparse} (all other parameters also meet its preconditions, as our parameter regime for $b,\phi,c,\phi$ is no weaker than its), so
		
		\[
			\wspc{R}{U} \geq \left(1-\frac{1}{3\log n}\right)\cdot \left(1+\frac{1}{\log n}\right)\cdot\alpha \cdot \phi \geq \left(1+\frac{1}{2\log n}\right) \alpha \cdot \phi,
		\]
		which contradicts the definition of $\bsc$. We have then showed property~\eqref{itdef:sparse3}.

		\item \textit{Properties~\eqref{itdef:smallside3} and~\eqref{itdef:goodestim3}}. From the previous point we know that $\unspc{R}{U} <(1+\frac{1}{\log n})\alpha \cdot \phi$. Since our parameter regime is only stronger than that of \cref{lem:volumesparsecuts} if we set \smash{$\hat{\phi}=(1+\frac{1}{\log n})\alpha \cdot \phi$}, we get $\wvolc{R}{U}= (1 \pm c/\log n)\unvolc{R}{U}$ and $\wvolc{U \setminus R}{U}= (1 \pm c/\log n)\unvolc{U \setminus R}{U}$. Recalling that $\wvolc{R}{U} \le \wvolc{U\setminus R}{U}$ (see~\eqref{itdef:smallside2} for $H[U]^\circ$), this yields
		\begin{equation*}
			\unvolc{R}{U} \le \frac{1}{1-\frac{c}{\log n}} \wvolc{R}{U} \le  \frac{1}{1-\frac{c}{\log n}}\wvolc{U\setminus R}{U} \le \frac{1+\frac{c}{\log n}}{1-\frac{c}{\log n}}\unvolc{U\setminus R}{U} \, ,
		\end{equation*}
		so $\unvolc{R}{U} \le (1+c)\unvolc{U\setminus R}{U}$. Recalling that $\nu = \wvolc{R}{U}$ (see~\eqref{itdef:goodestim2} for $H[U]^\circ$), one also has
		\begin{equation*}
			\left|\nu-\unvolc{R}{U}\right|=\left|\wvolc{R}{U}-\unvolc{R}{U}\right|\le c/\log n \cdot \unvolc{R}{U} \le c \unvolc{R}{U} \, .
		\end{equation*}

		\item \textit{Property~\eqref{itdef:bal3}}. This property aims to bound the volume of $R$ against the volume of every other $\phi$-sparse cut. We want to exploit the bound that $\bsc$ gives for these cuts in $H$, i.e. that for every cut $\emptyset \neq T \subsetneq U$ with $\wspc{R}{U} <  (1+\frac{1}{2\log n})  \phi$ and $\wvolc{T}{U} \le \wvolc{U \setminus T}{U}$ one has $\wvolc{T}{U} \le \lambda \wvolc{R}{U} $ (see~\eqref{itdef:bal2} for $H[U]^\circ$). Now one can see that if we prove all such $T$ to be also $\phi$-sparse in $G[U]^\circ$, then we could conclude simply using the fact that the volumes of $H[U]^\circ$ are good proxies for those in $G[U]^\circ$. We now implement this strategy.
		
		By virtue of property~\eqref{itdef:sparse3}, we have that $\unspc{R}{U}<(1+\frac{1}{\log n})\alpha \cdot \phi$, and from an application of \cref{lem:volumesparsecuts} we have $\wvolc{R}{U}= (1 \pm c/\log n)\unvolc{R}{U}$ (again, the preconditions of the lemma are met if we take \smash{$\hat{\phi}=(1+\frac{1}{\log n})\alpha \cdot \phi$}).
		
		Consider now any other cut $\emptyset \neq T \subsetneq U$ with $\unspc{T}{U} < \phi$ and $\unvolc{T}{U} \le \unvolc{U \setminus T}{U}$. Note that for any such $T$ we have $\wvolc{T}{U}= (1 \pm c/\log n)\unvolc{T}{U}$ and $\wvolc{U \setminus T}{U}= (1 \pm c/\log n)\unvolc{U \setminus T}{U}$, again by applying \cref{lem:volumesparsecuts} (setting $\hat{\phi} =\phi \le \phi/b$, and all other parameters satisfy the needed preconditions).
		
		Moreover, any such $T$ is also $((1+\frac{1}{2\log n})\ \cdot \phi)$-sparse in $H[U]^\circ$, by applying Lemma~\ref{lem:iffsparsity} with parameter \smash{$\hat{\phi}=(1+\frac{1}{2\log n}) \phi \ge (1+\frac{1}{3\log n}) \phi$}. From~\eqref{itdef:bal2} for $H[U]^\circ$, we then know
		\begin{equation*}
			\min\{\wvolc{T}{U},\wvolc{U \setminus T}{U}\} \le \lambda \wvolc{R}{U} \, .
		\end{equation*}
		Recall that $\wvolc{R}{U}$, $\wvolc{T}{U}$, $\wvolc{U \setminus T}{U}$ all are within a $(1 \pm c/\log n)$ multiplicative error of $\unvolc{R}{U}$, $\unvolc{T}{U}$, $\unvolc{U \setminus T}{U}$ respectively. Hence, $\min\{\wvolc{T}{U},\wvolc{U \setminus T}{U}\} \le \lambda \wvolc{R}{U}$ translates to
		\begin{equation*}
			\unvolc{T}{U} = \min\{\unvolc{T}{U},\unvolc{U \setminus T}{U}\} \le (1+10c/\log n) \lambda\cdot \unvolc{R}{U} \le (1+c) \lambda\cdot \unvolc{R}{U} \, .
		\end{equation*}
	\end{itemize}
	\noindent
	We have then proved~\eqref{itdef:bal3}.
	\end{proof}
\section{Space efficient recursive partitioning}
\label{sec:algexpdec}

\noindent
The main goal of this section is to show Theorems~\ref{th:main} and~\ref{th:mainpoly}, restated here for convenience.

\mainbld*
\mainbldpoly*
\noindent
The only difference between the algorithms of Theorem~\ref{th:main} and~\ref{th:mainpoly} lies in the balanced sparse cut subroutine that they employ. Therefore, we present and analyse them in a unified manner by deferring the choice of the balanced sparse cut algorithm to later. Using the results from \cref{sec:spars} and \cref{sec:balsparse} (i.e. \cref{lem:distrcutstreamadd} and \cref{lem:reduction} respectively), we prove the following.

\begin{restatable}[Unified algorithm]{lemma}{unified}
	\label{lem:tech}
	Let $G=(V,E)$ be a graph given in a dynamic stream, let $b \in (0,1)$ be a parameter such that $b \le {1}/{\log n}$. Let $\bsc$ be an $(\alpha,\lambda)$-BSCA where \smash{$\alpha \le \frac{1}{b\log n}$} and \smash{$\lambda= O(1)$}. Also let $k$ be an integer such that \smash{$ k \ge {\log {(n^5 \alpha )} }/{ \log  {(b^{-1/2}}/{\lambda})} $}, and $k \le \log n$. Then there is an algorithm that maintains a linear sketch of $G$ in $\otil( n/b^3) \cdot \alpha^{O(k)}$ space. For all $\epsilon \in [n^{-2},b\log n]$, the algorithm decodes the sketch using $\bsc$ to compute, with high probability, a $(b,\epsilon,\phi,\gamma)$-BLD of $G$ for
	\begin{equation*}
		\phi = \Omega\left(\frac{\epsilon}{\alpha \log n}\right) \quad \text{ and } \quad \gamma = O(\alpha^{k+1}) \, .
	\end{equation*}
	Moreover, the decoding runs in $\otil( n/b^3) \cdot \alpha^{O(k)}+S(n,\otil( n/b^2) \cdot \alpha^{O(k)})$ space and  {$T(n,\otil( n/b^2) \cdot \alpha^{O(k)}) \cdot \poly(n) $} time, where $S(p,q)$ and $T(p,q)$, denote the space and time complexity of $\bsc$ on a graph with $p$ vertices and $q$ edges respectively. Furthermore, the decoding only makes calls to $\bsc(\cdot, \psi)$ with sparsity parameters $\psi \in (0,1)$ that satisfy $\psi \le \frac{1}{10\alpha}$.
\end{restatable}

\noindent
From Lemma~\ref{lem:tech} one readily obtains the theorems by either using a brute force $(1,1)$-BSCA or a polynomial time $(O(\log^3 n), O(1))$-BSCA.

\begin{proof}[Proof of Theorem~\ref{th:main}]
	It is easy to get an exponential time $(1,1)$-BSCA as per Definition~\ref{def:bsca}. Given a graph and a parameter $\psi \in (0,1)$, one can do this by brute-forcing all possible cuts in the input graph: either return $\omega=\bot$ if all the cuts are $\psi$-expanding, or return $\omega=(R,\nu)$ where~$R$ is the $\psi$-sparse cut with maximum volume, and $\nu$ is its volume. The space requirement of this algorithm is linear in the size of its input.
	
	We then have parameters $\alpha=1,\lambda=1$, and we set $k=\log n$.  Hence, for any $b \le 1/\log^2 n$ and $\epsilon \le b \log n$, we meet the preconditions  of  \cref{lem:tech}: $b \le 1/\log n$, $\alpha = 1 \le \frac{1}{b\log n}$, $\lambda= O(1)$, $k \le \log n$, and
	\begin{equation*}
		\frac{\log {n^5\alpha }}{\log  \frac{1}{\lambda \sqrt{b}}} \le \frac{O(\log n)}{\Omega(\log \log n)} \le k \, ,
	\end{equation*}
	where the first inequality follows because $\alpha=\lambda=1$, and $b \le 1/\log^2 n$.
	Therefore, by \cref{lem:tech} there is a dynamic stream algorithm that with high probability outputs a $(b, \epsilon, \phi, \gamma)$-BLD of $G$ where
	\begin{equation*}
		\phi = \Omega\left(\frac{\epsilon}{\alpha \log n}\right) = \Omega\left(\frac{\epsilon}{ \log n}\right) \quad \text{ and } \quad \gamma = O(\alpha^{k+1})= O(1) \, .
	\end{equation*}
	Its space complexity is
	\begin{equation*}
		\otil(k \cdot n/b^3) \cdot \alpha^{O(k)} = \otil(n/b^3) \, ,
	\end{equation*}
	since for the brute force algorithm we have $S(p,q)=\otil(q)$. The running time is $2^{O(n)}$ since the brute force algorithm has $T(p,q)=\otil(q)\cdot 2^p$.
\end{proof}
\begin{proof}[Proof of Theorem~\ref{th:mainpoly}]
	For the polynomial time version, we use the following BSCA.
	\begin{restatable}[\cite{offlineexpdec, liexpdec,leightonrao}]{theorem}{fastbalsparse}
		\label{th:fastbalsparse}
		For $b \in (0,1)$, let $p \le n$, \smash{$q \le   \min\{p^2, n/b^2 \cdot \log^{O(\frac{\log n}{\log 1/b})}n \}$} and $W \le \poly(n)$ be integers.
		For $p$-vertex $q$-edge input graphs with weighted edges and weighted self-loops, and total volume bounded by $W$, there is a $(c_{\sw} \log^3 n,C_{\sw})$-BSCA that runs in $\poly(n)$ time and \smash{$ n/b^2 \cdot \log^{O(\frac{\log n}{\log 1/b})}n $} bits of space, uses \smash{ $n/b^2 \cdot \log^{O(\frac{\log n}{\log 1/b})}n $} random bits, and works with probability \smash{$1-1/\poly(n)$} for every input sparsity parameter \smash{$0 <\psi \le \frac{1}{10 c_{\sw} \log^3 n}$},  where \smash{$c_{\sw},C_{\sw} \ge 1$} are absolute constants.
	\end{restatable}
	\noindent
	Hereafter let us set $\alpha=c_{\sw} \log^3 n,\lambda=C_{\sw}$ and \smash{$ k = \frac{\log {n^5 \alpha } }{ \log  {b^{-1/2}}/{\lambda}} $}. We want to apply \cref{lem:tech} in this parameter regime with \cref{th:fastbalsparse} as $\bsc$.
	
	One can see that the decoding of \cref{lem:tech} (i.e. \Cref{alg:expdec,alg:trim}) calls $\bsc$ on disjoint portions of \smash{$b^{-1} \cdot O(k \log n)=\otil(1/b) $} many sparsifiers (see the proof of \cref{lem:tech}). Then, because each sparsifier has at most \smash{$ \otil( n/b^2) \cdot \alpha^{O(k)}=\otil(n/b^2)\cdot \log^{O(\frac{\log n}{\log 1/b})}n $} edges (see the proof of \cref{lem:tech}), we only need  \smash{$\otil(n/b^3)\cdot \log^{O(\frac{\log n}{\log 1/b})}n$} random bits in order to avoid adaptivity and at the same time ensure that all the calls work correctly with probability $1-1/\poly(n)$. As this fits into our claimed space budget, we can then assume that this BSCA always outputs correctly. From \cref{lem:tech} we also know that the algorithm never calls $\bsc$ with sparsity parameters larger than $\frac{1}{10\alpha}$, so our setting matches the requirement of \cref{th:fastbalsparse}. We defer the proof of \cref{th:fastbalsparse}  to \cref{apndx:selfloops}.
	 
	We now verify the rest of the preconditions of \cref{lem:tech}. For any $b \le 1/\log^5 n$ and $\epsilon \le b \log n$, have that $\lambda=C_{\sw} = O(1)$ and
	\begin{equation*}
		\alpha=c_{\sw} \log^3 n \le \log^4 n \le \frac{1}{b \log n} \, .
	\end{equation*}
	We also have that $k = O(\log n /\log\log n) \le \log n$, since $\alpha = O(\log^3 n)$, $\lambda = O(1)$, and $b \le 1/\log^5 n$. Therefore, the lemma gives a dynamic stream algorithm that with high probability outputs a $(b, \epsilon, \phi, \gamma)$-BLD of $G$ where
	\begin{equation*}
		\phi = \Omega\left(\frac{\epsilon}{\alpha \log n}\right) = \Omega\left(\frac{\epsilon}{ \log^4 n}\right) \quad \text{ and } \quad \gamma = O(\alpha^{k+1})= \log^{O(\frac{\log n}{\log 1/b})}n \, .
	\end{equation*}
	Since the BSCA from \cref{th:fastbalsparse} takes at most \smash{$n/b^2\cdot \log^{O(\frac{\log n}{\log 1/b})}n$} bits of space, the overall space complexity is
	\begin{equation*}
		\otil(k \cdot n/b^3) \cdot \alpha^{O(k)} + n/b^2\cdot \log^{O(\frac{\log n}{\log 1/b})}n = n/b^3\cdot \log^{O(\frac{\log n}{\log 1/b})}n \, .
	\end{equation*}
 The running time is $\poly(n)$, since  the BSCA from \cref{th:fastbalsparse} takes at most $\poly(n)$ time.
\end{proof}

\paragraph{Overview of the algorithm.} We now give the BLD construction for Lemma~\ref{lem:tech}. Algorithms~\ref{alg:expdec} and~\ref{alg:trim} outline the core procedure, which is almost the same as the one of~\cite{streamexpdec}. For BLD parameters $b$ and $\epsilon$ we set $\phi$ appropriately in terms of~$\epsilon$, and our goal is to show that calling Algorithm~\ref{alg:expdec} as $\textsc{Decompose}(V,0)$ gives a BLD of $G$.

\begin{algorithm}[p]
	\caption{\ed: space efficient implementation of \cref{alg:metaed} for BLD construction}
	\label{alg:expdec}
	
	\begin{algorithmic}[1]
		\LeftComment{$b, \phi \in (0,1)$ are global BLD parameters}
		
		\LeftComment{$\alpha, \lambda$ are the parameters of $\bsc$}
		\LeftComment{$C \ge 1, \, c \in (0,1)$ are constant parameters}
		\Procedure{$\ed(U, \ell)$}{} \Comment{cluster $U \subseteq V$, recursion level $\ell \ge 0$ }
		\State $H=(V,E',w) \gets \textsc{Sparsifier}^c_\ell(U,b)$ \label{enum:freshsampleed} \Comment{sparsifier for $U$ in this level}
		\State $\omega \gets \bsc(H[U]^\circ, (1+\frac{1}{2 \log n}) \cdot \phi)$ \label{enum:balsparsecall}
		\Statex
		\If{$\omega = \bot$}  \Return $\{U\}$ \label{enum:stoprec} \EndIf \Comment{if $U$ is an expander, return}
		\State $(R,\nu) \gets \omega$ \Comment{otherwise, if  $R$ is balanced enough, recurse on both sides}
		\If{ $\nu\ge \frac{1}{C\lambda}\wvolc{U}{U}$}  \Return $\ed(R,\ell+1)\,  \cup \, \ed(U \setminus R,\ell+1)$ \label{enum:balrecurse} \EndIf
		\Statex
		\State $(S,\exp) \gets \trim(U,\ell)$ \label{enum:gotrim}  \Comment{find an expander inside $U$, or settle for a less sparse balanced cut}
		\If{$\exp = \top$} \Comment{if $S$ induces an expander, recurse on the complementary only} \label{condition:exp}
		\State \Return $\{S\} \cup \ed(U \setminus S,\ell+1)$ \label{rec:onside}
		\Else \Comment{otherwise, $S$ should be decently sparse and balanced, so recurse on both sides}
		\State \Return $\ed(S,\ell+1)\,  \cup \, \ed(U \setminus S,\ell+1)$  \label{rec:twosides}
		\EndIf
		\EndProcedure
	\end{algorithmic}
\end{algorithm}

\begin{algorithm}[p]
	\caption{\trim: trim some cuts to find an expander, or find a less sparse but balanced cut}
	\label{alg:trim}
	
	\begin{algorithmic}[1]
		\LeftComment{$b, \phi \in (0,1)$ are global BLD parameters}
		\LeftComment{$\alpha, \lambda$ are the parameters of $\bsc$}
		\LeftComment{$C \ge 1, \, c \in (0,1)$ are constant parameters}
		\LeftComment{$k$ is a positive integer}
		\Procedure{$\trim(U, \ell)$}{}\Comment{cluster $U \subseteq V$, recursion level $\ell \ge 0$}
		
		\State $b_0 \gets b, \,\,\, \phi_0 \gets \phi, \,\,\, A \gets U$
		\For{$j=1,\dots,k+1$} \label{enum:outerloop}
		\State \label{line:alpha} $b_j \gets \frac{ b_{j-1}}{(1+\frac{1}{\log n}) \cdot \alpha}, \,\,\, \phi_j \gets \frac{ \phi_{j-1}}{(1+\frac{1}{\log n}) \cdot \alpha}$
		\Statex
		\For{$h=1,2,\dots, \infty$} \label{enum:innerloop}
		\State \smash{$H=(V,E',w) \gets \textsc{Sparsifier}^c_{\ell,j,h}(A,b_j)$} \label{enum:freshsampletrim} \Comment{sparsifier for $A$ in this level and iteration}
		\State $\omega \gets \bsc(H[A]^\circ, (1+\frac{1}{2\log n}) \cdot \phi_j)$ \label{enum:trimcallbalsparse} 
		\Statex
		\If{$\omega = \bot$} \Return $(A,\top)$ \Comment{$\top$ means $A$ is an expander} \label{enum:expfound}\EndIf
		\State $(R,\nu) \gets \omega$
		\Statex
		\If{$\nu \ge \frac{1}{C\lambda}\wvolc{U}{U}$} \Return $(R,\bot)$  \Comment{$\bot$ means $R$ is a cut} \EndIf \label{enum:relaxouter} 
		\If{$\nu \ge \frac{1}{C\lambda}\left(\wvolc{U}{U}\right)^{1-({j-1})/{k}}$} \label{condition:baltrim} 
		\State $A \gets A \setminus R$ \label{enum:trim}  \Comment{trim}
		\Else 
		\State \textbf{break} \label{enum:break}\Comment{go to next outer loop iteration, if this happens we should have $j \le k$}
		\EndIf
		\EndFor
		\EndFor
		
		\EndProcedure
	\end{algorithmic}
\end{algorithm}

Algorithm~\ref{alg:expdec} is a space efficient implementation of \cref{alg:metaed} for BLD. In line~\eqref{enum:balsparsecall}, it uses a BSCA on a sparsifier either to certify that the given cluster is already an expander, or to find a sparse cut in it. In the former case it can simply return the cluster. In the latter case, if the condition~\eqref{enum:balrecurse} finds the cut to be balanced enough, the algorithm can afford to recurse on both sides keeping the recursion depth under control. Otherwise, it calls Algorithm~\ref{alg:trim}. This procedure is supposed to either trim off a small volume mass from the cluster in line~\eqref{enum:trim} so that the remainder induces an expander (see return statement in line~\eqref{enum:expfound}), or find a less sparse but more balanced cut than the one obtained in line~\eqref{enum:balsparsecall} of \cref{alg:expdec} (see return statement in line~\eqref{enum:relaxouter}). In the former case, \cref{alg:expdec} can then recurse only on the trimmed part of the cluster (recursion of line~\eqref{rec:onside}), whereas in the latter it is still acceptable to recurse on both sides while keeping the recursion depth under control (recursion of line~\eqref{rec:twosides}).

As both algorithms run a BSCA on weighted subgraphs of $G$, we use \cref{lem:reduction} to analyse their behaviour with respect to the input $G$. It is then crucial that the graphs on which the BSCA is run are sparsifiers for the cluster at hand. Therefore, we will implement lines~\eqref{enum:freshsampleed} and~\eqref{enum:freshsampletrim}  in \Cref{alg:expdec,alg:trim} respectively via the algorithm from Lemma~\ref{lem:distrcutstreamadd}. Crucially, we want to instantiate enough independent copies of the sparsification algorithm, so that lines~\eqref{enum:freshsampleed} and~\eqref{enum:freshsampletrim} can always access a fresh sparsifier. The memory requirement is then determined by the sparsifiers that we maintain, while the correctness will be proved using \cref{lem:reduction}.

\paragraph{Roadmap.} In the next sections we prove Lemma~\ref{lem:tech}. In \cref{subsec:offline}, we work under the assumption that  line~\eqref{enum:freshsampleed} of \cref{alg:expdec} and~\eqref{enum:freshsampletrim}  of \cref{alg:trim} always behave as expected, meaning that they return a sparsifier for the cluster at hand. In this ideal setting, we show that Algorithm~\ref{alg:expdec} outputs a BLD in small recursion depth and few iterations. Then, with \cref{lem:distrcutstreamadd}, in \cref{subsec:dyn} we lift the assumption that  line~\eqref{enum:freshsampleed} of \cref{alg:expdec} and~\eqref{enum:freshsampletrim}  of \cref{alg:trim} are deterministically correct, thus concluding a proof of \cref{lem:tech}.

\subsection{Offline analysis}
\label{subsec:offline}
The algorithms have hardwired parameters $k,\alpha,\lambda,C,c$. The BLD parameters $b,\phi \in (0,1)$ are given as input (for now, we parametrize the BLD by $b$ and $\phi$ instead of $b$ and $\epsilon$, and will later specify how to define $\phi$ from $\epsilon$).
Hereafter, whenever we use $b,\phi,k,C,c,\alpha,\lambda,(b_j)_{j},(\phi_j)_{j}$, we are referring to the corresponding parameter in \Cref{alg:expdec,alg:trim}.

As one can see from the pseudocode, the input graph $G$ is not simply fed to \Cref{alg:expdec,alg:trim}, but  they have restricted access to it: their interaction with $G$ is limited to lines~\eqref{enum:freshsampleed} and~\eqref{enum:freshsampletrim} respectively, i.e. they only have sparsifier-access to it. In this section we work under the assumption that the algorithms can always obtain a sparsifier that meets the precondition of \cref{lem:reduction}, i.e. they are cluster sparsifiers  with appropriate error (see \cref{def:spars}).

\begin{assump}
	\label{assump:sparsifiers1}
	Every time line~\eqref{enum:freshsampleed} of \cref{alg:expdec} calls \smash{$\textsc{Sparsifier}^c_{\ell}(U,b)$} for some $\ell \ge 0,\, U \subseteq V$, it gets a graph $H=(V,E',w)$ such that $H[U] \approx_\delta G[U]$, where $\delta = c^2b/\log n$.
\end{assump}

\begin{assump}
	\label{assump:sparsifiers2}
	Every time line~\eqref{enum:freshsampletrim} of \cref{alg:trim} calls $\textsc{Sparsifier}^c_{\ell,j,h}(A,b_j)$ for some $\ell~\ge~0$, $j,h~\ge~1$, $A \subseteq V$, it gets a graph $H=(V,E',w)$ such that $H[A] \approx_{\delta_j} G[A]$, where $\delta_j = c^2b_j/\log n$.
\end{assump}

\noindent
Throughout this section, we will use the definition $\delta=c^2b/\log n$ and  $\delta_j = c^2b_j/\log n$ for all $j \in [k+1]$, as in \cref{assump:sparsifiers1} and \cref{assump:sparsifiers2}. To illustrate why in the second assumption we have $\delta_j$ instead of $\delta$, we make the following observation.
\begin{obs}
	\label{rmk:bsctrim}
	For  a constant $c \le 1/30$, $\phi < b$, $b \le c$, \smash{$ \alpha \le \frac{1}{2b}$}, a cluster $A~\subseteq~V$, and a sparsifier  $H=(V,E',w)$ such that \smash{$H[A] \approx_{\delta_j} G[A]$}, \cref{lem:reduction} guarantees that calling an $(\alpha,\lambda)$-BSCA as $\bsc(H[A]^\circ, (1+\frac{1}{2\log n}) \cdot \phi_j)$ gives a \smash{$(\phi_j,(1+\frac{1}{\log n})\alpha,(1+c)\lambda, c)$}-BSCW of~$G[A]^\circ$  (as in line~\eqref{enum:trimcallbalsparse} of \cref{alg:trim}).
\end{obs}

\begin{proof}[Proof of \cref{rmk:bsctrim}]
For  a constant $c \le 1/30$, parameters $b, \phi, \delta,\alpha$ such that $\phi < b$, $b \le c$, $ \alpha \le \frac{1}{2b}$, a cluster $A~\subseteq~V$, and a sparsifier  $H=(V,E',w)$ such that $H[A] \approx_\delta G[A]$, \cref{lem:reduction} guarantees that calling an $(\alpha,\lambda)$-BSCA as $\bsc(H[A]^\circ, (1+0.5/{\log n}) \cdot \phi)$ gives a  \smash{$(\phi,(1+\frac{1}{\log n})\alpha,(1+c)\lambda, c)$}-BSCW of~$G[A]^\circ$. Then, we recall that \smash{$G[A]^\circ = G[A]^{b/\phi}$} and \smash{$H[A]^\circ = H[A]^{b/\phi}$}, so it also holds that \smash{$G[A]^\circ = G[A]^{b_j/\phi_j}$} and \smash{$H[U]^\circ = H[U]^{b_j/\phi_j}$}, since we can see from \cref{alg:trim} that $b_j/\phi_j=b/\phi$ for all $j \in [k+1]$. Moreover, observe the following: $\phi < b$ if and only if $\phi_j < b_j$, $ \alpha \le \frac{1}{2b}$ implies \smash{$ \alpha \le \frac{1}{2b_j}$} (by definition of $b_j$ in \cref{alg:trim}), and $b \le c$ implies $b_j \le c$ (again, by definition of $b_j$ in \cref{alg:trim}). One can then use \cref{lem:reduction} replacing $b,\phi,\delta$ with $b_j,\phi_j,\delta_j$ respectively, hence getting \cref{rmk:bsctrim}.
\end{proof}

\noindent
The goal of this section is threefold: in \cref{subsubsec:iter}, we bound the number of iterations that any execution of \cref{alg:trim} can go through; in \cref{subsubsec:depth}, we bound the recursion depth of \cref{alg:expdec}; in \cref{subsubsec:correct}, we show the output of calling \cref{alg:expdec} as $\ed(V,0)$ is a BLD of $G$. Interestingly,  \cref{subsubsec:correct} crucially uses the results from \cref{subsubsec:iter} and  \cref{subsubsec:depth} about recursion depth and number of iterations. Before, we give some preliminary results: in \cref{subsubsec:nested} we give properties of cuts of boundary-linked subgraphs; in \cref{subsubsec:estim}, we show how relations between volume estimates (such as the comparisons performed by \cref{alg:expdec} in line~\eqref{enum:balrecurse} and by \cref{alg:trim} in lines~\eqref{enum:relaxouter} and~\eqref{condition:baltrim}) translate to relations between the actual volumes of the corresponding cuts.

\subsubsection{Properties of nested cuts}
\label{subsubsec:nested}
When \cref{alg:expdec} makes a cut in the input cluster $U$, this cut is then fed as input to the recursive call in line~\eqref{enum:balrecurse}, which will in turn try to find a cut inside it. Similarly, \cref{alg:trim} makes a cut in the current cluster $A$ in line~\eqref{enum:trim}, and in the next iteration it will try to find a cluster inside the remainder of $A$. The leitmotif of these algorithms is that they try to make cuts inside a cluster that is one side of a cut previously made. It is then useful for the analysis to relate volume and sparsity of nested cuts.

\begin{lemma}[Properties of nested cut]
	\label{lem:relationsvol}
	Let $b, \phi \in (0,1)$ and $\delta > 0$ be such that $\phi \le b$, $\delta < 1/b$. Fix a cluster $U \subseteq V$, let $\emptyset \neq S \subsetneq U$ be a cut such that $\unspc{S}{U}<\delta \cdot \phi$, and let $T \subseteq S$. Then we have the following relations between volumes:
	\begin{equation*}
		\unvolc{T}{U} \overset{(1)}{\le} \unvolc{T}{S} \overset{(2)}{<} \unvolc{T}{U}+\delta b \cdot \min\{\unvolc{S}{U}, \unvolc{U \setminus S}{U}\}   \overset{(3)}{\le} \unvolc{U}{U}\, .
	\end{equation*}
	If it also holds that $\emptyset \neq T \subsetneq S$ has $\unspc{T}{S} < \delta \cdot \phi$ and $\unvolc{T \cup (U \setminus S)}{U} < \frac{1}{2} \unvolc{U}{U}$ then
	\begin{equation*}
		\unspc{T \cup (U \setminus S)}{U} < \delta \cdot \phi \, .
	\end{equation*}
\end{lemma}

\begin{proof}For convenience let $\psi= \delta \cdot \phi$. Inequality~(1) holds simply because $T \subseteq S \subseteq U$, so the number of self-loops attached to $T$ cannot decrease when going from $G[U]^\circ$ to $G[S]^\circ$. Inequality~(2) can be obtained by rewriting the volumes in $G[S]^\circ$ as
	\begin{align}
		\unvolc{T}{S} = \vol(T) + \frac{b-\phi}{\phi} \unbor{T}{S} & = \unvolc{T}{U} + \frac{b-\phi}{\phi} |E(T,U \setminus S)| \\
		& \leq \unvolc{T}{U} + \frac{b}{\phi} |E(T,U \setminus S)| \, .  \label{ineq:bound-st-ut}
	\end{align}
	Since $S$ is $\psi$-sparse in $G[U]^\circ$, we can upper-bound $|E(T,U \setminus S)|$ by $\uncut{S}{U} < \psi \min\{\unvolc{S}{U}, \unvolc{U \setminus S}{U}\}$, thus getting the bound. For inequality~(3), note $\unvolc{T}{U} \le \unvolc{S}{U}$ because $T \subseteq S$ and $\delta  b\cdot \min\{\unvolc{S}{U}, \unvolc{U \setminus S}{U}\} \le \unvolc{U \setminus S}{U}$ for $\delta \le 1/b$. Since $\unvolc{S}{U}+\unvolc{U \setminus S}{U}=\unvolc{U}{U}$, we have the bound (3).
	
	Suppose now that $\emptyset \neq T \subsetneq S$ has $\unspc{T}{S} < \psi$ and $\unvolc{T \cup (U \setminus S)}{U} < \frac{1}{2} \unvolc{U}{U}$. First observe that because $\unspc{T}{S} < \psi$, the definition of $\unspc{T}{S} $ yields $\uncut{T}{S} < \psi \unvolc{T}{S}$. By \cref{ineq:bound-st-ut}, we also have
	\begin{equation}
		\label{ineq:boundsvsu}
		\uncut{T}{S} < \psi \unvolc{T}{S} \leq \psi \unvolc{T}{U} + \frac{\psi b}{\phi} \cdot |E(T,U \setminus S)| \le \psi \unvolc{T}{U} +|E(T,U \setminus S)| \, ,
	\end{equation}
	where we used that $\psi  b / \phi  \le 1$ since $\psi = \delta\phi$ and $\delta \le 1/b$. Due to $\unvolc{T \cup (U \setminus S)}{U} < \frac{1}{2} \unvolc{U}{U}$, the definition of $\unspc{T \cup (U \setminus S)}{U}$ finally yields
	\begin{align*}
		\unspc{T \cup (U \setminus S)}{U} & = \frac{\uncut{(T \cup (U \setminus S))}{U}}{\unvolc{T \cup (U \setminus S)}{U}} \\
		\text{since $T \subseteq S$, so $T \cap (U \setminus S) = \emptyset$} \quad & = \frac{\uncut{T}{S} +\uncut{S}{U}-|E(T,U \setminus S)|}{\unvolc{T}{U}+\unvolc{U \setminus S}{U}} \\
		\text{by~\eqref{ineq:boundsvsu}} \quad& < \frac{\psi \unvolc{T}{U} + |E(T,U \setminus S)| +\uncut{S}{U}-|E(T,U \setminus S)|}{\unvolc{T}{U}+\unvolc{U \setminus S}{U}} \\
		& = \frac{\psi \unvolc{T}{U} +\uncut{S}{U}}{\unvolc{T}{U}+\unvolc{U \setminus S}{U}} \\
		& \le \max\left\{\psi, \frac{\uncut{S}{U}}{\unvolc{U \setminus S}{U}}\right\} = \max\left\{\psi, \unspc{S}{U}\right\} \\
		\text{since $\unspc{S}{U} <\psi$}\quad& \le \psi \, .
	\end{align*}
\end{proof}
\noindent
\Cref{alg:trim} makes cuts iteratively. It finds a cut in the cluster, it trims it off in line~\eqref{enum:trim}, it finds another cut in the remainder, it trims it off, and so on, provided that these cuts verify certain conditions: their individual volumes are neither too large (if one of them was very large in volume, the algorithm would return in line~\eqref{enum:relaxouter}) nor too small (see condition~\eqref{condition:baltrim}). The following lemma proves properties of such cut sequences. \Cref{fig:nested} should help the reader to visualize the statement and the proof.

\begin{figure}[h]
	\begin{minipage}[center]{\textwidth}
		\centering
		\includegraphics[scale=0.8]{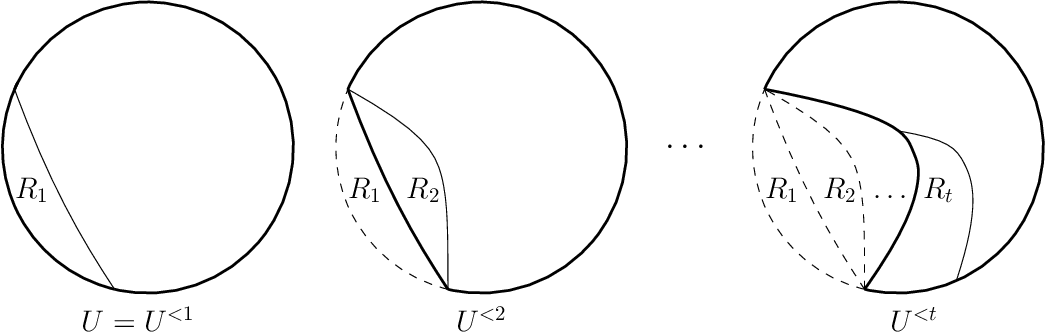}
	\end{minipage}
	\caption{Sequence of nested cuts.}
	\label{fig:nested}
\end{figure}

\begin{lemma}[Properties of a sequence of nested cuts]
	\label{lem:seqbounds}
	Let $b, \phi \in (0,1)$ be such that $\phi \le b$, and $\delta \in (0,1]$. Fix a  cluster $U \subseteq V$, let $\theta \in \mathbb{R}_{\ge 0}$ be such that $\theta \le \frac{1}{5}\unvolc{U}{U}$, and let $R_1, \dots, R_t$ be a sequence of cuts \smash{$\emptyset \neq R_i \subsetneq U^{<i}$} where \smash{$ U^{<i} = U \setminus (\cup_{z =1}^{i-1} R_z)$} for all $i \in [t]$. If
	\begin{enumerate}
		\item \label{prec:sparse} for all $\emptyset \neq  X \subsetneq U$ with $\unspc{X}{U} < \delta \phi$ and $ \unvolc{X}{U} \le \unvolc{U \setminus X}{U}$ one has $\unvolc{X}{U}~<~\theta$,~and
		\item \label{prec:balsparse} for all $i \in [t]$ one has \smash{$\unspc{R_i}{U^{<i}} < \delta \phi$} and $ \unvolc{R_i}{U^{<i}} < \frac{1}{5}\unvolc{U}{U}$,
	\end{enumerate}
	then we have the following.
	\begin{enumerate}
		\item \label{part:1} For all $i \in [t]$, one has $\unspc{\cup_{z=1}^i R_z}{U} < \delta \phi$ and  $\unvolc{\cup_{z=1}^i R_z}{U} < \theta$.
		\item  \label{part:2b} For all $i \in [t]$, if there is $\theta^{\low} \in \mathbb{R}_{\ge 0}$ such that $\unvolc{R_i}{U^{<i}} \ge \theta^{\low}$ then one also has $\unvolc{R_i}{U} \ge \theta^{\low}-b\cdot \theta$.
		\item  \label{part:2c} Let $ U^{<t+1} = U \setminus (\cup_{z =1}^{t} R_z)$. If $\emptyset \neq R \subsetneq U^{<t+1}$ is a cut such that \smash{$\unspc{R}{U^{<t+1}} < \delta \phi$} and \smash{$\theta^{\low} \le \unvolc{R}{U^{<t+1}} \le \frac{5}{9}\unvolc{U^{<t+1}}{U^{<t+1}}$} for some $\theta^{\low} \le \frac{4}{9}\unvolc{U}{U}$, then \smash{$\unspc{R}{U} < \delta \phi \cdot \frac{\unvolc{U}{U}}{\theta^{\low}-b \cdot \theta}$} and $\theta^{\low}-b \cdot \theta \le \unvolc{R}{U} \le \frac{5}{9}\unvolc{U}{U}$.
	\end{enumerate}
\end{lemma}
 
 \noindent
 The end goal of part~\eqref{part:1} is to upper bound the number of edges crossing and the remaining volume when \cref{alg:trim} returns $(A,\top)$ and therefore \cref{alg:expdec} recurses on $U \setminus A$ only (see line~\eqref{enum:stoprec} in \cref{alg:expdec}). The idea is that in this case $U \setminus A$ corresponds to the union of nested cuts that are trimmed off by \cref{alg:trim}.
 
 Part~\eqref{part:2b} will be useful to bound the number of iterations that \cref{alg:trim} can take. In every iteration a cut is trimmed off, and from condition~\eqref{condition:baltrim} we know the volume of these cuts to be not too small in $U^{<i}$. From part~\eqref{part:2b} we then know these volumes to be not too small in $U$ either, and hence the number of such cuts cannot be too large, since they are disjoint.
 
  The end goal of part~\eqref{part:2c} is to upper bound the number of edges crossing and the remaining volume when \cref{alg:trim} returns $(R,\bot)$ and therefore \cref{alg:expdec} recurses on $R$ and $U \setminus R$ (see line~\eqref{enum:balrecurse} in \cref{alg:expdec}). The idea is that in this case, every cut trimmed off by \cref{alg:trim} is sparse and far from being balanced (because of condition~\eqref{enum:relaxouter} in \cref{alg:trim}), whereas the last one is both sparse and somewhat balanced (because of condition~\eqref{enum:relaxouter} in \cref{alg:trim}).
 \\~\\
 In the remainder of this section we prove the three parts of this lemma.
 \begin{proof}[Proof of \cref{lem:seqbounds}, part~\eqref{part:1}]
 	Intuitively, the union of a sequence of nested sparse and small cuts should be sparse and small too. We prove this by induction on $i \in [t]$. Denote for convenience~$\psi~=~\delta~\phi$.
 	
 	The base case for $i=1$ readily follows. Since $U^1 =U$, we have $\unspc{R_1}{U} < \psi$ and also $\unvolc{R_1}{U} < \frac{1}{5}\unvolc{U}{U}$ from precondition~\eqref{prec:balsparse}. Moreover,  by precondition~\eqref{prec:sparse}, there is no $\psi$-sparse cut $X$ with $\theta \le \unvolc{X}{U} < \frac{1}{2}\unvolc{U}{U}$. Hence, it must be the case that $\unvolc{R_1}{U} < \theta$.
 	
 	For the inductive step, take $i \in [t-1]$ and assume that $\unspc{\cup_{z=1}^i R_z}{U} < \psi$ and $\unvolc{\cup_{z=1}^i R_z}{U} < \theta$. We now apply Lemma~\ref{lem:relationsvol} with $S = U^{<i+1}$ and $T = R_{i+1}$. We can do so because $\unspc{S}{U} = \unspc{\cup_{z=1}^i R_z}{U} < \psi$, so the lemma gives
 	\begin{align*}
 		\unvolc{R_{i+1}}{U} \le \unvolc{R_{i+1}}{U^{<i+1}} \le \frac{1}{5}\unvolc{U}{U}  < \frac{1}{4}\unvolc{U}{U}\, .
 	\end{align*}
 	Combining this bound with the inductive hypothesis that $\unvolc{\cup_{z=1}^i R_z}{U} < \theta$, one has
 	\begin{equation*}
 		\unvolc{\cup_{z=1}^{i+1} R_z}{U} = \unvolc{\cup_{z=1}^i R_z}{U}+\unvolc{R_{i+1}}{U} < \theta+\unvolc{R_{i+1}}{U} < \frac{1}{5}\unvolc{U}{U}+ \frac{1}{4}\unvolc{U}{U}  < \frac{1}{2}\unvolc{U}{U} \, ,
 	\end{equation*}
 which rewrites as $\unvolc{T \cup (U \setminus S)}{U} < \frac{1}{2} \unvolc{U}{U}$ with  $S = U^{<i+1}$ and $T = R_{i+1}$. Also, from precondition~\eqref{prec:balsparse}, we have $\unspc{R_{i+1}}{U^{<i+1}} < \psi$, or in other words $\unspc{T}{S} < \psi$ with $S = U^{<i+1}$ and $T = R_{i+1}$. Then, the second part of Lemma~\ref{lem:relationsvol} gives $\unspc{\cup_{z=1}^{i+1} R_z}{U} < \psi$, so one part of the induction is proved. For the other part, we combine $\unspc{\cup_{z=1}^{i+1} R_z}{U} < \psi$ and $\unvolc{\cup_{z=1}^{i+1} R_z}{U} < \frac{1}{2}\unvolc{U}{U}$ with precondition~\eqref{prec:sparse}: since we assume there is no $\psi$-sparse cut in $G[U]^\circ$ with volume sandwiched between $\theta$ and $\frac{1}{2}\unvolc{U}{U}$, it must be the case that $\unvolc{\cup_{z=1}^{i+1} R_z}{U} < \theta$.
\end{proof}

 \begin{proof}[Proof of \cref{lem:seqbounds}, part~\eqref{part:2b}]
 	We prove this using part~\eqref{part:1} and \cref{lem:relationsvol}. The case of $i=1$ follows from the preconditions, since $U^1=U$. Let then $i =2, \dots, t$, and denote for convenience~$\psi~=~\delta~\phi$. Then, part~\eqref{part:1} gives that $\unspc{\cup_{z=1}^{i-1} R_z}{U} < \psi$. If we set $S=U^{<i}$ and $T=R_i$ we have $\unspc{S}{U} = \unspc{\cup_{z=1}^{i-1} R_z}{U} < \psi$, and can then apply \cref{lem:relationsvol} to get the lower bound
	\begin{align*}
		\unvolc{R_i}{U} & > \unvolc{R_i}{U^{<i}} - \delta b \cdot \unvolc{\cup_{z=1}^{i-1} R_z}{U} \\
		\text{since $\delta \le 1$} \quad & \ge  \unvolc{R_i}{U^{<i}} - b \cdot \unvolc{\cup_{z=1}^{i-1} R_z}{U} \, .
	\end{align*}
	Part~\eqref{part:1} also gives $\unvolc{\cup_{z=1}^{i-1} R_z}{U} < \theta $, and since we assume that $ \unvolc{R_i}{U^{<i}} \ge \theta^{\low}$, one concludes
		\begin{equation*}
		\unvolc{R_i}{U} > \unvolc{R_i}{U^{<i}} - b \cdot \unvolc{\cup_{z=1}^{i-1} R_z}{U}  \ge \theta^{\low} - b\cdot \theta \, .
	\end{equation*}
\end{proof}

 \begin{proof}[Proof of \cref{lem:seqbounds}, part~\eqref{part:2c}]
	We begin by analysing the volume $\unvolc{R}{U}$. Again, from part~\eqref{part:1} we have
	\begin{equation*}
		\unspc{\cup_{i=1}^{t} R_i}{U} < \delta \phi\, , \quad \unvolc{\cup_{i=1}^{t} R_i}{U} < \theta \,.
	\end{equation*}
	Setting $\psi = \delta \phi$, $S=U^{<t+1}$ and $T=R$,  we have $\unspc{S}{U} = \unspc{\cup_{z=1}^t R_z}{U} < \psi$, so we can apply \cref{lem:relationsvol} to yield
	\begin{equation} \label{eq:lbR}
		\unvolc{R}{U} > \unvolc{R}{U^{<t+1}} - b \cdot \unvolc{\cup_{i=1}^{t} R_i}{U} > \theta^{\low}-b \cdot \theta \, ,
	\end{equation}
	and
	\begin{equation}  \label{eq:ubR}
		\unvolc{R}{U} \le \unvolc{R_t}{U^{<t+1}} \le \frac{5}{9}\unvolc{U^{<t+1}}{U^{<t+1}} \, .
	\end{equation}
	We can once more apply Lemma~\ref{lem:relationsvol} with $S=T=U^{<t+1}$ (since again $\unspc{S}{U} = \unspc{\cup_{z=1}^t R_z}{U} < \psi$), so
	\begin{equation}
		 \label{eq:ubU}
		\unvolc{U^{<t+1}}{U^{<t+1}} \le \unvolc{U}{U} \, .
	\end{equation}
	We then have from~\eqref{eq:lbR},~\eqref{eq:ubR},~\eqref{eq:ubU}
	\begin{equation}
		\label{eq:sandwichR}
		 \theta^{\low}-b \cdot \theta \le \unvolc{R}{U} \le  \frac{5}{9}\unvolc{U^{<t+1}}{U^{<t+1}} \le  \frac{5}{9}\unvolc{U}{U} \, .
	\end{equation}
	Next we analyse the sparsity of $R$ by first upper-bounding the quantity $\uncut{R}{U}$:
	\begin{align}
		\uncut{R}{U} & \le {\uncut{R}{U^{<t+1}}+\uncut{(\cup_{i =1}^{t}R_i)}{U}} \\
		\text{since $\unspc{R}{U^{<t+1}} <\psi$ by assumption}\quad& < {\psi \cdot \unvolc{R}{U^{<t+1}}+\uncut{(\cup_{i =1}^{t}R_i)}{U}}\\
		\text{since $\unspc{\cup_{i =1}^{t}R_i}{U} <\psi$ from part~\eqref{part:1}} \quad & < {\psi \cdot \unvolc{R}{U^{<t+1}}+\psi \cdot \unvolc{\cup_{i =1}^{t}R_i}{U}} \\
		\text{since $\unvolc{R}{U} < \frac{5}{9}\unvolc{U}{U}$ from~\eqref{eq:sandwichR}} \quad& < {\psi \cdot  \frac{5}{9} \cdot \unvolc{U}{U}+\psi \cdot \unvolc{\cup_{i =1}^{t}R_i}{U}}\\
		\text{since $\unvolc{\cup_{i =1}^{t}R_i}{U}< \theta $ from  part~\eqref{part:1}}\quad & < \psi \cdot  \frac{5}{9} \cdot \unvolc{U}{U}+\psi \cdot \theta  \\
		\text{by assumption on $\theta$}\quad & < \psi \cdot {\unvolc{U}{U}} \, . \label{eq:ubcutR}
	\end{align}
On the other hand, the bounds in~\eqref{eq:sandwichR} yield
\begin{align}
	\min\{\unvolc{R}{U},\unvolc{U \setminus R}{U}\} & \ge 	\min\left\{ \theta^{\low}-b \cdot \theta, \frac{4}{9}\unvolc{U}{U}\right\} \\
	\text{since $\theta^{\low} \le \frac{4}{9} \unvolc{U}{U} $ by assumption} \quad & = \theta^{\low}-b \cdot \theta \, .  \label{eq:lbminvolR}
\end{align}
From the bound~\eqref{eq:ubcutR} and~\eqref{eq:lbminvolR} we finally conclude
\begin{equation*}
	\unspc{R}{U} = \frac{	\uncut{R}{U}}{\min\{\unvolc{R}{U},\unvolc{U \setminus R}{U}\}} < \frac{\psi \cdot {\unvolc{U}{U}}}{\theta^{\low}-b \cdot \theta} \, .
\end{equation*}
\end{proof}

\subsubsection{Relations between volume estimates}
\label{subsubsec:estim}
In this section we show how relations between volume estimates (such as the comparisons performed by \cref{alg:expdec} in line~\eqref{enum:balrecurse} and by \cref{alg:trim} in lines~\eqref{enum:relaxouter} and~\eqref{condition:baltrim}) translate to relations between the actual volumes of the corresponding cuts.

\begin{lemma}[Relations between volume estimates]
	\label{lem:algovols}
	Let $b, \phi, c \in (0,1)$ be such that $c \leq 1/4$ is a constant and $\phi \le b$. Fix a cluster $U \subseteq V$, and also let $H=(V,E',w)$ such that $H[U]^\circ \approx_c G[U]^\circ$, and let $\omega$ be a  $(\psi,\alpha',\lambda', c)$-BSCW of~$G[U]^\circ$ for some $\psi \in (0,\phi]$ and $\alpha',\lambda' \ge 1$. If $\omega = (R,\nu)$ for some $\emptyset \neq R \subsetneq U$ and $\nu \in \mathbb{R}_{>0}$, then for any $\rho,\tau \in [0,1]$ and $X \subseteq V$ we have
	\begin{enumerate}
		\item if $\nu \ge \rho \cdot \left(\wvolc{X}{X}\right)^\tau$, then $\unvolc{R}{U} \ge \rho/2 \cdot \left(\unvolc{X}{X}\right)^\tau$.
		\item if $\nu < \rho \cdot \left(\wvolc{X}{X}\right)^\tau$, then $\unvolc{R}{U} < 2\rho \cdot \left(\unvolc{X}{X}\right)^\tau$ and for every cut $\emptyset \neq  S \subsetneq U$ with $\unspc{S}{U} < \psi$ and $ \unvolc{S}{U} \le \unvolc{U \setminus S}{U}$ one has $\unvolc{S}{U} < 2\lambda' \rho \cdot \left(\unvolc{X}{X}\right)^\tau $.
	\end{enumerate}
\end{lemma}
\begin{proof}
	From property~\eqref{item:goodestimr} of the definition of BSCW (see \cref{def:bscw}), we know $(1-c)\unvolc{R}{U} \le \nu \le (1+c)\unvolc{R}{U} $, and from property~\eqref{item:balr} we know  that for every cut $\emptyset \neq  S \subsetneq U$ with $\unspc{S}{U} < \psi$ and $ \unvolc{S}{U} \le \unvolc{U \setminus S}{U}$ one has $\unvolc{S}{U} < \lambda'  \cdot \unvolc{R}{U}$.
	
	On the other hand, $(1-c)\unvolc{X}{X} \le \wvolc{X}{X} \le (1+c)\unvolc{X}{X}$ for any $X \subseteq V$. To see why, first note that $\wvolc{X}{X}=\wvol{X}+\frac{b-\phi}{\phi}\wbor{X}{X}=\wvol{X}+\frac{b-\phi}{\phi}\wcutg{X}$, and analogously \smash{$\unvolc{X}{X}=\vol(X)+\frac{b-\phi}{\phi}\uncutg{X}$}. Then, since $H[U]^\circ \approx_c G[U]^\circ$, we have the approximation guarantee for global cuts (see \cref{def:spars}): applying it to all singleton cuts in $X$ we can approximate $\vol(X)$, and applying it to $X$ we can approximate $\uncutg{X}$.
	
	From the above discussion it follows that for all $\rho,\tau \in [0,1]$ and $X \subseteq V$, $\nu \ge \rho \cdot \left(\wvolc{X}{X}\right)^\tau$ implies $\unvolc{R}{U} \ge \rho \frac{1-c}{1+c} \cdot \left(\unvolc{X}{X}\right)^\tau \ge \rho/2 \cdot \left(\unvolc{X}{X}\right)^\tau$, since $c \le 1/4$.  If instead $\nu < \rho \cdot \left(\wvolc{X}{X}\right)^\tau$, we have \smash{$\unvolc{R}{U} < \rho \frac{1+c}{1-c} \cdot \left(\unvolc{X}{X}\right)^\tau \le 2\rho \cdot \left(\unvolc{X}{X}\right)^\tau$}. Moreover, since by definition of $\omega$ one has $\unvolc{S}{U} < \lambda'  \cdot \unvolc{R}{U}$  for every cut $\emptyset \neq  S \subsetneq U$ with $\unspc{S}{U} < \psi$ and $ \unvolc{S}{U} \le \unvolc{U \setminus S}{U}$, we get $\unvolc{S}{U} < \lambda'  \cdot \unvolc{R}{U} \le \lambda' 2\rho \cdot \left(\unvolc{X}{X}\right)^\tau$ for any such $S$.
\end{proof}

\subsubsection{Bounding the number of iterations}
\label{subsubsec:iter}

\noindent
Combining the structural properties of nested cuts together with \cref{lem:algovols}, we can show the correctness of the algorithms. We begin by observing that whenever \cref{alg:trim} breaks the inner loop in line~\eqref{enum:break} we have $j \le k$. This ensures that restraining the outer loop to indices $j=1,\dots, k+1$ (see line~\eqref{enum:outerloop}) does not lead to undefined behaviour.

\begin{lemma}[\cref{alg:trim} does few outer loop iterations]
	\label{lem:boundouterloop}
	Let $b, \phi,c \in (0,1)$ such that $c \leq 1/30$ is a constant, $\phi < b$ and $b \le c$. Also let $C \ge 3$ be a constant, and let $\bsc$ be an $(\alpha,\lambda)$-BSCA where $\alpha \le \frac{1}{2b}$. If \cref{assump:sparsifiers2} holds, then if Algorithm~\ref{alg:trim} breaks in line~\eqref{enum:break} it is in some outer loop iteration $j$ with $j \le k$.
\end{lemma}
\begin{proof}
	We prove this by contradiction, so suppose the algorithm enters line~\eqref{enum:break} with $j=k+1$ for some inner loop iteration $h \ge 1$. Let $A$ be the value of the corresponding variable in \cref{alg:trim} at the beginning of inner loop iteration $h$ within outer loop iteration $k+1$. Also let $H=(V,E',w)$ be the graph gotten from line~\eqref{enum:freshsampletrim} in the same iteration. From the description of the algorithm, since it entered line~\eqref{enum:break}, the call $\bsc(H[A]^\circ, (1+{1}/{(2\log n)}) \cdot \phi_{k+1})$ from line~\eqref{enum:trimcallbalsparse} has returned a BSCW of the form $(R,\nu)$ such that \smash{$\nu< \frac{1}{C\lambda}(\wvolc{U}{U})^{1-(k+1-1)/k}=\frac{1}{C\lambda}(\wvolc{U}{U})^0$}. Moreover, since we have $c \le 1/30$, $\phi < b$, $b \le c$, $\alpha \le \frac{1}{2b}$,  and $H[A] \approx_{\delta_{k+1}} G[A]$ (from \cref{assump:sparsifiers2}), \cref{lem:reduction} ensures that $(R,\nu)$ is a \smash{$(\phi_j,(1+\frac{1}{\log n})\alpha,(1+c)\lambda, c)$}-BSCW of~$G[A]^\circ$ (see \cref{rmk:bsctrim}).
	
	We now want to use Lemma~\ref{lem:algovols} with $\psi = \phi_{k+1}$, $ \alpha'=(1+\frac{1}{\log n})\alpha$, $\lambda' =(1+c)\lambda$, $\rho= \frac{1}{C\lambda}$, $\tau=0$, and $X=U$. These parameters verify the conditions for applying \cref{lem:algovols} to $(R,\nu)$ on $G[A]^\circ$: $c \le 1/30 \le 1/4$, $\psi = \phi_{k+1} \le \phi$ (by definition of $\phi_{k+1}$ in the algorithm), $\alpha',\lambda' \ge 1$, $\rho \le 1$, $\tau \in [0,1]$, $H[A] \approx_{c} G[A]$ (since $H[A] \approx_{\delta_{k+1}} G[A]$ and $\delta_{k+1} \le c$), and  hence $(R,\nu)$ is a \smash{$(\phi_j,(1+\frac{1}{\log n})\alpha,(1+c)\lambda, c)$}-BSCW of~$G[A]^\circ$. Then \cref{lem:algovols} gives
	\begin{equation*}
		\unvolc{R}{A} < 2 \rho (\unvolc{U}{U})^\tau = \frac{2}{C\lambda} < 1 \, ,
	\end{equation*}
	which in particular implies $R=\emptyset$. However, this contradicts the fact that $(R,\nu)$ is a BSCW of~$G[A]^\circ$ (see \cref{def:bscw}).
\end{proof}

\noindent
Next, we show that \cref{alg:trim} verifies the preconditions~\eqref{prec:sparse} and~\eqref{prec:balsparse} of \cref{lem:seqbounds}. Specifically, if \cref{alg:expdec} calls \cref{alg:trim}, there cannot be sparse and balanced cuts in $U$, and also every cut trimmed by \cref{alg:trim} is sparse and small in volume. This will later allow us to apply \cref{lem:seqbounds} for bounding the number of inner loop iterations (and also for bounding the recursion depth, in the next section).

\begin{lemma}[\cref{alg:trim} makes a sequence of nested sparse small cuts]
	\label{lem:boundsequencetrim}
	Let $b, \phi,c \in (0,1)$ such that $c \leq 1/30$ is a constant, $\phi < b$ and $b \le c$. Also let $C \ge 11$ be a constant, and let $\bsc$ be an $(\alpha,\lambda)$-BSCA where $\alpha \le \frac{1}{2b}$. Consider an execution of \cref{alg:trim} called from \cref{alg:expdec} on input cluster $U \subseteq V$. Also let $R_1,\dots,R_q \subseteq U$ be the sequence of $q$ cuts trimmed in line~\eqref{enum:trim} by \cref{alg:trim}, and let $U^{<i}=U \setminus (\cup_{z-1}^{i-1}R_z)$. If \cref{assump:sparsifiers1} and \cref{assump:sparsifiers2} hold, then
	\begin{enumerate}
		\item \label{satprec:sparse} for all \smash{$\emptyset \neq  S  \subsetneq U$} with \smash{$\unspc{S}{U} < \phi$, $ \unvolc{S}{U} \le \unvolc{U \setminus S}{U}$}, one has \smash{$\unvolc{S}{U}<\frac{3}{C}\unvolc{U}{U}$},
		\item \label{satprec:balsparse} for all $i \in [q]$ one has \smash{$\unspc{R_i}{U^{<i}} < \phi$} and $ \unvolc{R_i}{U^{<i}} < \frac{1}{5}\unvolc{U}{U}$.
	\end{enumerate}
\end{lemma}
\begin{proof}
	From the description of Algorithms~\ref{alg:expdec} and~\ref{alg:trim}, we see that Algorithm~\ref{alg:trim} starts only if line~\eqref{enum:balsparsecall} of \cref{alg:expdec} returns a BSCW $\omega=(R,\nu)$ such that $\nu < \frac{1}{C \lambda}\wvolc{U}{U}$, as per line~\eqref{enum:gotrim} (otherwise, \cref{alg:expdec} would not call \cref{alg:trim}). Now we want to apply \cref{lem:reduction} to $\omega$. Because we have $c\le 1/30$, $\phi < b$, $b \le c$, $\delta=c^2 b/\log n $, $\alpha \le \frac{1}{2b}$ and $H[U] \approx_{\delta} G[U]$ from \cref{assump:sparsifiers1}, all the preconditions of \cref{lem:reduction} are met for us to apply it to the call $\bsc(H[U]^\circ, (1+\frac{1}{2\log n}) \cdot \phi)$ from line~\eqref{enum:balsparsecall} of \cref{alg:expdec}: hence, $\omega$ is a  \smash{$(\phi,(1+\frac{1}{\log n})\alpha,(1+c)\lambda, c)$}-BSCW of~$G[U]^\circ$. Since we observed just above that $\omega$ is of the form $(R,\nu)$ and \smash{$\nu < \frac{1}{C \lambda}\wvolc{U}{U}$}, we now use \cref{lem:algovols} with parameters $\psi=\phi$, \smash{$ \alpha'=(1+\frac{1}{\log n})\alpha$}, $\lambda' =(1+c)\lambda$, \smash{$\rho = \frac{1}{C\lambda}$}, $\tau=1$, and $X=U$. These parameters verify the requirement of \cref{lem:algovols} since we have $ c\le 1/30 \le 1/4$, $H[U] \approx_{c} G[U]$ (from \cref{assump:sparsifiers1}, since $\delta \le c$), $\psi \le \phi$, $\alpha',\lambda' \ge 1$, $\rho,\tau \le 1$, and  $(R,\nu)$ is a \smash{$(\phi,(1+\frac{1}{\log n})\alpha,(1+c)\lambda,c)$}-BSCW of~$G[U]^\circ$. Then we get that for every cut $\emptyset \neq  S \subsetneq U$ with $\unspc{S}{U} < \phi$ and $ \unvolc{S}{U} \le \unvolc{U \setminus S}{U}$ one has
	\begin{equation*}
		\unvolc{S}{U} < 2\lambda' \rho \cdot \unvolc{U}{U} = \frac{2(1+c)\lambda}{C\lambda}  \unvolc{U}{U} \leq \frac{3}{C} \cdot \unvolc{U}{U} \, .
	\end{equation*}
	This gives part~\eqref{satprec:sparse}.
	
	For all $i \in [q]$, denote by $j(i) \in [k+1]$ the outer loop iteration of \cref{alg:trim} where $R_i$ was trimmed, and also let $U^{<i}=U \setminus (\cup_{z=1}^{i-1} R_z)$ (i.e., at the beginning of the iteration where $R_i$ has been trimmed we had $A = U^{<i}$). Then, from the description of the algorithm we know that $R_i$ was paired with a value $\nu_{R_i}$ such that $(R_i,\nu_{R_i})$ was the result of $\bsc(H[U^{<i}]^\circ, (1+\frac{1}{2\log n}) \cdot \phi_{j(i)-1})$, where $H=(V,E',w)$ verifies \smash{$H[U^{<i}] \approx_{\delta_{j(i)}} G[U^{<i}] $} because of \cref{assump:sparsifiers2}. As we have $ c\le 1/30$, $\phi < b$, $b \le c$, \smash{$\alpha \le \frac{1}{2b}$}, \cref{lem:reduction} ensures that the pair $(R_i,\nu_{R_i})$ is a \smash{$(\phi_{j(i)},(1+\frac{1}{\log n})\alpha,(1+c)\lambda,c)$}-BSCW of~$G[U^{<i}]^\circ$ for all $i \in [q]$ (see \cref{rmk:bsctrim}).  From the description of the algorithm  we also know that $\nu_{R_i}$ must also not satisfy condition~\eqref{enum:relaxouter}, as otherwise \cref{alg:trim} would have returned, so $\nu_{R_i} < \frac{1}{C \lambda}\wvolc{U}{U}$ for all $i \in [q]$. For each $i \in [q]$, we now use \cref{lem:algovols} with parameters $\psi=\phi_{j(i)}$, \smash{$ \alpha'=(1+\frac{1}{\log n})\alpha$}, $\lambda' =(1+c)\lambda$, $\rho = \frac{1}{C\lambda}$, $\tau=1$, and $X=U$. These parameters verify the requirement for applying \cref{lem:algovols} to $(R_i,\nu_{R_i})$ in $G[U^{<i}]$: $ c\le 1/30 \le 1/4$, $H[U^{<i}] \approx_{c} G[U^{<i}]$ (from \cref{assump:sparsifiers2} and $\delta_{j(i)} \le c$), $\psi=\phi_{j(i)}\le \phi$, $\alpha',\lambda' \ge 1$, $\rho,\tau \le 1$, and $(R_i,\nu_{R_i})$ is a \smash{$(\phi_{j(i)},(1+\frac{1}{\log n})\alpha,(1+c)\lambda, c)$}-BSCW of~$G[U^{<i}]^\circ$. Then, for all $i \in [q]$ \cref{lem:algovols} gives
	\begin{equation}
		\label{eqclaim:bal}
		\unvolc{R_i}{U^{<i}} < \frac{2}{C\lambda}\unvolc{U}{U} < \frac{1}{5}\unvolc{U}{U}\, ,
	\end{equation}
	since $C \ge 11$.
	Furthermore, because every $(R_i,\nu_{R_i})$ is a \smash{$(\phi_{j(i)},(1+\frac{1}{\log n})\alpha,(1+c)\lambda, c)$}-BSCW of~$G[U^{<i}]^\circ$, we have by \cref{def:bscw}
	\begin{equation}
		\label{eqclaim:sparse}
		\unspc{R_i}{U^{<i}} < \left(1+\frac{1}{ \log n}\right) \alpha \phi_{j(i)} = \left(1+\frac{1}{ \log n}\right) \alpha \frac{\phi}{\left( \left(1+\frac{1}{ \log n}\right) \alpha\right)^{j(i)}}  \le \phi \, ,
	\end{equation}
	since $j(i) \ge 1$. The bounds in~\eqref{eqclaim:bal} and~\eqref{eqclaim:sparse} give part~\eqref{prec:balsparse}.
\end{proof}

\noindent
Finally, we argue that \cref{alg:trim} goes through few inner loop iterations within each outer loop iteration. The idea for this is that we can upper bound the total volume trimmed in line~\eqref{enum:trim}, but we can also lower bound the volume of each individual cut.

\begin{lemma}[\cref{alg:trim} does few inner loop iterations]
	\label{lem:boundinnerloop}
	Let $b, \phi,c \in (0,1)$ such that $c \leq 1/30$ is a constant, $\phi < b$ and $b \le c$. Also let $C \ge 31$ be a constant, let $\bsc$ be an $(\alpha,\lambda)$-BSCA where $\alpha \le \frac{1}{2b}$ and $\lambda \le \frac{c}{2b}$, and let $k$ be an integer such that \smash{$k \ge \log \frac{n^2}{\phi}/\log \frac{c}{\lambda \cdot b}$}. If \cref{assump:sparsifiers1} and \cref{assump:sparsifiers2} hold, then \cref{alg:trim} goes through at most $1/{b}$ inner loop iterations within each outer loop iteration.
\end{lemma}
\begin{proof}
	Let $U$ be the input cluster of \cref{alg:trim}. Fix hereafter an outer loop iteration $j=1,\dots, k+1$, and let $A$ be the value of the corresponding variable of the algorithm at the beginning of the $j$-th outer loop iteration. Because line~\eqref{enum:trim} of \cref{alg:trim} removes a non-empty cut from $A$ at every inner loop iteration, there are finitely many iterations. Denote then by $t$ the number of inner loop iterations for outer loop iteration $j$. Then, for every  inner loop iteration $h\in [t]$ define $A^{<h}$ to be the value of variable $A$ at the beginning of inner loop iteration $h$.
	
	Since \cref{assump:sparsifiers2} holds,  we also know that to every inner loop iteration $h \in [t]$ corresponds a graph $H=(V,E',w)$ from line~\eqref{enum:freshsampletrim} which satisfies $H[A^{<h}] \approx_{\delta_j} G[A^{<h}]$. As we have $c \le 1/30$, $\phi < b$, $b \le c$, $\alpha \le \frac{1}{2b}$,  and $H[A^{<h}] \approx_{\delta_j} G[A^{<h}]$, \cref{lem:reduction} ensures that calling $\bsc(H[A^{<h}]^\circ, (1+\frac{1}{2\log n}) \cdot \phi_j)$ gives a \smash{$(\phi_j,(1+\frac{1}{\log n})\alpha,(1+c)\lambda, c)$}-BSCW of~$G[A^{<h}]^\circ$ in line~\eqref{enum:trimcallbalsparse} of \cref{alg:trim} (see \cref{rmk:bsctrim}). Moreover, from the description of the algorithm, it must be the case that the BSCW obtained is of the form $(R_h,\nu_h)$ for all $h \in [t-1]$, where $\emptyset \neq R_h \subsetneq A^{<h}$ and $\nu_h \in \mathbb{R}_{> 0}$. One can then see that for every $h \in [t]$ we have $A^{<h}=A \setminus (\cup_{z=1}^{h-1}R_z)$, and because of conditions~\eqref{enum:relaxouter} and~\eqref{condition:baltrim} we have that for every $h \in [t-1]$ the following holds:
	\begin{equation*}
		\frac{1}{C\lambda}\left(\wvolc{U}{U}\right)^{1-(j-1)/k} \le \nu_h < \frac{1}{C\lambda}\wvolc{U}{U} \, .
	\end{equation*}
	We now want to translate the above bounds to $\unvolc{R_h}{A^{<h}}$ using \cref{lem:algovols}. To do so, note that $H[A^{<h}] \approx_{\delta_j} G[A^{<h}]$ in particular implies $H[A^{<h}] \approx_{c} G[A^{<h}]$ (since we have defined $\delta_j=c^2 b_j/\log n \le c$), and as discussed before we know that
	\begin{equation}
		\label{eq:bscw}
		\text{$(R_h,\nu_h)$ is a $\left(\phi_j,\left(1+\frac{1}{\log n}\right)\alpha,(1+c)\lambda,c\right)$-BSCW of~$G[A^{<h}]^\circ$}
	\end{equation}
	Let us then set $\psi=\phi_j$, $ \alpha'=(1+\frac{1}{\log n})\alpha$, $\lambda' =(1+c)\lambda$, $\rho = \frac{1}{C\lambda}$, $\tau$ being either $1-(j-1)/k$ (for the lower bound) or $1$ (for the upper bound), and $X=U$. These parameters verify the conditions for applying \cref{lem:algovols} to $(R_h,\nu_h)$ on $G[A^{<h}]^\circ$: $c \le 1/30 \le 1/4$, $\psi =\phi_j\le \phi$, $\alpha',\lambda' \ge 1$, $\rho \le 1$, both values of $\tau$ are bounded by $1$ (this is because $j \le k+1$), $H[A^{<h}] \approx_{c} G[A^{<h}]$ (since from \cref{assump:sparsifiers2} we have \smash{$H[A^{<h}] \approx_{\delta_j} G[A^{<h}]$}), and $(R_h,\nu_h)$ is a $(\phi_j,(1+\frac{1}{\log n})\alpha,(1+c)\lambda,c)$-BSCW of \smash{$G[A^{<h}]^\circ$} (see~\eqref{eq:bscw}). Then \cref{lem:algovols} gives
	\begin{equation}
		\label{sandwich:actual}
	\frac{1}{2C\lambda}\left(\unvolc{U}{U}\right)^{1-(j-1)/k} \le \unvolc{R_h}{A^{<h}} <\frac{2}{C\lambda}\unvolc{U}{U} \, ,
	\end{equation}
	where the lower bound is obtained with $\tau=1-(j-1)/k$ and the upper bound with $\tau=1$. The application of \cref{lem:algovols} with $\tau=1$ also gives that for every cut $\emptyset \neq  S \subsetneq A^{<h}$ with $\unspc{S}{A^{<h}} < \phi_j$ and $ \unvolc{S}{A^{<h}} \le \unvolc{A^{<h}\setminus S}{A^{<h}}$ one has
	\begin{equation}
		\label{eq:absolbal}
		\unvolc{S}{A^{<h}} < 2\lambda' \rho \cdot \unvolc{U}{U} = \frac{2(1+c)\lambda}{C\lambda}  \unvolc{U}{U} \le \frac{3}{C} \cdot \unvolc{U}{U} \, .
	\end{equation}
	We consider now the case of $j=1$ and $2 \le j \le k+1$ separately.
	
	\paragraph{Case $j=1$.} In this case $A$ is identical to the input cluster $U$.	We then want to apply Lemma~\ref{lem:seqbounds} with $\delta=1$ and $\theta = \frac{3}{C}\unvolc{U}{U}$ to the sequence of cuts $R_1,R_2, \dots, R_{t-1}$ in the cluster~$U$. Let for convenience $U^i = U \setminus (\cup_{z =1}^{i-1} R_z)$ for all $i \in [t]$.
	
	First note that this value of $\theta$ meets the requirement of \cref{lem:seqbounds} that $\theta< \frac{1}{5}\unvolc{U}{U}$, since $C \ge 16$. Secondly, our set of assumptions matches that of \cref{lem:boundsequencetrim}, so we get that for all \smash{$\emptyset \neq  S  \subsetneq U$} with \smash{$\unspc{S}{U} < \phi$, $ \unvolc{S}{U} \le \unvolc{U \setminus S}{U}$}, one has \smash{$\unvolc{S}{U}<\frac{3}{C}\unvolc{U}{U}$}, and for all $i \in [q]$ one has \smash{$\unspc{R_i}{U^{<i}} < \phi$} and $ \unvolc{R_i}{U^{<i}} < \frac{1}{5}\unvolc{U}{U}$. These facts match preconditions~~\eqref{prec:sparse} and~\eqref{prec:balsparse} of \cref{lem:seqbounds} for $\delta$ and $\theta$ as we set them above.
	
	Part~\eqref{part:1} of \cref{lem:seqbounds} then upper-bounds the total volume of the sequence as
	\begin{equation}
		\label{eq:totalubjone}
		\unvolc{\cup_{h=1}^{t-1} R_h}{U} < \frac{3}{C}\unvolc{U}{U} \, .
	\end{equation}
	Since the lower bound of~\eqref{eq:bscw} with $j=1$ means that $ \unvolc{R_h}{U^{<h}} \ge \frac{1}{2C\lambda}\unvolc{U}{U}$ for all $h \in [t-1]$, using part~\eqref{part:2b} of \cref{lem:seqbounds} with $\theta^{\low} =\frac{1}{2C\lambda}\unvolc{U}{U}$ we further get
		\begin{equation}
		\label{eq:indivlbjone}
		\unvolc{R_h}{U} \ge \frac{1}{2C\lambda}\unvolc{U}{U} - b\cdot \frac{3}{C}\unvolc{U}{U}  \ge  \frac{1}{3C\lambda}\unvolc{U}{U} \quad \text{ for all $h \in [t-1]$}\, ,
	\end{equation}
since  $\lambda \le \frac{1}{18b}$.
	Therefore, taking the ratio of the cumulative upper bound in~\eqref{eq:totalubjone} and the individual lower bounds from~\eqref{eq:indivlbjone}, we conclude that $t-1$ cannot exceed
	\begin{equation*}
		\frac{\unvolc{\cup_{h=1}^{t-1} R_h}{U}}{\min_{h \in [t-1]}\unvolc{R_h}{U}} \le \frac{\frac{3}{C}\unvolc{U}{U}}{\frac{1}{3C\lambda}\unvolc{U}{U}} \le 9 \lambda \, .
	\end{equation*}
Thus, the number of inner loop iterations for $j=1$ is at most $9\lambda +1 \le \frac{1}{2b}+1 \le 1/b$.
	
	\paragraph{Case $2 \le j \le k+1$.} The approach is analogous to the first iteration, except that we use the bounds from \cref{lem:seqbounds} on the sequence $R_1,\dots,R_{t-1}$ in the cluster $A$ (for $j=1$ we applied the lemma to the cluster $U$). We remark that, in this case, $A$ can also be seen as the value of the corresponding variable of \cref{alg:trim} at the very end of outer loop iteration $j-1$.
	
	Consider the call $\bsc(H[A]^\circ, (1+\frac{1}{2\log n}) \cdot \phi_{j-1})$ to $\bsc$ in the last inner loop iteration within outer loop iteration $j-1$, where $H=(V,E',w)$ is accordingly the graph obtained from line~\eqref{enum:freshsampletrim} at the last inner loop iteration within outer loop iteration~${j-1}$. Since \cref{assump:sparsifiers2} holds,  we know that the graph $H$ satisfies \smash{$H[A] \approx_{\delta_{j-1}} G[A]$}. As we have $c \le 1/30$, $\phi < b$, $b \le c$, $\alpha \le \frac{1}{2b}$,  and \smash{$H[A] \approx_{\delta_{j-1}} G[A]$}, \cref{lem:reduction} ensures that the call \smash{$\bsc(H[A]^\circ, (1+\frac{1}{2\log n}) \cdot \phi_{j-1})$} gives a \smash{$(\phi_{j-1},(1+\frac{1}{\log n})\alpha,(1+c)\lambda,c)$}-BSCW $\omega$ of $G[A]^\circ$ (see \cref{rmk:bsctrim}). On the other hand, from the description of the algorithm we know that iteration $j$ starts only if this $\omega$ is of the form $(R,\nu)$ with $\nu < \frac{1}{C \lambda}(\wvolc{U}{U})^{1-(j-2)/k}$ (due to condition~\eqref{condition:baltrim}). We now use \cref{lem:algovols} with parameters $\psi=\phi_{j-1}$, \smash{$ \alpha'=(1+\frac{1}{\log n})\alpha$}, $\lambda' =(1+c)\lambda$, $\rho = \frac{1}{C\lambda}$, $\tau=1-(j-2)/k$, and $X=U$. These parameters verify the requirement of \cref{lem:algovols} since we have $ c\le 1/30 \le 1/4$, $H[A] \approx_{c} G[A]$ (from \cref{assump:sparsifiers2} and $\delta_{j-1} \le c$), $\psi=\phi_{j-1}\le \phi$, $\alpha',\lambda' \ge 1$, $\rho,\tau \le 1$, and $(R,\nu)$ is a \smash{$(\phi_{j-1},(1+\frac{1}{\log n})\alpha,(1+c),c\lambda)$}-BSCW of~$G[A]^\circ$. Then we get that for every cut $\emptyset \neq  S \subsetneq A$ with $\unspc{S}{A} < \phi_{j-1}$ and $ \unvolc{S}{A} \le \unvolc{A \setminus S}{A}$ one has
	\begin{equation}
		\label{eq:generalcasenobsc}
		\unvolc{S}{A} < 2\lambda' \rho \left(\unvolc{U}{U}\right)^\tau =  \frac{2(1+c)\lambda}{C\lambda} \left( \unvolc{U}{U}\right)^{1-(j-2)/k} \le \frac{3}{C} \cdot \left(\unvolc{U}{U}\right)^{1-(j-2)/k} \, .
	\end{equation}
	This is still not enough to apply \cref{lem:seqbounds} to cluster $A$: the bound in~\eqref{sandwich:actual} bounds the volumes of nested cuts $R_h$ in $A^{<h}$ by $\frac{2}{C\lambda}\unvolc{U}{U}$, but applying \cref{lem:seqbounds} demands an upper bound of $\frac{1}{5}\unvolc{A}{A}$ on the volumes of cuts $R_h$ in $A^{<h}$. Hence, we first show the volume of $U$ to be not much larger than that of $A$.
	
	\begin{claim}
		\label{claim:voluvola}
		We have $\unvolc{U}{U} \le 2 \unvolc{A}{A}$.
	\end{claim}
	\begin{proof}
		If $A=U$ the statement is trivial. Otherwise, we know there must have been a sequence of nested cuts $S_1,\dots,S_r \subseteq U$ that \cref{alg:trim} trimmed off from $U$ in line~\eqref{enum:trim} so that $A=U \setminus (\cup_{i \in [r]} S_i)$. We then want to apply Lemma~\ref{lem:seqbounds} to the sequence of cuts $S_1,S_2, \dots, S_r$ in the cluster~$U$ with $\delta=1$ and $\theta = \frac{3}{C}\unvolc{U}{U}$. Let for convenience $U^{<i} = U \setminus (\cup_{z =1}^{i-1} S_z)$ for all $i \in [r+1]$.
		
		First note that this value of $\theta$ meets the requirement of \cref{lem:seqbounds} that $\theta< \frac{1}{5}\unvolc{U}{U}$, since $C \ge 16$. Secondly, our set of assumptions matches that of \cref{lem:boundsequencetrim}, so we get that for all \smash{$\emptyset \neq  S  \subsetneq U$} with \smash{$\unspc{S}{U} < \phi$, $ \unvolc{S}{U} \le \unvolc{U \setminus S}{U}$}, one has \smash{$\unvolc{S}{U}<\frac{3}{C}\unvolc{U}{U}$}, and for all $i \in [r]$ one has \smash{$\unspc{S_i}{U^{<i} } < \phi$} and $ \unvolc{S_i}{U^{<i}} < \frac{1}{5}\unvolc{U}{U}$ (the latter follows because part~\eqref{satprec:balsparse} of \cref{lem:boundsequencetrim} bounds the sparsity and volume of every cut in the entire sequence of cuts trimmed by the algorithm, and $S_1,\dots,S_r$ is a prefix of the entire sequence). These facts match preconditions~~\eqref{prec:sparse} and~\eqref{prec:balsparse} of \cref{lem:seqbounds} for $\delta$ and $\theta$ as set above. Part~\eqref{part:1} of \cref{lem:seqbounds} then upper-bounds the total volume of the sequence as
		\begin{equation}
			\label{bound:claim}
			\unvolc{\cup_{i=1}^{r} S_i}{U} < \frac{3}{C}\unvolc{U}{U} \, .
		\end{equation}
		To conclude it is sufficient to notice
		\begin{align*}
			\unvolc{U}{U} & = \unvolc{A}{U}+\unvolc{U \setminus A}{U}\\
			& = \unvolc{A}{U}+\unvolc{\cup_{i =1}^r S_i}{U} \\
			\text{$\unvolc{A}{A}$ cannot have less self-loops than $\unvolc{A}{U}$}\quad& \le \unvolc{A}{A}+\unvolc{\cup_{i =1}^r S_i}{U} \\
			\text{by~\eqref{bound:claim}}\quad & < \unvolc{A}{A}+\frac{3}{C}\unvolc{U}{U} \, ,
		\end{align*}
	and the claim follows since $C \ge 6$.
	\end{proof}
	
	\noindent
	By \cref{claim:voluvola} and~\eqref{sandwich:actual}, we have that for every $h \in [t-1]$
	\begin{equation}
		\label{eq:ubwrta}
		\unvolc{R_h}{A^{<h}} < \frac{2}{C\lambda}\unvolc{U}{U} \le \frac{4}{C\lambda}\unvolc{A}{A} < \frac{1}{5}\unvolc{A}{A} \, ,
	\end{equation}
	since $C \ge 21$. Furthermore, because every $(R_h,\nu_h)$ is a $(\phi_j,(1+\frac{1}{\log n})\alpha,(1+c)\lambda,c)$-BSCW of~\smash{$G[A^{<h}]^\circ$}, we have
	\begin{equation}
		\label{eq:sparse}
		\unspc{R_h}{A^{<h}} < \left(1+\frac{1}{ \log n}\right) \alpha \phi_{j} = \left(1+\frac{1}{ \log n}\right) \alpha \frac{\phi}{\left( \left(1+\frac{1}{ \log n}\right) \alpha\right)^{j}}  \le \phi_{j-1} \, .
	\end{equation}
	We are now in shape to apply Lemma~\ref{lem:seqbounds} to the sequence of cuts $R_1,R_2, \dots, R_{t+1}$ in the cluster~$A$. Specifically, we do this with $\delta=\phi_{j-1}/\phi$ and $\theta = \frac{3}{C}(\unvolc{U}{U})^{1-(j-2)/k}$. First observe that $\delta \le 1$ (by definition of $\phi_{j-1}$) and $ \theta \le  \frac{3}{C}\unvolc{U}{U} \le \frac{1}{5}\unvolc{A}{A}$ (again by \cref{claim:voluvola} and $C \ge 31$), as demanded by \cref{lem:seqbounds}. For these parameters, one can see that precondition~\eqref{prec:balsparse} is given by~\eqref{eq:ubwrta} and~\eqref{eq:sparse}. Precondition~\eqref{prec:sparse} is instead given by~\eqref{eq:generalcasenobsc}. Part~\eqref{part:1} of \cref{lem:seqbounds} then upper-bounds the total volume of the sequence as
	\begin{equation}
		\label{eq:ubgeneraliter}
		\unvolc{\cup_{h=1}^{t-1} R_h}{A} < \theta =  \frac{3}{C}\left(\unvolc{U}{U}\right)^{1-(j-2)/k} \, .
	\end{equation}
	We are left with lower bounding the volume of each individual cut. We do this with part~\eqref{part:2b} of \cref{lem:seqbounds} with $\theta^{\low} = \frac{1}{2C\lambda}(\unvolc{U}{U})^{1-(j-1)/k}$, which meets the requirement that $	\unvolc{R_h}{A} \ge \theta^{\low}$ thanks to the bound in~\eqref{sandwich:actual}. Part~\eqref{part:2b} of \cref{lem:seqbounds} then gives that for all $h \in [t-1]$
	\begin{align}
			\label{eq:lbgeneraliter}
		\unvolc{R_h}{A} \ge \theta^{\low} - b \cdot \theta = \frac{1}{2C\lambda}\left(\unvolc{U}{U}\right)^{1-(j-1)/k} - b \cdot  \frac{3}{C}\left(\unvolc{U}{U}\right)^{1-(j-2)/k}\, .
	\end{align}
	Therefore, taking the ratio of the upper bound on the total volume from~\eqref{eq:ubgeneraliter} and the lower bound on the volume of each individual cut from~\eqref{eq:lbgeneraliter}, we conclude that $t-1$ is at most
	\begin{align*}
		\frac{\unvolc{\cup_{h=1}^{t-1} R_h}{A}}{\min_{h \in [t-1]} \unvolc{R_h}{A}} & \le
		\frac{\frac{3}{C}\left(\unvolc{U}{U}\right)^{1-(j-2)/k}}{\frac{1}{2C\lambda}\left(\unvolc{U}{U}\right)^{1-(j-1)/k} - b \cdot \frac{3}{C}\left(\unvolc{U}{U}\right)^{1-(j-2)/k}} \\
		& = \frac{\frac{3}{C}\left(\unvolc{U}{U}\right)^{1/k}}{\frac{1}{2C\lambda} - b \cdot \frac{3}{C}\left(\unvolc{U}{U}\right)^{1/k}} \, .
	\end{align*}
	Next, we observe that $\unvolc{U}{U} \le \vol(V)\cdot b/\phi \le \vol(V)/{\phi} \le n^2/\phi$ (since in the worst case we over-count every edge in $G$ as a self-loop in $G[U]^{\circ}$ and  at most $b/\phi$ times). Then, using that \smash{$k \ge \log (\frac{n^2}{\phi}) / \log (\frac{c}{\lambda \cdot b})$}, it holds that \smash{$\left(\unvolc{U}{U}\right)^{1/k} \leq \frac{c}{\lambda \cdot b}$}. Now we can upper-bound the above as
	\begin{align*}
		\frac{\frac{3}{C}\left(\unvolc{U}{U}\right)^{1/k}}{\frac{1}{2C\lambda} - b \cdot \frac{3}{C}\left(\unvolc{U}{U}\right)^{1/k}}  \le \frac{\frac{3}{C}\left(\unvolc{U}{U}\right)^{1/k}}{\frac{1}{2C\lambda} - b \cdot \frac{3}{C}\frac{c}{\lambda \cdot b}} \le \frac{\frac{3}{C} \cdot \left(\frac{n^2}{\phi}\right)^{1/k}}{\frac{1}{5C \lambda} } \le \frac{\frac{3}{C} \cdot \frac{c}{\lambda \cdot b}}{\frac{1}{5C \lambda}} = 15 \frac{c}{b}\, ,
	\end{align*}
	since $c \le \frac{1}{30}$.	Hence, the number of inner loop iterations for $2 \le j \le k+1$ is at most $ \frac{15c}{b} +1\le \frac{1}{b}$.
\end{proof}

\subsubsection{Bounding the recursion depth}
\label{subsubsec:depth}
\noindent
From the description of the algorithms one can see that we recurse on both sides of a cut (in line~\eqref{rec:twosides} of \cref{alg:expdec}) only when the conditions~\eqref{enum:balrecurse} in \cref{alg:expdec} and~\eqref{enum:relaxouter} in \cref{alg:trim} suggest that the cut is balanced, so as to maintain the recursion depth small. When instead \cref{alg:expdec} recurses on one side only, in line~\eqref{rec:onside}, we can use the properties of nested cuts to upper bound the volume of the recursed side. To control the number of self-loops that we add after each recursion, we also want to ensure that we recurse on sparse cuts only.

\begin{lemma}[\cref{alg:expdec} recurses on balanced sparse cuts]
	\label{lem:recursioncut}
	Let $b, \phi,c \in (0,1)$ such that $c \le 1/30$ is a constant, $\phi < b$, $b \le c$. Also  let $C \ge 16$ be a constant, and let $\bsc$ be an $(\alpha,\lambda)$-BSCA where $\alpha \le \frac{1}{2b}$ and $\lambda \le \frac{1}{18b}$. If \cref{assump:sparsifiers1} and \cref{assump:sparsifiers2} hold, whenever Algorithm~\ref{alg:expdec} with input cluster $U \subseteq V$ makes a recursive call, the recursion is on a cluster $\emptyset \neq X \subsetneq U$ such that
	\begin{equation*}
		\unspc{X}{U} < (3C\lambda+2\alpha)\cdot \phi \quad \text{ and } \quad \unvolc{X}{U} \le \left(1-\frac{1}{3C\lambda}\right)\unvolc{U}{U} \, .
	\end{equation*}
\end{lemma}
\begin{proof}
	There are three lines where a recursive call can happen:~\eqref{enum:balrecurse},~\eqref{rec:onside},~\eqref{rec:twosides}. We analyse them separately.
	
	\paragraph{Recursion in line~\eqref{enum:balrecurse}.} Let $H=(V,E',w)$ be the graph returned in line~\eqref{enum:freshsampleed} of \cref{alg:expdec}. Since the algorithm recurses in line~\eqref{enum:balrecurse}, the call $ \bsc(H[U]^\circ, (1+\frac{1}{2 \log n}) \cdot \phi)$ of line~\eqref{enum:balsparsecall} must have returned a BSCW of the form $(R,\nu)$ such that
	\begin{equation}
		\label{eq:lbr}
		\nu \ge \frac{1}{C \lambda}\wvolc{U}{U} \, .
	\end{equation}
	As we have $c \le 1/30$, $\phi < b$, $b \le c$, $\alpha \le \frac{1}{2b}$, and $H[U] \approx_\delta G[U]$ (because of \cref{assump:sparsifiers1}),  Lemma~\ref{lem:reduction} gives that $(R,\nu)$ is a  \smash{$(\phi,(1+\frac{1}{\log n})\alpha,(1+c)\lambda,c)$}-BSCW of~$G[U]^\circ$. By \cref{def:bscw}, this means
	\begin{equation*}
		\unspc{R}{U} < \left(1+\frac{1}{\log n}\right)\alpha\cdot \phi \le 2\alpha \cdot \phi \, .
	\end{equation*}
	We now want to use Lemma~\ref{lem:algovols} with $\psi = \phi$, $ \alpha'=(1+\frac{1}{\log n})\alpha$, $\lambda' =(1+c)\lambda$, $\rho= \frac{1}{C\lambda}$, $\tau~=~1$, and $X=U$. These parameters verify the conditions for applying \cref{lem:algovols} to $(R,\nu)$ on $G[U]^\circ$: $c \le 1/30 \le 1/4$, $\psi \le \phi$, $\alpha',\lambda' \ge 1$, $\rho \le 1$, $\tau \le 1$, $H[U] \approx_{c} G[U]$ (since $H[U] \approx_{\delta} G[U]$ and $\delta \le c$), and $(R,\nu)$ is a \smash{$(\phi,(1+\frac{1}{\log n},c)\alpha,(1+c)\lambda)$}-BSCW of~$G[U]^\circ$. Then \cref{lem:algovols} with~\eqref{eq:lbr} gives
	\begin{equation*}
		\unvolc{R}{U} \geq  \frac{\rho}{2} (\unvolc{U}{U})^\tau = \frac{1}{2C\lambda} \unvolc{U}{U}\, ,
	\end{equation*}
	which we can use to upper bound the volume of $U \setminus R$:
	\begin{equation*}
		\unvolc{U \setminus R}{U} = \unvolc{U}{U} - \unvolc{R}{U} \leq \left( 1 -  \frac{1}{2C\lambda} \right) \unvolc{U}{U}.
	\end{equation*}
	Finally, by property \eqref{item:smallsider} of BSCW applied to $R$, $\unvolc{R}{U} \leq (1 + c) \unvolc{U \setminus R}{U}$. Because $c \leq 1/30 \le 1/2$ and $C \ge 2$,
	\begin{equation*}
		\unvolc{R}{U} \leq \frac{1+c}{2} \unvolc{U}{U} \leq \left( 1 -  \frac{1}{2C\lambda} \right) \unvolc{U}{U}.
	\end{equation*}
	Since in this case the algorithm recurses on $R$ and $U\setminus R$, we have the claim.
	
	\paragraph{Recursion in lines~\eqref{rec:onside} and~\eqref{rec:twosides}.} These cases can only happen if \cref{alg:expdec} resorts to \cref{alg:trim}. Let then $R_1,\dots, R_q \subseteq U$ be the sequence of cuts trimmed off by \cref{alg:trim} in line~\eqref{enum:gotrim}. We then want to apply Lemma~\ref{lem:seqbounds} to the sequence of cuts $R_1,R_2, \dots, R_{q}$ in the cluster~$U$ with $\delta=1$ and $\theta = \frac{3}{C}\unvolc{U}{U}$. Let for convenience $U^{<i} = U \setminus (\cup_{z =1}^{i-1} R_z)$ for all $i \in [q+1]$.
	
	First note that this value of $\theta$ meets the requirement of \cref{lem:seqbounds} that \smash{$\theta< \frac{1}{5}\unvolc{U}{U}$}, since $C \ge 16$. Secondly, our set of assumptions matches that of \cref{lem:boundsequencetrim}, so we get that for all \smash{$\emptyset \neq  S  \subsetneq U$} with \smash{$\unspc{S}{U} < \phi$, $ \unvolc{S}{U} \le \unvolc{U \setminus S}{U}$}, one has \smash{$\unvolc{S}{U}<\frac{3}{C}\unvolc{U}{U}$}, and for all $i \in [q]$ one has \smash{$\unspc{R_i}{U^{<i}} < \phi$} and $ \unvolc{R_i}{U^{<i}} < \frac{1}{5}\unvolc{U}{U}$. These facts match preconditions~~\eqref{prec:sparse} and~\eqref{prec:balsparse} of \cref{lem:seqbounds} for $\delta$ and $\theta$ as we set them above.
	\begin{itemize}
		\item \emph{Line~\eqref{rec:onside}}. Part~\eqref{part:1} of \cref{lem:seqbounds} then upper-bounds the total sparsity and volume of the sequence as
		\begin{equation}
			\label{bound:recursion}
			\unspc{\cup_{z=1}^{q} R_z}{U} < \phi \quad \text{and} \quad \unvolc{\cup_{z=1}^{q} R_z}{U} < \frac{3}{C}\unvolc{U}{U} \, .
		\end{equation}
		For \cref{alg:expdec} to recurse in line~\eqref{rec:onside}, it must be the case that \cref{alg:trim} returned a pair $(S,\top)$ in line~\eqref{enum:gotrim} (see condition~\eqref{condition:exp}). Looking into \cref{alg:trim}, one can then see that if it returns a pair of the form $(S,\top)$, it must be the case that $S$ is the value of variable $A$ at the end of the algorithm, i.e. $S=A=U \setminus (\cup_{z \in [q]} R_z)$. In line~\eqref{rec:onside}, the recursive call is $\ed(U \setminus S,\ell+1)$, so we recurse on $U \setminus S = \cup_{z \in [q]} R_z$.  From~\eqref{bound:recursion} we then have the claim since $C \ge 6$.
		
		\item \emph{Line~\eqref{rec:twosides}}. The difference from the previous case is that for \cref{alg:expdec} to recurse in line~\eqref{rec:twosides}, it must be the case that \cref{alg:trim} returned a pair $(S,\bot)$ in line~\eqref{enum:gotrim} (see condition~\eqref{condition:exp}), and in particular this implies that \cref{alg:trim} returned in line~\eqref{enum:relaxouter}.
		
		Let $p$ be the outer loop iteration where \cref{alg:trim} returns, and let $H=(V,E',w)$ be the graph from line~\eqref{enum:freshsampletrim} in the last inner loop iteration of outer loop iteration $p$ (i.e. the inner loop iteration where it returns). Since in such iteration we must have that variable $A$ actually equals $U^{<q+1} $, let $\omega$ be the result of $\bsc(H[U^{q+1}]^\circ, (1+\frac{1}{2 \log n})\phi_{p})$ in that same iteration.
		As we have $ c \le 1/30$, $\phi < b$, $b \le c$, \smash{$\alpha \le \frac{1}{2b}$}, and \smash{$H[U^{<q+1}] \approx_{\delta_p} G[U^{<q+1}]$} (from \cref{assump:sparsifiers2}), \cref{lem:reduction} ensures that $\omega$ is a \smash{$(\phi_p,(1+\frac{1}{\log n})\alpha,(1+c)\lambda,c)$}-BSCW of~$G[U^{<q+1}]^\circ$ (see \cref{rmk:bsctrim}). Moreover, since \cref{alg:trim} returned in line~\eqref{enum:relaxouter}, it must be the case that $\omega$ is of the form $(S,\nu)$, where $S$ is the cut returned from \cref{alg:trim} to \cref{alg:expdec} in the pair $(S,\bot)$. From the description of the algorithm it must also be the case that
		\begin{equation}
			\label{bound:tworeclb}
			\nu \ge \frac{1}{C\lambda} \wvolc{U}{U} \, .
		\end{equation}
		As we have $c \le 1/30 \le 1/4$, $\phi < b$, $b \le c$, \smash{$\alpha \le \frac{1}{2b}$}, and \smash{$H[U^{<q+1}] \approx_{c} G[U^{<q+1}]$} (from \cref{assump:sparsifiers2} and $\delta_p \le c$), the fact that $(S,\nu)$ is a \smash{$(\phi_p,(1+\frac{1}{\log n})\alpha,(1+c)\lambda,c)$}-BSCW of~$G[U^{<q+1}]^\circ$, and the bound~\eqref{bound:tworeclb}, \cref{lem:algovols} gives
		\begin{equation}
			\label{bound:tworeclbactual}
			\unvolc{S}{U^{<q+1}} \ge \frac{1}{2C\lambda} \unvolc{U}{U} \, .
		\end{equation}
		Since $(S,\nu)$ is a \smash{$(\phi_p,(1+\frac{1}{\log n})\alpha,(1+c)\lambda,c)$}-BSCW of~$G[U^{<q+1}]^\circ$, by \cref{def:bscw} we also know
		\begin{equation*}
			\unvolc{S}{U^{<q+1}} \le (1+c) \unvolc{U \setminus S}{U^{<q+1}} \, ,
		\end{equation*}
		so
		\begin{equation}
			\label{bound:notmuchmorethanhalf}
			\unvolc{S}{U^{<q+1}} \le \frac{5}{9}\unvolc{U}{U^{<q+1}} \, ,
		\end{equation}
	since $c \le 1/30 \le  1/10$. Now, setting $R=S$ and $\theta^{\low}=\frac{1}{2C\lambda} \unvolc{U}{U}$, the bounds in~\eqref{bound:tworeclbactual} and~\eqref{bound:notmuchmorethanhalf}, together with the fact that $\theta^{\low}  \le \frac{4}{9}\unvolc{U}{U}$ (since $C \ge 2$), allow us to apply part~\eqref{part:2c} of \cref{lem:seqbounds}. This gives
	\begin{align}
		\unspc{S}{U} & < \phi \frac{\unvolc{U}{U}}{\theta^{\low}-b\cdot \theta} \\
		& = \phi \cdot \frac{\unvolc{U}{U}}{\frac{1}{2C\lambda}\unvolc{U}{U}- b \cdot \frac{3}{C}\unvolc{U}{U}} \\
		\text{since $b \le \frac{1}{18 \lambda}$}\quad & \le 3C\lambda \cdot \phi \label{bound:sparse} \, , 
	\end{align}
	and
	\begin{equation}
		\frac{1}{3C\lambda}\unvolc{U}{U} \le \theta^{\low}-b\cdot \theta \le \unvolc{S}{U} \le \frac{5}{9} \unvolc{U}{U} \, .  \label{bound:bal}
	\end{equation}
	We can now conclude: since the result of \cref{alg:trim} is $(S,\bot)$, \cref{alg:expdec} recurses on both~$S$ and $U \setminus S$ in line~\eqref{rec:twosides}. Let $X \in \{S,U\setminus S\}$. Then from~\eqref{bound:sparse} we have that $\unspc{X}{U} < 3C\lambda \cdot \phi$, and from~\eqref{bound:bal} we have that $\unvolc{X}{U} \le (1-\frac{1}{3C\lambda}) \unvolc{U}{U}$ since $C \ge 2$.
	
	\end{itemize}
\end{proof}

\noindent
We now upper-bound the recursion depth of the algorithm, using the above lemma.

\begin{lemma}[\cref{alg:expdec} has small recursion depth]
	\label{lem:bounddepth}
	Let $b, \phi,c \in (0,1)$ such that $c \le 1/30$ is a constant, $\phi < b$, $b \le c$. Also let $C \ge 16$ be a constant, and let $\bsc$ be an $(\alpha,\lambda)$-BSCA where \smash{$\alpha \le \frac{1}{144C\lambda b}$, $\lambda \le\min\{ \frac{1}{18b}, \, \frac{1}{15C \sqrt{b}}\}$}. If \cref{assump:sparsifiers1} and \cref{assump:sparsifiers2} hold, then the recursion depth of $\ed(V,0)$ (i.e. \cref{alg:expdec}) is at most $9C\lambda \cdot \log n $.
\end{lemma}

\begin{proof}
	The crux is that recursion only happens on one of the sides of a not too large cut. As our set of assumptions matches that of \cref{lem:recursioncut} (note that the preconditions are the same except \smash{$\alpha \le \frac{1}{144C\lambda b} \le \frac{1}{2b}$ and $\lambda \le\min\{ \frac{1}{18b}, \, \frac{1}{15C \sqrt{b}}\}\le \frac{1}{18b}$}), we get the following: for any recursion from $U$ to $X$
	\begin{equation}
		\label{bound:depth}
		\unspc{X}{U} < (2\alpha+3C\lambda)\phi \quad \text{ and } \quad \unvolc{X}{U} \le \left(1-\frac{1}{3C\lambda}\right) \unvolc{U}{U}
	\end{equation}
	To get a bound on the recursion depth, we translate these bounds to a bound on $\unvolc{X}{X}$.
	
	We use \cref{lem:relationsvol} with $\delta=3C\lambda+2\alpha$, $S=X$, and $T=X$. These setting verifies the requirements of \cref{lem:relationsvol}: $\delta \le 1/b$ (since $\alpha \le \frac{1}{4b}$, and $\lambda \leq \frac{1}{6Cb}$), $\unspc{S}{U} < \delta \phi$ from~\eqref{bound:depth}, and $T \subseteq S$. Hence, inequality (2) from \cref{lem:relationsvol} gives
	\begin{align*}
		\unvolc{X}{X}& = \unvolc{T}{S} \\
		\text{inequality (2) of~\cref{lem:relationsvol}}\quad &< \unvolc{T}{U} + \delta b \cdot \min\{\unvolc{S}{U}, \unvolc{U \setminus S}{U}\} \\
		& \le  (1+\delta b) \cdot \unvolc{X}{U} \\
		\text{from~\eqref{bound:depth}}\quad & \le \left(1+\delta b\right) \left(1-\frac{1}{3C\lambda}\right) \unvolc{U}{U} \\
		\text{since $\alpha \le \frac{1}{144C\lambda b}$ and $\lambda \le \frac{1}{15C \sqrt{b}}$}\quad & \le \left(1+\frac{1}{36C\lambda}\right) \left(1-\frac{1}{3C\lambda}\right) \unvolc{U}{U} \\
		& \le  \left(1-\frac{1}{4C\lambda}\right) \unvolc{U}{U} \, .
	\end{align*}
	From this we conclude that if we do a recursion on $X$ at level $\ell$, then
	\begin{equation*}
	\unvolc{X}{X} \le  \left(1-\frac{1}{4C\lambda}\right)^\ell \unvolc{V}{V} =  \left(1-\frac{1}{4C\lambda}\right)^\ell \vol(V) <  \left(1-\frac{1}{4C\lambda}\right)^\ell n^2 \, .
	\end{equation*}
	The bound follows.
\end{proof}

\subsubsection{Correctness}
\label{subsubsec:correct}

Finally, one can see from the algorithms that they recurse only on the sides of cuts resulting from $\bsc$. We know such cuts to be sparse, so intuitively there should be few inter-cluster edges in the resulting partition of $V$. We also know that when the algorithm returns a cluster $\{U\}$, it is because the corresponding BSCW declares it to be an expander. One can then prove that the result of this process is a BLD.

\begin{lemma}[\cref{alg:expdec} outputs a BLD]
	\label{lem:bld}
	Let $b, \phi,c \in (0,1)$ such that $c \le 1/30$ is a constant, $\phi < b$, $b \le c$. Also let $C \ge 31$ be a constant, let $\bsc$ be an $(\alpha,\lambda)$-BSCA where \smash{$\alpha \le \frac{1}{144C\lambda b}$, $\lambda \le\min\{ \frac{1}{18b}, \, \frac{1}{15C \sqrt{b}}\}$}, and let $k$ be an integer such that $k \le \log n$. If \cref{assump:sparsifiers1} and \cref{assump:sparsifiers2} hold, then the result of $\ed(V,0)$ (i.e. \cref{alg:expdec}) 
	is a $(b,\epsilon,\phi,\gamma)$-BLD of $G$ where
	\begin{equation*}
		\epsilon = 4 (3C\lambda+2\alpha) \cdot \phi \cdot 9C\lambda \cdot \log n \cdot  \e^{2b (3C\lambda+2\alpha) 9C\lambda \cdot \log n} \quad \text{ and } \quad \gamma = 6 \alpha^{k+1} \, .
	\end{equation*}
\end{lemma}

\begin{proof}
	First note that our assumptions subsume the preconditions of \cref{lem:boundouterloop}. Hence, if Algorithm~\ref{alg:trim} breaks in line~\eqref{enum:break} it is in some outer loop iteration $j$ with $j \le k$, so the behaviour of the algorithm is well defined (see the outer loop for $j=1,\dots,k+1$ in line~\eqref{enum:outerloop} of \cref{alg:trim}).
	
	From the description we then see that $\ed(V,0)$ always outputs a valid partition $\mathcal{U}$ of~$V$. We show that $\mathcal{U}$ verifies the two properties of Definition~\ref{def:expdec}.
	
	\paragraph{Property~\eqref{property:expander}.} Let $X \in \mathcal{U}$, which can either be the result of the return statement in line~\eqref{enum:stoprec} or in line~\eqref{rec:onside} of Algorithm~\ref{alg:expdec}.
	\begin{itemize}
		\item \emph{Return statement in line~\eqref{enum:stoprec}}. Let $U \subseteq V$ be the input cluster to the  instance of \cref{alg:trim} that returned $X$. Also let $H=(V,E',w)$ be the graph from line~\eqref{enum:freshsampleed} and $\omega$ be the result of $\bsc(H[U]^\circ, (1+\frac{1}{2 \log n}) \cdot \phi)$ in line~\eqref{enum:balsparsecall} of that same instance. From the description of the algorithm, we know that the return statement in line~\eqref{enum:stoprec} only occurs if $\omega=\bot$, and in particular $X$ is the input cluster $U$.
		
		As our parameter regime implies $c \le 1/30$, $\phi < b$, $b \le c$, $\alpha \le \frac{1}{2b}$, and also we have $H[U] \approx_\delta G[U]$ from  \cref{assump:sparsifiers1}, we can use Lemma~\ref{lem:reduction}. This ensures that $\omega$ is a $(\phi,(1+\frac{1}{\log n})\alpha,(1+c)\lambda,c)$-BSCW of~$G[U]^\circ$. From \cref{def:bscw}, $\omega=\bot$ implies that $G[U]^\circ$ is a $\phi$-expander.
		
		\item \emph{Return statement in line~\eqref{rec:onside}}.  In this case, \cref{alg:expdec} must have called \cref{alg:trim}, which in turn returned a pair $(S,\top)$. Then, \cref{alg:trim} must have returned in line~\eqref{enum:expfound}. Denoting by~$A$ the value of the corresponding variable in \cref{alg:trim} when it returns, we must have $A=X$. Also let  $p \in [k+1]$ be the outer loop iteration where \cref{alg:trim} returns, let $H=(V,E',w)$ be the graph from line~\eqref{enum:freshsampletrim}, and let $\omega$ be the result of $\bsc(H[A]^\circ, (1+\frac{1}{2 \log n}) \cdot \phi_p)$ in line~\eqref{enum:balsparsecall} of the inner loop iteration where the algorithm returns.
		
		As our parameter regime implies $c \le 1/30$, $\phi < b$, $b \le c$, $\alpha \le \frac{1}{2b}$, and also we have $H[A] \approx_{\delta_p} G[A]$ from  \cref{assump:sparsifiers2}, we can use Lemma~\ref{lem:reduction}. This ensures that $\omega$ is a $(\phi_p,(1+\frac{1}{\log n})\alpha,(1+c)\lambda, c)$-BSCW of~$G[U]^\circ$ (see \cref{rmk:bsctrim}). From \cref{def:bscw}, $\omega=\bot$ implies that $G[A]^\circ=G[X]^\circ$ is a $\phi_p$-expander. Since $p \le k+1$, $k \le \log n$, and $\phi_p$ is defined as
		\begin{equation*}
			\phi_p = \frac{\phi}{\left(\left(1+\frac{1}{\log n}\right)\alpha\right)^p} \, ,
		\end{equation*}
		we conclude that $G[X]^\circ$ is in fact a $\psi$-expander with
		\begin{equation*}
			\psi = \frac{\phi}{\left(\left(1+\frac{1}{\log n}\right)\alpha\right)^p}  \ge \frac{\phi}{\left(\left(1+\frac{1}{\log n}\right)\alpha\right)^{k+1}} \ge \frac{\phi}{6\alpha^{k+1}} \, . 
		\end{equation*}
	\end{itemize}
	From the above case analysis we conclude that for every $X \in\mathcal{U}$, $G[U]^\circ$ is a $\phi/\gamma$-expander with $\gamma = 6 \alpha^{k+1}$.
	
	\paragraph{Property~\eqref{property:crossing}.} We will bound the number of inter-cluster edges by counting the number of edges cut at every recursion level, and then sum over levels. For convenience, let then $D=9C\lambda \cdot \log n$. We can then employ Lemma~\ref{lem:bounddepth} (since our set of assumptions matches that of the lemma), which bounds the recursion depth by $D$.
	
	Next, for $\ell=0,1,\dots, D$, let $\mathcal{I}^\ell$ be the set of instances of Algorithm~\ref{alg:expdec} at the $\ell$-th level of the recursion triggered from $\ed(V,0)$. Each instance $\mathcal{I} \in \mathcal{I}^\ell$ is a pair $\mathcal{I}=(U,X)$, where $U \subseteq V$ is the input cluster to the instance, and:
	\begin{enumerate}
		\item \label{case:ret} if the instance returns $\{U\}$ in line~\eqref{enum:stoprec}, then $X=\emptyset$;
		\item if the instance makes a recursive call, then $\emptyset \neq X \subsetneq U$ is a set on which the instance invokes Algorithm~\ref{alg:expdec}.
	\end{enumerate}
	Then, since our set of assumptions matches that of \cref{lem:recursioncut}, we get that $\unspc{X}{U} < (2\alpha+3C\lambda)\phi$ for every $\ell = 0,\dots, D$ and every $(U,X) \in \mathcal{I}^\ell$ such that $X\neq \emptyset$. Hence,
	\begin{equation}
		\label{bound:edgescrossing}
		\text{for all $\ell = 0,\dots, D$, and all $(U,X) \in \mathcal{I}^\ell$, we have }\quad \uncut{X}{U} < (2\alpha+3C\lambda)\phi \unvolc{X}{U} \, .
	\end{equation}
	With this notation, define the set of edges cut at level $\ell$ to be
	\begin{equation*}
		C_\ell = \bigcup_{(U,X) \in \mathcal{I}^\ell} E(X,U \setminus X) \, .
	\end{equation*}
	Now let $\mu = (3C\lambda+2\alpha)$, and for $\ell=0, \dots, D$ also define
	\begin{equation*}
		T(\ell) = \mu \phi \vol(V) +2  \mu \cdot b\sum_{\ell'=0}^{\ell-1}T(\ell') \, .
	\end{equation*}
	To bound the number of inter cluster edges, we first show that $|C_\ell| \le T(\ell)$, then we resolve the recurrence $T(\ell)$, and finally we sum over $\ell=0,\dots, D$ to conclude.
	
	\begin{claim}
		\label{claim:levelcut}
		For every $\ell=0,1,\dots, D$ one has $|C_\ell| \le T(\ell)$.
	\end{claim}
\begin{proof}
	We prove the statement by induction. At level $0$ there can be only one instance $(V,X) \in \mathcal{I}^0$, so the number of edges cut is $|C_0| = \uncut{X}{V}$. By~\eqref{bound:edgescrossing}
	\begin{equation*}
		|C_0| < \mu \phi  \unvolc{X}{V} =  \mu \phi  \vol(X) \le \mu \phi \vol(V) = T(0) \, ,
	\end{equation*}
thus proving the base case.

Next take any $1 \le \ell \le D$, and assume that $|C_{\ell'}|\le T(\ell')$ for all $\ell'=0,\dots, \ell-1$. By~\eqref{bound:edgescrossing} we know $\uncut{X}{U} < \mu \phi \cdot \unvolc{X}{U}$ for all $(U,X) \in \mathcal{I}^\ell$. Using the definition of $\unvolc{X}{U}$, we get
	\begin{align}
		\uncut{X}{U}  < \mu \phi \cdot \unvolc{X}{U} & =  \mu \phi \cdot \left(\vol(X) + \frac{b-\phi}{\phi}\unbor{X}{U}\right) \\
		\text{since $\unbor{X}{U} \le \uncutg{U}$}\quad & \le \mu \phi \cdot \vol(X) + \mu b \cdot \uncutg{U} \, . \label{bound:cutedges}
	\end{align}
	Observe that the union of all $E(U , V \setminus U)$ on level $\ell$ is a subset of all edges cut on previous recursion levels. This means that
	\begin{equation}
		\sum_{(U,X) \in \mathcal{I}^\ell} \uncutg{U} \leq 2 \sum_{\ell' =0}^{\ell-1} |C_{\ell'}| \, ,\label{bound:cutedges_sum}
	\end{equation}
	since every $e \in C_{\ell'}$ is counted at most twice by the sum over $\mathcal{I}^\ell$.

	Using the induction hypothesis and observing that the cuts $\{X : (U,X) \in \mathcal{I}^\ell\}$ must be disjoint, we can upper-bound $|C_\ell|$ as
	\begin{align*}
		|C_\ell| & = \sum_{(U,X) \in \mathcal{I}^\ell}\uncut{X}{U} \\
		\text{by~\eqref{bound:cutedges} and disjointness of the $X$'s}\quad & < \mu \phi \vol(V) + \mu  b \sum_{(U,X) \in \mathcal{I}^\ell} \uncutg{U}\\
		\text{by~\eqref{bound:cutedges_sum}}\quad & \le \mu\phi \vol(V) + 2 \mu  b \sum_{\ell' =0}^{\ell-1} |C_{\ell'}| \\
		\text{by the induction hypothesis}\quad & \le \mu\phi \vol(V) + 2 \mu b \sum_{\ell' =0}^{\ell-1} T(\ell') \\
		& = T(\ell)\, .
	\end{align*}
	thus showing the inductive statement.
\end{proof}
\noindent
Now we solve the recurrence.
\begin{claim}
	\label{claim:recurrence}
	$T(\ell) \le \mu \phi \vol(V) (1+2\mu b)^\ell$ for all $\ell \ge 0$.
\end{claim}
\begin{proof}
	Let $\ell\ge 1$ and assume the statement holds true for all $0\le \ell'\le \ell-1$. Then,
		\begin{align*}
		T(\ell) & = \mu\phi \vol(V) + 2\mu b \sum_{\ell'=0}^{\ell-1} T(\ell') \\
		& \le \mu\phi \vol(V) + 2\mu b \sum_{\ell'=0}^{\ell-1} \mu \phi \vol(V) (1+2\mu b)^{\ell'} \\
		& = \mu\phi \vol(V) \left(1+2\mu b\sum_{\ell'=0}^{\ell-1} (1+2\mu b)^{\ell'}\right) \\
		& = \mu\phi \vol(V) \left(1+2\mu b \frac{(1+2\mu b)^{\ell}-1}{1+2\mu b-1}\right) \\
		& = \mu\phi \vol(V) (1+2\mu b)^{\ell} \, .
	\end{align*}
\end{proof}
	\noindent
	 Now combing \cref{claim:levelcut} and \cref{claim:recurrence}, one has that the number of inter-cluster edges in $\mathcal{U}$ is
	\begin{align*}
		\sum_{U \in \mathcal{U}} \uncutg{U} \le 2 \sum_{\ell=0}^{D} |C_\ell| \le 2 \sum_{\ell=0}^{D} T(\ell) \le 2 \mu\phi \vol(V)\cdot (D+1) \cdot (1+2\mu b)^{D} \le 4 \mu \cdot \phi \cdot \vol(V) \cdot D \cdot  e^{2b\mu D}\, .
	\end{align*}
\end{proof}

\subsection{Dynamic stream implementation}
\label{subsec:dyn}
In the previous sections we analysed \cref{alg:expdec} and \cref{alg:trim} in an ideal, offline setting with ``free'' access to sparsifiers. To get \cref{lem:tech},  we only need to lift this assumption, and combine lemmas from the previous section.

\unified*

\begin{proof}
	Let $C=40, c=1/40$, and define the quantities $D=9C\lambda \cdot \log n$, $\delta=c^2b/\log n$, and $\delta_j=\delta \cdot (\alpha(1+1/\log n))^{-j}$ for $j \in [k+1]$. We process the stream as follows.
	\begin{enumerate}
		\item We maintain $D+1$ independent samples from the distribution~$\mathcal{D}_\delta$ using the algorithm of  \cref{lem:distrcutstreamadd}, resulting in a collection of sparsifiers $H_\ell$ for $0 \le \ell \le D$.
		\item For every $j \in [k+1]$, we maintain $(D+1) \cdot 1/b$ samples from the distribution~$\mathcal{D}_{\delta_j}$ using the algorithm of  \cref{lem:distrcutstreamadd}, resulting in a collection of sparsifiers $H_{\ell,j,h}$ for $0 \le \ell \le D$, $j~\in~[k+1]$, $h \in [1/b]$.
	\end{enumerate}
	
	\noindent
	After having processed the stream as described above, we want to use \cref{alg:expdec}  and \cref{alg:trim} to decode these samples into a BLD. In \cref{alg:expdec} and \cref{alg:trim} we have hardwired parameters $k,\alpha,\lambda,C,c$: the parameters $k,\alpha,\lambda$ are as in the lemma statement, and $C,c$ are set as defined above (i.e. $C=40, c=1/40$). \cref{alg:expdec} and \cref{alg:trim} also use BLD parameters $b,\phi$: $b$ is just the value in the lemma statement, and we set
	\begin{equation*}
		\phi = \frac{\epsilon}{4 (3C\lambda+2\alpha) \cdot D \cdot  \e^{2b (3C\lambda+2\alpha) D}} \, .
	\end{equation*}
	The only thing left to specify about the algorithms is the sparsifier-access. Since we process the stream by essentially computing a number of sparsifiers, there is a natural way to do this: for $0 \le \ell \le D$,  we define the procedure $\textsc{Sparsifier}^c_\ell(U,\cdot)$ from line~\eqref{enum:freshsampleed} in \cref{alg:expdec} to return $H_\ell$ for every $U \subseteq V$;  similarly, for $0 \le \ell \le D$, $j~\in~[k+1]$, $h \in [1/b]$, we define the procedure $\textsc{Sparsifier}^c_{\ell,j,h}(A,\cdot)$ from line~\eqref{enum:freshsampletrim} in \cref{alg:trim} to return $H_{\ell,j,h}$ for every $A \subseteq V$.
	
	\paragraph{Correctness.} Since $C,\lambda = O(1)$, $\alpha b \le 1/\log n$, $D = O(\log n)$, one can see that $\phi= \Omega(\frac{\epsilon}{\alpha \log n})$. From this value of $\phi$ we also get $\phi < b$, since $C,\lambda,\alpha \ge 1$, $\epsilon \le b \log n$, the exponential is at least $1$, and $D > \log n$ . Then, our parameter regime has all the conditions for using \cref{lem:bounddepth}, \cref{lem:boundinnerloop}, \cref{lem:bld}: we have an $(\alpha,\lambda)$-BSCA $\bsc$, $\phi < b$, $b \le 1/\log n \le c \le 1/30$, $C \ge 31$, $\alpha \le \frac{1}{b \log n} \le \frac{1}{144C \lambda b} \le \frac{1}{2b}$, $\lambda = O(1)$ so $\lambda \le \min\{\frac{1}{18b}, \, \frac{1}{15C\sqrt{b}}, \frac{c}{2b}\}$,  \smash{$ k \ge \frac{\log {n^5 \alpha } }{ \log  {b^{-1/2}}/{\lambda}} $} implies \smash{$ k \ge \log \frac{n^2}{\phi}/\log  \frac{c}{\lambda \cdot b}$} (since \smash{$\phi= \Omega(\frac{\epsilon}{\alpha \log n})$}, $\epsilon \ge 1/n^2$, and $c$ is a constant), and $k \le \log n$. We now just want to lift \cref{assump:sparsifiers1} and \cref{assump:sparsifiers2}, so as to apply these lemmas. From our definition of \smash{$\textsc{Sparsifier}_\ell(U,\cdot)$} and \smash{$\textsc{Sparsifier}_{\ell,j,h}(A,\cdot)$} above, we have the following claim.
	\begin{claim}
		\label{claim:lift}
		Assumptions~\ref{assump:sparsifiers1} and~\ref{assump:sparsifiers2} hold with high probability as long as $\ell \le D$ and $h \le 1/b$.
	\end{claim}
	\begin{proof}
		The claim follows from the way we process the stream, described above, and from our definition of $\textsc{Sparsifier}^c_\ell(U,\cdot)$ and $\textsc{Sparsifier}^c_{\ell,j,h}(A,\cdot)$.
		\begin{itemize}
			\item For every $0 \le \ell \le D$ and any $U \subseteq V$, $\textsc{Sparsifier}^c_\ell(U,\cdot)$ outputs the graph $H_\ell$, independently sampled from $\mathcal{D}_\delta$.  \cref{alg:expdec} only calls $\textsc{Sparsifier}^c_\ell(U,b)$ for a cluster $U$ that was found using randomness from levels before $\ell$ (simply by the description of the algorithm, see recursive calls in lines~\eqref{enum:balrecurse},~\eqref{rec:onside},~\eqref{rec:twosides} of \cref{alg:expdec}). Thus, by \cref{lem:distrcutstreamadd} we know that $H_\ell[U] \approx_\delta G[U]$ with high probability. Since there can be at most $n$ calls to $\textsc{Sparsifier}^c_\ell(U,b)$ for any $\ell$, we can conclude that $H_\ell$ works for each of them with high probability.
			\item For every $0 \le \ell \le D,\,j \in[k+1],\,h \in [1/b]$ and any $A \subseteq V$, $\textsc{Sparsifier}^c_{\ell,j,h}(A,\cdot)$ outputs the graph $H_{\ell,j,h}$, independently sampled from $\mathcal{D}_{\delta_j}$. Similarly as before, we can see from the description of \cref{alg:trim} that $\textsc{Sparsifier}^c_{\ell,j,h}(A,b_j)$ can only be called on a cluster $A$ that depends on the randomness of a previous level (in case $A$ equals the input cluster $U$ to \cref{alg:trim}) or a previous iteration (in case $A$ was updated in line~\eqref{enum:trim}). Thus, $A$ is independent from $H_{\ell,j,h}$. Thus, by \cref{lem:distrcutstreamadd} we know that $H_{\ell,j,h}[A] \approx_{\delta_j} G[A]$ with high probability. Also in this case there can be at most $n$ calls to $\textsc{Sparsifier}^c_{\ell,j,h}(A,b_j)$ for any triple $\ell,j,h$, so $H_{\ell,j,h}$ works for each of them with high probability.
		\end{itemize}
	\end{proof}
	\noindent
	Provided  \cref{assump:sparsifiers1} and \cref{assump:sparsifiers2} hold, \cref{lem:bounddepth} and \cref{lem:boundinnerloop} are guaranteeing that there is no point in the execution of the algorithm with recursion level $\ell$ larger than $D$ or with inner loop index $h$ (see line~\eqref{enum:innerloop}) larger than $1/b$. On the other hand, \cref{claim:lift} ensures that \cref{assump:sparsifiers1} and \cref{assump:sparsifiers2} hold with high probability as long as $\ell \le D$ and $h \le 1/b$. Hence, we have that the results from \cref{lem:bounddepth}, \cref{lem:boundinnerloop}, \cref{lem:bld} hold simultaneously with high probability. In particular, the output $\mathcal{U}$ of $\ed(V,0)$
	is a $(b,\epsilon,\phi,\gamma)$-BLD of $G$ with high probability for $\gamma = O(\alpha^{k+1})$ (note that we set $\phi$ exactly by rearranging the value of $\epsilon$ given by \cref{lem:bld}).
	
	\paragraph{Space complexity.} The space requirement for processing the stream is dominated by the sparsifiers that we maintain, plus the space used by any call to $\bsc$ for the decoding phase. From \cref{lem:distrcutstreamadd} we know that sampling from $\mathcal{D}_{\delta}$ can be done by maintaining a linear sketch of \smash{$\otil(n/\delta^2)$} bits, and sampling from $\mathcal{D}_{\delta_j}$ can be done maintaining a linear sketch of $\otil(n/\delta_j^2)$ bits.  Since we maintain $D+1$ samples from~$\mathcal{D}_\delta$, and $(D+1) \cdot 1/b$ samples from $\mathcal{D}_{\delta_j}$ for every $j \in [k+1]$, the space complexity for processing the stream is
	\begin{equation*}
		(D+1) \cdot \otil\left(\frac{n}{\delta^2}\right) + (D+1) \cdot \frac{1}{b}\cdot \sum_{j=1}^{k+1} \otil\left(\frac{n}{\delta_j^2}\right) \, .
	\end{equation*}
	Because we defined $D=9C\lambda \cdot \log n$, $\delta=c^2b/\log n$, $\delta_j=\delta \cdot (\alpha(1+1/\log n))^{-j}$, and we have $C,\lambda = O(1),\,k = O(\log n), c = \Omega(1)$, we use at most
	\begin{equation*}
		O(\log n) \otil\left(\frac{n}{b^2}\right) + O(\log n)  O\left(\frac{1}{b}\right)O(k) \cdot \otil\left(\frac{n}{b^2}\right) \cdot \alpha^{O(k)} = \otil\left(\frac{n}{b^3}\right) \cdot \alpha^{O(k)}
	\end{equation*}
	bits of space for the sparsifiers. Finally, note that \Cref{alg:expdec,alg:trim} run $\bsc$ on sparsifiers with at most $n$ vertices and at most
	\begin{equation*}
		\max\left\{\otil(n/\delta^2),\max_{j \in [k+1]}\left\{\otil(n/\delta_{j}^2)\right\}\right\} =  \otil\left(\frac{n}{b^2}\right) \cdot \alpha^{O(k)}
	\end{equation*}
	edges, by \cref{lem:distrcutstreamadd}. As we can reuse the space for each call, this takes an extra $S(n,\otil(n/b^2)\alpha^{O(k)})$ bits of space, leading to the claimed space complexity for the decoding.
	
	\paragraph{Time complexity.}
	Throughout the recursion of $\ed(V,0)$,  \Cref{alg:expdec,alg:trim} make at most $\poly(n)$ calls to $\bsc$. As noted above, each of these calls is on a graph with at most $n$ vertices and \smash{$\otil(n/b^2)\alpha^{O(k)}$} edges. \cref{lem:distrcutstreamadd} also guarantees that decoding the linear sketch to get a sample from $\mathcal{D}_\delta$ and $\mathcal{D}_{\delta_j}$ takes $\poly(n)$ time. Hence the running time of the decoding is \smash{$\poly(n)\cdot T(n,\otil(n/b^2)\alpha^{O(k)})$}.
	
	\paragraph{Sparsity parameter in BSCA calls.} By the way we set $\phi$ and by definition of the algorithm, we have that every call to $\bsc(\cdot, \psi)$ is made with a sparsity parameter $\psi \in (0,1)$ that is $\psi \le \phi \le \frac{1}{10\alpha}$.
	
\end{proof}
\section{Lower bound for two-level expander decomposition}
\label{sec:lb}
The goal of this section is to show that repeatedly computing EDs on the inter-cluster edges necessarily incurs a dependence on the sparsity parameter. Formally, the result is the following.

\lb*
\noindent
By Yao's minimax principle, it is sufficient to give a distribution $\mathcal{G}$ over graphs with vertex set $V$ such that any deterministic streaming algorithm that computes, with probability at least $0.9$, an $(\epsilon,\phi,\ell)$-RED of $G \sim \mathcal{G}$  takes $\Omega(n/\epsilon)$ bits of space. We are not going to use the fact that the stream can be dynamic to show the lower bound, so any stream that reveals the edges of $G \sim \mathcal{G}$  in an arbitrary order serves our purposes.

We then proceed as follows. In \cref{subsec:harddistrib} we describe our hard distribution over graphs, and show a structural property of its two-level $(\epsilon,\phi)$-RED. In \cref{subsec:reductioncomm} we introduce a communication problem, $\recover$, and show a reduction to RED. Finally, in \cref{subsec:hardcomm} we give a communication lower bound for $\recover$. We combine these results into proving~\cref{th:lb} in \cref{subsec:lbproof}.

\subsection{Hard distribution}
\label{subsec:harddistrib}
The randomness of our hard distribution $\mathcal{G}$ arises from sampling Erd\H{o}s-R\'{e}nyi graphs. We recall the definition for convenience of the reader.
\begin{definition}\label{def:er}
	Let $N$ be an integer and let $p \in [0,1]$. The distribution of Erd\H{o}s-R\'{e}nyi graphs $\er(N, p)$ is the distribution over $N$-vertex graphs $H=(U,F)$ where each pair $e \in \binom{U}{2}$ is in $F$ with probability $p$ independently of other pairs.
\end{definition}
\noindent
Let $d,m$ be integers such that $3 \le d < m < n$. We now define our hard distribution $\mathcal{G}$ over $n$-vertex  graphs $G=(V, E)$, and we illustrate it in \Cref{fig:lb}.

\begin{definition}[Distributions $\mathcal{G}$ and $\mathcal{G}'$]
	\label{def:harddistr}
 We partition $V=[n]$ arbitrarily  into two sets $S$ and $T$ with $n/2$ vertices each, and further partition $S$ into $n/m$ sets $S_1,\dots,S_{n/m}$ with $m/2$ vertices each. The distribution $\mathcal{G}_{n,d,m}$ is supported over graphs $G=(V,E)$ where $E$ is defined as follows.

\begin{enumerate}
\item  For each $i~\in~[n/m]$, the induced subgraph $G[S_i]$ is sampled independently from the Erd\H{o}s-R\'{e}nyi distribution \smash{$\er({m}/{2},4d/m)$} (see Definition~\ref{def:er} above). We denote by $\mathcal{G}'_{n,d,m}$ the distribution of the subgraph $G[S]$, i.e. \smash{$\mathcal{G}'_{n,d,m}= \er({m}/{2},4d/m)^{\otimes n/m}$}.

\item The induced subgraph $G[T]$ is a {\em fixed} $d$-regular $\psi$-expander, for $\psi = \Omega(1)$ independent of $d, m, n$. 

\item We fix for convenience an arbitrary labelling $s_{i,1},\dots,s_{i,m/2}$ of the vertices in each $S_i, \, i \in [n/m]$, and we sample $K \sim \unif([m/2])$. Then, for every $i\in [n/m]$, we add $dm/2$ edges from $s_{i,K}$ to $T$ so that each $t \in T$ has $d$ incident edges in $E(S,T)$ (one way to do so is by partitioning $T$ into $n/(dm)$ subsets $T_1, \dots, T_{n/(dm)}$ with $dm/2$ vertices each and connecting $s_{i,K}$ to each vertex of $T_{\lceil i/d \rceil}$ for every $i \in [n/m]$).
\end{enumerate}
Since the parameters $n,d,m$ are fixed hereafter, we may drop the subscript from $\mathcal{G}_{n,d,m}$ and $\mathcal{G}'_{n,d,m}$ and simply denote them by $\mathcal{G}$ and $\mathcal{G}'$ respectively.
\end{definition}

\begin{figure}[h]
	\begin{minipage}[center]{\textwidth}
		\centering
		\includegraphics[scale=0.93]{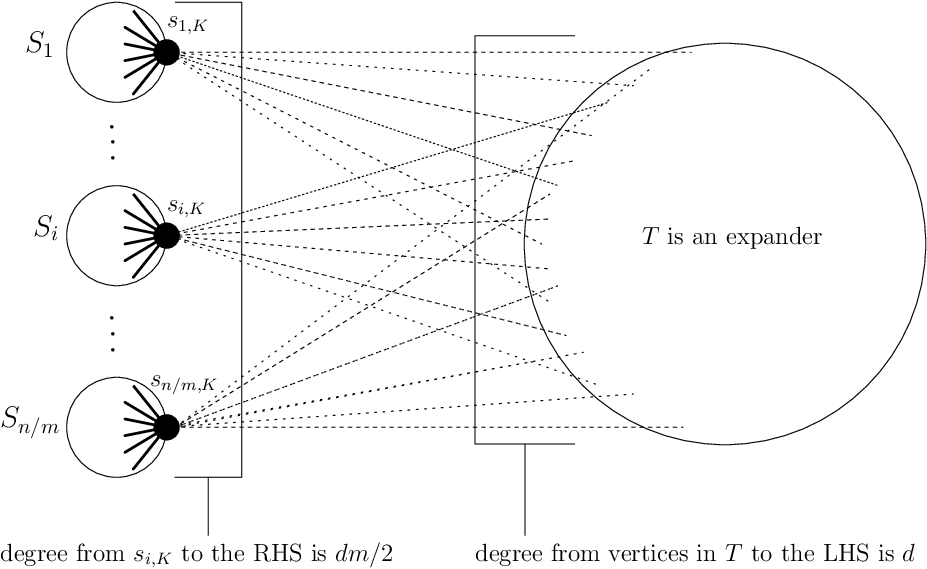}
	\end{minipage}
	\caption{Illustration of the graph we use for proving the lower bound. Thick bullets represent important vertices, thick lines represent important edges, dotted lines represent edges connecting the important vertices to $T$.}
	\label{fig:lb}
\end{figure}

\noindent
Roughly speaking, our hard instances should be composed of $n/m$ regular expanders that are densely connected to $T$, which is also an expander, through a selection of ``special'' vertices. We will show that the hardness arises from recovering information about certain important vertices and edges, defined below and also illustrated in \Cref{fig:lb}.

\begin{definition}[Important vertices and edges] 
	\label{def:impedges}
Let $G=(V,E) \in \supp(\mathcal{G})$. Denote by $k \in [m/2]$ the unique index such that $\{s_{i,k}:i \in [n/m]\}= \{s \in S: \{t \in T: \{s,t\} \in E\} \neq \emptyset \}$. Then, we define the set of important vertices $V^*=\{s_{i,k}:i \in [n/m]\}$ to be the set of vertices of $S$ that are adjacent to~$T$, and define the set of important edges $E^*=\{\{u, v^*\} :  u \in S\, , v^* \in V^*\}$ to be the set of edges in the induced subgraph $G[S]$ that are incident on $V^*$.
\end{definition}

\begin{remark}
	\label{rmk:harddistrib}
	Note that any sample $G \sim \mathcal{G}$ can be defined by $(G',K)$, where $G'=(S,E')$ is the graph consisting of the disjoint union of the subgraphs $(G[S_i])_{i \in [n/m]}$, and $K \in [m/2]$ is the random index determining the important vertices and edges. This is because the edges within $G[T]$ are the same for all $G \sim \mathcal{G}$, and the edges $E(S,T)$ are determined by $K$. By virtue of this observation we abuse notation and write $G=(G',K) \sim \mathcal{G}$. Moreover, for any $(G',K) \sim \mathcal{G}$ one has $G' \sim \mathcal{G}' $, and conversely, for any $G' \sim \mathcal{G}'$ and $K \sim \unif([m/2])$ one also has~$(G',K) \sim \mathcal{G}$.
\end{remark}

\noindent
In the remainder of the section, we call a graph $(1 \pm \delta)\Delta$-regular for an integer $\Delta>0$ and $\delta \in [0,1]$, if each of its vertices has degree in the range $[(1-\delta)\Delta , (1+\delta)\Delta]$. With this terminology, we show that in any valid $(\epsilon,\phi)$-ED of $G$, most of the important edges are crossing edges, as depicted in \Cref{fig:decomp}. This works under the assumption that each $G[S_i]$ is a near-regular expander.

\begin{figure}[h]
	\begin{minipage}[center]{\textwidth}
		\centering
		\includegraphics[scale=0.83]{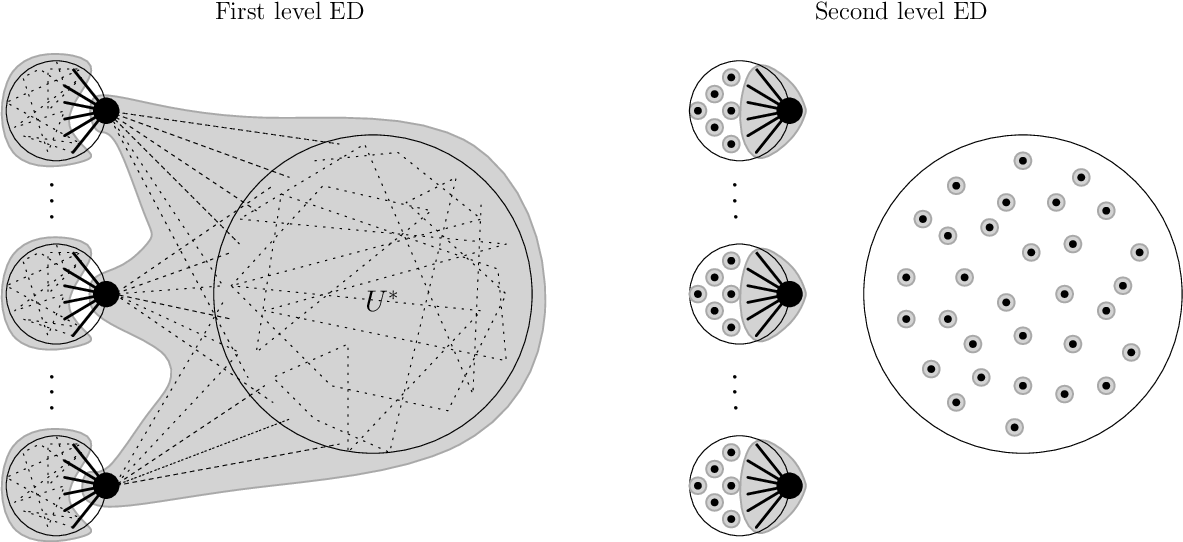}
	\end{minipage}
	\caption{Illustration of the ideal expander decomposition of the graph. Thick bullets represent important vertices, smaller bullets represent ordinary vertices, thick lines represent important edges, dotted lines represent the rest of the edges. Grey areas represent the clusters in the decomposition.}
	\label{fig:decomp}
\end{figure}

\begin{lemma} \label{lem:special-edge}
	Let $G=(V,E) \in \supp(\mathcal{G})$, and let $\epsilon,\phi \in (0,1)$ such that $\epsilon \le 10^{-5} \cdot \psi \cdot \frac{d}{m}$ and \smash{$\phi \ge 11/m$}. If $n \ge 100 m$, $m\ge 500$, and every subgraph $(G[S_i])_{i \in [n/m]}$ is a $(1 \pm 1/10)2d$-regular $\frac{1}{3}$-expander, then any $(\epsilon,\phi)$-ED~$\mathcal{U}$ of $G$ satisfies
	\begin{equation*}
		\left| E^* \setminus \mathcal{U}\right| \ge \frac{4}{5} \cdot \left| E^*\right| \, .
	\end{equation*}
\end{lemma}

\begin{proof} 
	Let $G=(V,E) \in \supp(\mathcal{G})$ and let $k \in [m/2]$ such that the important vertices are $s_{i,k}$ for $i \in [n/m]$ (note that such value of $k$ can be uniquely inferred from seeing $G$ since the ordering of the vertices is fixed). Note that by \cref{def:harddistr} together with our assumption, the graph $G$ has at most $11/10\cdot dn/2+3/4 \cdot dn \le 1.3 dn$ edges.
	
	Hereafter, let $\mathcal{U}$ be any $(\epsilon,\phi)$-ED of $G$. We begin by arguing that there is a giant cluster in $\mathcal{U}$ that covers most of $T$ and most of the important vertices.
	
	\begin{claim}
		\label{claim:giantcluster}
		There is a component $U^* \in \mathcal{U}$ such that
		\begin{equation*}
			|T \cap U^*| \ge\left(1-12\frac{\epsilon}{\psi}\right) |T| \quad \text{and} \quad |V^* \cap U^*| \ge  \left(1-12\frac{\epsilon}{\psi}\right) |V^*| \, .
		\end{equation*}
	\end{claim}
	\begin{proof}
		For any valid $(\epsilon,\phi)$-ED $\mathcal{U}$, we must have $|E \setminus \mathcal{U}| \le \epsilon \cdot 1.3 dn$. Since $G[T]$ is a $d$-regular $\psi$-expander, there is a component $U^* \in \mathcal{U}$ that contains at least $(1-3 \epsilon/\psi) \cdot n/2$ vertices from $T$, i.e.
		\begin{equation*}
		|T \cap U^* | \ge\left(1-3\frac{\epsilon}{\psi}\right)\cdot \frac{n}{2} = \left(1-3\frac{\epsilon}{\psi}\right)|T|   \, .
		\end{equation*}
		To see why, note that if this was not the case then we could construct a cut $S = \cup_{U \in \mathcal{U}'} U \cap T$ for some $\mathcal{U}' \subseteq \mathcal{U}$ such that $3\epsilon/\psi \cdot |T| < |S| \le |T|/2$ and conclude that there would be at least $ \psi \vol_{G[T]}(S) > \epsilon \cdot 1.5 dn$ edges crossing $\mathcal{U}$.
		
		Since each vertex in $T$ has an edge to $d$ important vertices (see \cref{def:harddistr}), there are at least $(1-6\epsilon/\psi) \cdot n/m$ important vertices with at least $dm/4$ neighbours in $U^* \cap T$, i.e.
		\begin{equation*}
			\Big|\Big\{s \in V^*: \, \left|\{t \in U^* \cap T: \{s,t\} \in E\}\right| \ge \frac{1}{4}dm\Big\}\Big| \ge  \left(1-6\frac{\epsilon}{\psi}\right) \cdot \frac{n}{m} = \left(1-6\frac{\epsilon}{\psi}\right)\cdot |V^*| \, .
		\end{equation*}
		To see why, note that if this was not the case, we would have $|\{s,t\} \in E:\, s \in V^*\, , t\in T \setminus U^* \}|$ to be larger than $ (6\epsilon/\psi \cdot n/m) \cdot dm/4$, which is impossible since every $t \in T$ has $d$ incident edges in $E(S,T)$ and $|T \setminus U^*| \le 3\epsilon/\psi \cdot n/2$.
		
		Now, because most of the important vertices have many edges to $T \cap U^*$, most of these should also be in $U^*$, otherwise too many edges would cross: since $|E \setminus \mathcal{U}| \le \epsilon \cdot 1.3 dn$, there must be at most $6\epsilon/\psi \cdot n/m+\frac{\epsilon \cdot 1.3 dn}{dm/4} \le 12\epsilon/\psi \cdot n/m$ important vertices that are not in $U^*$, i.e.
		\begin{equation*}
			|V^* \cap U^*| \ge \left(1-12\frac{\epsilon}{\psi}\right) \cdot \frac{n}{m}  =  \left(1-12\frac{\epsilon}{\psi}\right) \cdot |V^*| \, .
		\end{equation*}
	\end{proof}
	\noindent
	We want to confine our attention to important vertices whose corresponding subgraph $G[S_i]$ locally preserves property~\eqref{property:crossingclassic} of \cref{def:expdecclassic}, i.e. has few inter-cluster edges as compared to its own volume. Formally, we consider the set of good vertices, defined as
	\begin{equation}
		\label{eq:defgoodvertex}
		V^*_{\text{g}} = \left\{s_{i,k} \in V^* : s_{i,k} \in U^* \, , \left|(E \setminus \mathcal{U}) \cap \binom{S_i}{2}\right| \le {200\epsilon} dm\right\} \, .
	\end{equation}
	We next show that the subgraph $G[S_i]$ corresponding to a good vertex is covered almost entirely by a single cluster, except the good vertex itself which is assigned to $U^*$.
	
	\begin{claim}
		\label{claim:separateclusters}
		Let $s_{i,k} \in V^*_{{\good}} $. Then, there is a cluster $U \in \mathcal{U}$ such that
		\begin{equation*}
			s_{i,k} \notin U \quad \text{and} \quad |U\cap S_i| \ge \left(1-700{\epsilon}\right)  |S_i| \, .
		\end{equation*}
	\end{claim}
\begin{proof}
	Since $G[S_i]$ is a $(1 \pm 1/10)2d$-regular $\frac{1}{3}$-expander by assumption, there is at least one component $U \in \mathcal{U}$ containing at least $(1-700\epsilon)m/2$ vertices in $S_i$, i.e.
	\begin{equation}
		\label{eq:badcomp}
		|U \cap S_i| \ge \left(1-700{\epsilon}\right) \frac{m}{2}  = \left(1-700{\epsilon}\right)  |S_i| \, ,
	\end{equation}
	To see why, note that if this was not the case we could construct a cut $R = \cup_{W \in \mathcal{U}'} W \cap S_i$ for some $\mathcal{U}' \subseteq \mathcal{U}$ such that \smash{$10/9 \cdot 700\epsilon \cdot |S_i| < |R| \le 9/11 \cdot  |S_i|/2$} and conclude that there would be at least $ 1/3\cdot \vol_{G[S_i]}(R) > 230\epsilon dm$ edges crossing $\mathcal{U}$ inside $S_i$, contradicting the fact that $s_{i,k} \in V^*_{{\good}} $ (see the definition in~\eqref{eq:defgoodvertex}).
	
	Hereafter, let $U \in \mathcal{U}$ be the cluster that satisfies~\eqref{eq:badcomp} (which is unique, since $\epsilon \le 10^{-5}$). We know from \cref{claim:giantcluster} that there is a cluster $U^* \in \mathcal{U}$ that contains all good vertices (by definition of good vertex, see~\eqref{eq:defgoodvertex}) and that covers most of $T$.	To show that $s_{i,k} \notin U$, we leverage the fact that $s_{i,k} \in U^*$ and that $G[U^*]$ must be a $\phi$-expander by property~\eqref{property:expanderclassic} of \cref{def:expdecclassic}. Suppose towards a contradiction that $U^* = U$, and let $Q = (U^* \cap S_i) \setminus \{s_{i,k}\}$. Because $G[S_i]$ is $(1 \pm 1/10)2d$-regular, the volume of this set in $G[U^*]$ is at least
	\begin{align*}
		\vol_{G[U^*]}(Q) & \ge -\deg_{G[S_i]}(s_{i,k}) + \sum_{u \in U^* \cap S_i} \deg_{G[U^* \cap S_i]}(u) \\
		& \ge  -\frac{11}{10}2d + \sum_{u \in U^* \cap S_i:\,\deg_{G[U^* \cap S_i]}(u) \ge (1-600{\epsilon}) \deg_{G[S_i]}(u)} \deg_{G[U^* \cap S_i]}(u) \, .
	\end{align*}
	We next observe that the sum above has at least $|U^* \cap S_i|/2$ many terms, for if this was not the case we would have at least $300\epsilon |U^* \cap S_i| \cdot 9/10 \cdot  2d \ge 200\epsilon dm$ edges from $U^* \cap S_i$ to $S_i \setminus U^*$, which would contradict $s_{i,k} \in V^*_{{\good}} $ (see the definition in~\eqref{eq:defgoodvertex}). We then continue from above and hence get
	\begin{align*}
		\vol_{G[U^*]}(Q) & \ge  -\frac{11}{10}2d + \frac{1}{2} |U^* \cap S_i|\cdot  \left(1-600{\epsilon}\right) \deg_{G[S_i]}(u)\\
		& \ge  -\frac{11}{10}2d + \frac{1}{2} |U^* \cap S_i|\cdot  \left(1-600{\epsilon}\right) \frac{9}{10} 2d\\
		\text{since $|S_i \cap U^*| \ge (1-700\epsilon) |S_i|$}\quad & \ge -\frac{11}{10}2d+\frac{9}{10}2d \frac{1}{2}\left(1-1400{\epsilon}\right) \frac{m}{2}\, .
	\end{align*}
	Moreover, we know that $|U^* \cap T| \ge (1-3\epsilon/\psi)n/2 \ge (1-1/8000)n/2$ by our assumption on $\epsilon$, so the fact that $G[T]$ is $d$-regular implies $\vol_{G[U^*]}(U^*) \ge \vol_{G[U^* \cap T]}(U^* \cap T) \ge dn/10$. Also, by near-regularity of $G$ and because $m \le n/100$, we have $\vol_{G[U^*]}(Q) \le 3d \cdot m/2 < dn/20 \le \frac{1}{2} \vol_{G[U^*]}(U^*)$. This means that $Q$ is the smaller side of the cut. On the other hand, the number of edges crossing the cut $Q$ in $G[U^*]$ is too small, since vertices in $Q$ can only connect to the rest of $U^*$ through edges incident on $s_{i,k}$, i.e.
	\begin{equation*}
		\uncut{Q}{U^*} = \deg_{G[S_i]}(s_{i,k}) \le  \frac{11}{10} 2d \, .
	\end{equation*}
	Therefore, $U_i = U^*$ implies
	\begin{equation*}
		\Phi_{G[U^*]} \le \Phi_{G[U^*]}(Q) = \frac{\uncut{Q}{U^*}}{\vol_{G[U^*]}(Q)} \le \frac{\frac{11}{10} 2d}{\frac{9}{10}d \left(1-1400{\epsilon}\right) \frac{m}{2}-\frac{11}{10}2d} \le \frac{10}{m} < \phi \, ,
	\end{equation*}
	a contradiction. Here we used that $m \ge 500$, $\phi \ge 11/m$, and $\epsilon \le 10^{-5}$.
\end{proof}
	\noindent
	By virtue of \cref{claim:separateclusters}, we have the following setting: for every good vertex $s_{i,k} \in V^*_{\text{g}}$, there is one component $U$ covering a $(1-700\epsilon)$ fraction of $S_i$ except $s_{i,k}$ itself. Therefore, many important edges incident on each good vertex are crossing edges, i.e.
	\begin{equation*}
		\{\{s,s'\} \in E^* : \, s = s_{i,k}, s' \notin U^*\} \subseteq E(U^* \cap S_i, S_i \setminus U^*) \subseteq E \setminus \mathcal{U} \quad \text{for all $s_{i,k} \in V^*_{\text{g}}$,}
	\end{equation*}
	and
	\begin{equation*}
		|U^* \cap S_i| \le  |U \setminus S_i| \le 700{\epsilon}  |S_i| = 700{\epsilon} \frac{m}{2}\quad \text{for all $s_{i,k} \in V^*_{\text{g}}$.}
	\end{equation*}
	In particular, the above conditions implies that, among the important edges incident on each good vertex, at most $700\epsilon \cdot m/2$ of them can land in the same cluster $U^*$ as the good vertex itself, i.e.
	\begin{equation}
		\label{eq:pervertexguarantee}
		|\{\{s,s'\} \in E^* : \, s = s_{i,k}, s' \notin U^*\}| \ge \frac{9}{10}2d - 700{\epsilon} \frac{m}{2}  \quad \text{for all $s_{i,k} \in V^*_{\text{g}}$.}
	\end{equation}
	If there are many good vertices, we can conclude that a lot of the important edges are crossing.
		\begin{claim}
		\label{claim:mannygood}
		One has $|V^*_{\good}| \ge  \frac{49}{50}|V^*| $.
	\end{claim}
	\begin{proof}
		Since the number of crossing edges is $|E \setminus \mathcal{U}| \le \epsilon \cdot 1.3 dn$, one has
		\begin{equation}
			\label{goodvertices:2}
			\left|\left\{s_{i,k} \in V^* : \, \left|(E \setminus \mathcal{U}) \cap \binom{S_i}{2}\right| \le {200\epsilon}dm \right\}\right| \ge \left(1-\frac{1}{100}\right)\frac{n}{m}= \left(1-\frac{1}{100}\right)|V^*| \, .
		\end{equation}
		Combining~\cref{claim:giantcluster} and~the above bound, we know $V^*_{\good}$ to be large:
		\begin{equation*}
			|V^*_{\text{g}}| \ge \left(1-12\frac{\epsilon}{\psi}-\frac{1}{100}\right) \cdot |V^*| \ge \left(1-\frac{1}{50}\right) \cdot |V^*| \, ,
		\end{equation*}
		where we use the assumption that $\epsilon \le 10^{-5} \cdot \psi \cdot \frac{d}{m}$.
	\end{proof}
\noindent
	By \cref{claim:mannygood} and the bound in~\eqref{eq:pervertexguarantee}, we can conclude the proof:
	\begin{align*}
		|E^*  \setminus \mathcal{U}| \ge \sum_{s_{i,k} \in V^*_{\text{g}}} |\{\{s,t\} \in E^* : \, s = s_{i,k}\}|   \ge   \left(\frac{9}{10} - 700{\epsilon} \frac{m}{2}\frac{1}{2d}\right)2d \cdot |V^*_{\text{g}}| \, ,
	\end{align*}
	and by \cref{claim:mannygood} we get
	\begin{align*}
	|E^*  \setminus \mathcal{U}|  \ge  \left(\frac{9}{10} - 700{\epsilon}\frac{m}{2}\frac{1}{2d}\right)2d \cdot \frac{49}{50} \cdot \frac{n}{m}  \ge \frac{4}{5} \cdot |E^*|  \, ,
	\end{align*}
	where we used $\epsilon \le 10^{-5} \psi \frac{d}{m}$ in the last inequality.
\end{proof}

\subsection{Reduction to a communication problem}
\label{subsec:reductioncomm}
We introduce the two-party one-way communication problem $\recover$, defined as follows.
\begin{definition}[Communication problem $\recover$]
	\label{def:recover}
	Let $\xi >0$  be a real parameter. In the communication problem $\recover_{\xi}$,
	Alice's input is a graph $G'=(S,E')$, and Bob's output is a set of pairs of vertices $F~\subseteq~\binom{S}{2}$ such that
	\begin{enumerate}
		\item \label{recover:fsmall} $|F| \le 6 \xi\cdot |E'| $, and
		\item \label{recover:flearns} at least a $1/10$ fraction of the pairs in $F$ are edges of $G'$, i.e. $|F \cap E'| \ge 1/10\cdot |E'|$.
	\end{enumerate}
\end{definition}
\noindent
We give a reduction from $\recover_{\epsilon m}$ to $(\epsilon, \phi,2)$-RED, in the setting where Alice's input $G'$ is an $n/2$-vertex graph drawn from $\mathcal{G}'$. We recall that \smash{$\mathcal{G}'=\mathcal{G}'_{n,d,m} = \er({m}/{2}, 4d/m)^{\otimes n/m}$} is the distribution over $n/2$-vertex graphs $G'=(S,E')$ consisting of the disjoint union of $n/m$ disjoint Erd\H{o}s-R\'{e}nyi graphs on $m/2$ vertices (see \cref{def:harddistr}). We then begin by verifying that, with high probability, an input $G' \sim \mathcal{G}'$ is a collection of $n/m$ nearly regular expanders.

\begin{lemma} \label{lem:leftexpander}
	Let $G' \sim \mathcal{G}'$. If $m\ge 500$ and $d \ge 600 \log n$ then,  with probability at least $1-1/n$, for every $i \in [n/m]$ one has that the subgraph $G'[S_i]$ is a $(1 \pm 1/10)2d$-regular $1/3$-expander
\end{lemma}

\begin{proof}
	The statement follows from the following fact about Erd\H{o}s-R\'{e}nyi graphs, whose proof is deferred to \cref{apndx:randomgraphs}.
	\begin{restatable}{fact}{erdos}
		\label{prop:er}
		Let $N \ge 10$ be an integer and $p \in [0,1]$. Then, letting $\bar{d} = p(N-1)$ one has
		\begin{equation*}
			\Pr_{H=(U,F) \sim \er(N,p)}\left[\left(\Phi_H < \frac{1}{3} \right)\text{ or } \left(\exists u \in U \text{ s.t. } \left|\deg_H(u)-\bar{d}\right|>\frac{1}{11} \bar{d}\right)\right] \le 4N \cdot \exp\left(-\frac{pN}{600}\right) \, .
		\end{equation*}
	\end{restatable}
	
	\noindent
	In our case, each subgraph $G'[S_i]$ is drawn from $\er({m}/{2}, 4d/m)$, and note that $m/2 \ge 250 \ge 10$, thus meeting the precondition of \cref{prop:er}. Hence, for any $i \in [n/m]$ \cref{prop:er} gives
	\begin{equation*}
		\Pr\left[\left(\Phi_{G'[S_i]} < \frac{1}{3}\right) \text{ or } \left(\exists u \in S_i \text{ s.t. } \left|\deg_{G'[S_i]}(u)-2d\right|>\frac{1}{10}2d\right)\right] \le 2m \cdot \exp\left(-\frac{d}{300}\right) \, ,
	\end{equation*}
	where we used $m \ge 500$ to express the degree bounds in terms of $2d$ instead of $\bar{d}$.
	Therefore, taking a union bound gives
	\begin{equation*}
		\Pr\left[\forall \, i \in [n/m]\, , \left( \Phi_{G'[S_i]} \ge \frac{1}{3} \text{ and } \forall u \in S_i\, , \left|\deg_{G'[S_i]}(u)-2d\right|\le\frac{1}{10} 2d\right)\right] \ge1- \frac{1}{n} \, ,
	\end{equation*}
	where we used $d \ge 600 \log n \ge 600 \ln n$.
\end{proof}
\noindent
By virtue of the above, most of Alice's inputs  can be embedded into a sample from $\mathcal{G}$ where a large fraction of important edges are inter-cluster edges of any ED, as per~\cref{lem:special-edge}. This will allow Alice and Bob to solve $\recover_{\epsilon m}$.

\begin{reduct}
	\label{red:recover}
	Let $\mathcal{A}$ be a deterministic streaming algorithm for computing a $2$-level $(\epsilon,\phi)$-RED of a graph given in a stream of edges. For a graph $G'=(S,E') \sim \mathcal{G}'$, we reduce $\recover_{\epsilon m}$ to computing a $2$-level $(\epsilon,\phi)$-RED as follows. Alice instantiates $\mathcal{A}$ for the vertex set $V$, and feeds her edges $E'$ and the fixed edges of $G[T]$ to $\mathcal{A}$. Then, she sends the memory state $\Pi$ of $\mathcal{A}$ to Bob. Then, Bob instantiates $m/2$ copies of $\mathcal{A}$, call them $\mathcal{A}_1, \dots, \mathcal{A}_{m/2}$, and sets the memory state of them to~$\Pi$. Next, for each $k \in [m/2]$ let $G_k = (V,E_k)$ be the graph defined by $(G',k)$, and note that Bob knows the edges between $S$ and $T$ (see \cref{def:harddistr}). Then, Bob feeds the edges $E_k(S,T)$ to~$\mathcal{A}_k$. For each $k \in [m/2]$, let $\mathcal{U}^k_1,\mathcal{U}^k_2$ be the output of $\mathcal{A}_k$. Bob finally constructs his output set $F$ as follows: for each $k \in [m/2]$, if the number of non-isolated vertices in $\mathcal{U}^k_2$ is at most $3 \epsilon d n$, add the pair $\{s,s_{i,k}\}$ to $F$ for every $i \in [n/m]$ and every $s \in S_i$ that is not an isolated vertex in $\mathcal{U}^k_2$. More formally, for a partition $\mathcal{U}$ of $V$, let $V \setminus \mathcal{U} = \{u \in V: \, \{u\} \notin \mathcal{U}\}$ be the set of non-isolated vertices in $\mathcal{U}$. With this notation, Bob outputs
	\begin{equation}
		\label{eq:boboutput}
		F = \bigcup_{k\in [m/2] : |V \setminus \mathcal{U}_2^k| \le 3\epsilon dn} F_k \, , \quad \quad \text{with } F_k=\left\{ \{s,s_{i,k}\} : i \in [n/m],\, s \in S_i \cap (V \setminus \mathcal{U}_2^k) \right\} \, .
	\end{equation} 
\end{reduct}
\noindent
We remark that if Bob learns the important edges $E^*_k$ for all $k \in [m/2]$, then he learns all of $E'$.
\begin{remark}
	\label{rmk:unionimportant}
	In the setting of~\cref{red:recover}, we have $E' = \bigcup_{k \in [m/2]} E^*_k$ and $|E'| = \frac{1}{2} \sum_{k \in [m/2]} |E^*_k|$.
\end{remark}
\noindent
As suggested by \cref{rmk:unionimportant}, it is enough to ensure that each instance $\mathcal{A}_k$ recovers the important edges of $G_k$ in order for Bob to learn the edges of $G'$. We in particular restrict ourselves to indices~$k$ for which the algorithm $\mathcal{A}_k$, the graph $G_k$, and the RED $\mathcal{U}_1^k,\mathcal{U}_2^k$ are well behaved, as defined below.

\begin{definition}
	\label{def:pk}
	In the setting of~\cref{red:recover}, for every $k \in [m/2]$ let $\mathcal{P}_k$ be the property that
	\begin{equation}
		\label{eq:propertypk}
		\text{$\mathcal{U}_1^k,\mathcal{U}_2^k$ is an $(\epsilon,\phi,2)$-RED of  $G_k$}\, ,\quad \left| E^*_k \setminus \mathcal{U}^k_1\right| \ge \frac{4}{5} \cdot \left|E^*_k\right| \, , \quad |E_k^*| \ge \frac{9}{5}\frac{dn}{m}\,, \quad |E_k|\le 1.3 dn\, .
	\end{equation}
\end{definition}
\noindent
We first show a preliminary result: if there are many $k \in [m/2]$ that satisfy $\mathcal{P}_k$, then Bob's output set contains many edges of $G'$.
\begin{lemma}
	\label{lem:almostthere}
	In the setting of~\cref{red:recover}, if $\epsilon \in (0,1)$ satisfies \smash{$\epsilon \le \frac{1}{30\sqrt{m}}$}, then
	\begin{equation*}
			|E' \cap F|  \ge \frac{49}{100} \frac{dn}{m}  |\{k\in [m/2] : \mathcal{P}_k\}| \, .
	\end{equation*}
\end{lemma}
\begin{proof}
Recall from \cref{red:recover} that $F$ is made by the union of $F_k$ over 	$k \in [m/2]$. We then first give a guarantee on the $F_k$'s individually.
	\begin{claim}
		\label{claim:pkrecover}
		Let $k \in [m/2]$ such that $\mathcal{P}_k$ holds. Then
		\begin{enumerate}
			\item \label{pk:1} $|V \setminus \mathcal{U}_2^k| \le 3 \epsilon dn$, and
			\item \label{pk:2} $|F_k  \cap (E^*_k \setminus \mathcal{U}^k_1) | \ge \frac{dn}{m}-\epsilon^2 \cdot 1.3 dn$.
		\end{enumerate}
	\end{claim}
	\begin{proof}[Proof of~\eqref{pk:1}]
		Since $\mathcal{P}_k$ is verified, $\mathcal{U}_1^k,\mathcal{U}_2^k$ is an $(\epsilon,\phi,2)$-RED of  $G_k$ (see \cref{def:pk}). By definition of RED (see \cref{def:expdecrec}), we know $|E_k \setminus \mathcal{U}^k_1| \le \epsilon |E_k|$. Moreover, one can note that the number of non-isolated vertices in $ \mathcal{U}_2^k$ is at most twice the number of edges in $E_k \setminus  \mathcal{U}_1^k$. This is because every non isolated vertex must have at least one edge incident on it in $E_k \setminus  \mathcal{U}_1^k$ (which is the input to the second level ED), as otherwise property~\eqref{property:expanderclassic} of \cref{def:expdecclassic} would be violated for any $\phi > 0 $. Thus,
		\begin{equation*}
			|V \setminus \mathcal{U}_2^k| \le 2|E_k \setminus \mathcal{U}^k_1| \le 2\epsilon |E_k|  \le 3 \epsilon dn\, ,
		\end{equation*}
		where the last inequality follows because $|E_k|\le 1.3 dn$ (see \cref{def:pk}).
	\end{proof}
	\begin{proof}[Proof of~\eqref{pk:2}]
		From part~\eqref{pk:1}, together with the assumptions that $\mathcal{U}_1^k,\mathcal{U}_2^k$ is an $(\epsilon,\phi,2)$-RED of $G_k$ and $|E_k|\le 1.3 dn$ (see \cref{def:pk}), we have
		\begin{equation*}
			\left|V \setminus \mathcal{U}^k_2 \right|   \le 3\epsilon dn\quad \text{and} \quad \left|E_k \setminus \mathcal{U}_2^k \right| = \left|\left(E_k \setminus \mathcal{U}_1^k \right) \setminus \mathcal{U}_2^k\right| \le \epsilon \left|E_k \setminus \mathcal{U}_1^k \right| \le \epsilon^2 \cdot |E_k| \le \epsilon^2 1.3 dn\, .
		\end{equation*}
		We want to argue that many of the important edges $E^*_k$ of $G_k$ are present in the set of pairs $F_k$ (see \cref{red:recover}). To do this, we make the following observation for any edge $\{s,s'\} \in E^*_k \setminus \mathcal{U}^k_1 $: if none of $s,s'$ is an isolated vertex in $\mathcal{U}_2^k$, i.e. $\{s\} \notin \mathcal{U}_2^k$ and $\{s'\} \notin \mathcal{U}_2^k$, then $\{s,s'\} \in F_k$; otherwise, if $\{s\} \in \mathcal{U}_2^k$ or $\{s'\} \in \mathcal{U}_2^k$, then $\{s,s'\} \in(E_k \setminus \mathcal{U}_1^k) \setminus \mathcal{U}_2^k$ (and this is because  $\mathcal{U}_2^k$ is an ED of the graph $(V,E_k \setminus \mathcal{U}_1^k)$ and $\{s,s'\} \in E_k \setminus \mathcal{U}_1^k$ by assumption). Therefore,
		\begin{align}
			\label{eq:boundfk}
			|F_k \cap (E^*_k \setminus \mathcal{U}^k_1)| \ge |E^*_k \setminus \mathcal{U}^k_1| - |(E_k \setminus \mathcal{U}_1^k) \setminus \mathcal{U}_2^k| \, .
		\end{align}
		To conclude the claim, we use again the property $\mathcal{P}_k$: since $\mathcal{U}_1^k,\mathcal{U}_2^k$ is an $(\epsilon,\phi,2)$-RED of  $G_k$, we have
		\begin{equation*}
			\left|(E_k \setminus \mathcal{U}_1^k) \setminus \mathcal{U}_2^k \right| \le \epsilon |E_k \setminus \mathcal{U}_1^k| \le \epsilon^2 \cdot |E_k| \, .
		\end{equation*}
		Using again $\mathcal{P}_k$, we know $| E^*_k \setminus \mathcal{U}^k_1| \ge \frac{4}{5} \cdot |E^*_k|$,  $|E_k^*| \ge \frac{9}{5}{dn}/{m}$,  and $|E_k| \le  1.3 dn$, so from~\eqref{eq:boundfk} we conclude
		\begin{equation}
			|F_k \cap (E^*_k \setminus \mathcal{U}^k_1)| \ge |E^*_k \setminus \mathcal{U}^k_1| - |(E_k \setminus \mathcal{U}_1^k) \setminus \mathcal{U}_2^k| \ge  \frac{4}{5} \cdot |E^*_k|- \epsilon^2 1.3 dn \ge \frac{dn}{m}- \epsilon^2 1.3 dn\, .
		\end{equation}
	\end{proof}
	
	\noindent
	From the definition of Bob's output $F$ (see \cref{red:recover}) and by \cref{rmk:unionimportant}, one has
	\begin{align*}
		|E' \cap F|  = \left| \left(\bigcup_{k \in [m/2]} E_k^*\right) \cap  \left(\bigcup_{k \in [m/2] : |V \setminus \mathcal{U}_2^k| \le 3\epsilon dn} F_k \right) \right|
		& \ge \left| \bigcup_{k\in [m/2] : |V \setminus \mathcal{U}_2^k| \le 3\epsilon dn}  \left( E_k^* \cap  F_k \right) \right| \\
		\text{by \cref{rmk:unionimportant}} \quad & \ge \frac{1}{2} \sum_{k\in [m/2] : |V \setminus \mathcal{U}_2^k| \le 3\epsilon dn}  \left| E_k^* \cap  F_k \right| \, ,
	\end{align*}
	and we finish using \cref{claim:pkrecover}:
	\begin{align*}
		|E' \cap F|  & \ge \frac{1}{2} \sum_{k\in [m/2] : |V \setminus \mathcal{U}_2^k| \le 3\epsilon dn}  \left| E_k^* \cap  F_k \right|  \\
		\text{by part~\eqref{pk:1} of \cref{claim:pkrecover}} \quad & \ge  \frac{1}{2} \sum_{k\in [m/2] : \mathcal{P}_k  \text{ holds}}  \left| E_k^* \cap  F_k \right| \, ,
	\end{align*}
	and then
		\begin{align*}
		|E' \cap F|  \ge  \frac{1}{2} \sum_{k\in [m/2] : \mathcal{P}_k  \text{ holds}}  \left| E_k^* \cap  F_k \right|  \ge  \frac{1}{2} \sum_{k\in [m/2] : \mathcal{P}_k \text{ holds}}  \left| (E^*_k \setminus \mathcal{U}^k_1) \cap F_k\right| \, , 
	\end{align*}
	so by part~\eqref{pk:2} of \cref{claim:pkrecover} we have
	\begin{align*}
		|E' \cap F|  \ge \frac{1}{2} \sum_{k\in [m/2] : \mathcal{P}_k  \text{ holds}} \left(\frac{dn}{m} -\epsilon^2 \cdot 1.3 dn\right)   \ge \frac{49}{100} \frac{dn}{m}  |\{k\in [m/2] : \mathcal{P}_k\}| \, ,
	\end{align*}
	where we used $\epsilon \le \frac{1}{30\sqrt{m}}$ in the last inequality.
\end{proof}
\noindent
Finally, we prove that an efficient RED algorithm gives an efficient $\recover$ protocol, using the lemma we just showed.
\begin{lemma} \label{lem:ind-recover}
	Let $\epsilon,\phi \in (0,1)$ such that \smash{$\epsilon \le \min\{ 10^{-5}\psi \frac{d}{m}, \frac{1}{30\sqrt{m}}\}$} and \smash{$\phi \ge 11/m$}. If $n \ge 100 m$, $m\ge 500$, $d \ge 600 \log n$, and there is a deterministic $L$-bit space streaming algorithm $\mathcal{A}$ that computes an $(\epsilon,\phi,2)$-RED with probability at least $9/10$ over inputs $G  \sim \mathcal{G}$, then there is a deterministic protocol that solves $\recover_{\epsilon m}$ with probability at least $3/5$ over inputs $G' \sim \mathcal{G}'$ with communication complexity at most $L$.
\end{lemma}
\begin{proof}
	The protocol is given by \cref{red:recover}.  We observe that the communication complexity bound immediately follows, since $\mathcal{A}$ has space complexity $L$ by assumption, and Alice only sends to Bob the state of the algorithm.
	
	We now focus on proving correctness of the protocol with constant probability. We do this by first proving that $F$ satisfies each of the two properties of \cref{def:recover} provided that some conditions hold, and then verify that such conditions occur with constant probability.
	
	\begin{claim}[Property~\eqref{recover:fsmall} of \cref{def:recover}]
		\label{claim:deterministic1}
		If $\left|\left\{k \in [m/2]: \, \mathcal{P}_k\right\} \right| \ge\frac{m}{3}$, then $|F| \le 6 \cdot \epsilon m |E'| $.
	\end{claim}

	\begin{claim}[Property~\eqref{recover:flearns} of \cref{def:recover}]
		\label{claim:deterministic2}
		If $\left|\left\{k \in [m/2]: \, \mathcal{P}_k\right\} \right| \ge\frac{m}{3}$, then $|F \cap E'| \ge 1/10 \cdot |E'|$.
	\end{claim}
	
	\begin{claim}
		\label{claim:manypks}
		One has $
			\Pr_{G' \sim \mathcal{G}'}\left[\left|\left\{k \in [m/2]: \, \mathcal{P}_k\right\} \right| \ge {m}/{3}\right] \ge {3}/{5}$.
	\end{claim}
	\noindent
	By \cref{claim:deterministic1} and \cref{claim:deterministic2}, the protocol given by \cref{red:recover} succeeds provided that the input $G'$ satisfies $\left|\left\{k \in [m/2]: \, \mathcal{P}_k\right\} \right| \ge\frac{m}{3}$. By \cref{claim:manypks}, the latter condition is verified with probability at least $3/5$. The lemma statement is then proven, modulo showing \cref{claim:deterministic1}, \cref{claim:deterministic1}, and \cref{claim:manypks}, which we do next.
	
		\begin{proof}[Proof of \cref{claim:deterministic1}]
		As per \cref{red:recover}, $F$ is the union of $F_k=\{ \{s,s_{i,k}\} : i \in [n/m],\, s \in S_i \cap (V \setminus \mathcal{U}_2^k) \}$ over $k$ such that $|V \setminus \mathcal{U}_2^k| \le 3\epsilon dn$. Note that $|\{ \{s,s_{i,k}\} : i \in [n/m],\, s \in S_i \cap (V \setminus \mathcal{U}_2^k) \}|$ is in fact upper-bounded by $|V \setminus \mathcal{U}_2^k|$. One then gets that $|F| \le 3\epsilon dn \frac{m}{2}$. In order to relate this bound to $|E'|$, we use our assumption: for at least $m/3$ values of $k \in [m/2]$, $|E^*_k| \ge 9/5 \cdot dn/m$. By \cref{rmk:unionimportant}, this means $|E'| \ge 1/2 \cdot m/3 \cdot (9/5 \cdot dn/m) = 3/10 \cdot dn$, so one concludes $|F| \le 6 \epsilon |E'| m$.
	\end{proof}
	
		\begin{proof}[Proof of \cref{claim:deterministic2}]
		We can use \cref{lem:almostthere}, since our bound on $\epsilon$ ensures the precondition of the lemma. We then have $|E' \cap F|  \ge \frac{49}{100} \frac{dn}{m}  |\{k\in [m/2] : \mathcal{P}_k\}| $.  Now, we conclude the claim by virtue of our assumption:
		\begin{equation*}
			|E' \cap F| \ge \frac{49}{100} \frac{dn}{m}  |\{k\in [m/2] : \mathcal{P}_k\}|  \ge \frac{m}{3} \frac{49}{100} \frac{dn}{m} \ge \frac{1}{10} |E'| \, ,
		\end{equation*}
		where in the last inequality we used that our assumption further implies $|E'| \le |E_k| \le 1.3dn$ for some $k \in [m/2]$.
	\end{proof}

	\begin{proof}[Proof of \cref{claim:manypks}]
		We recall that for $G' \sim \mathcal{G}'$ and $K \sim \unif([m/2])$ one has $(G',K) \sim \mathcal{G}$ (see \cref{rmk:harddistrib}). For~$k \in [m/2]$, we then interchangeably write $\mathcal{A}(G)$ and $\mathcal{A}(G',k)$ to refer to an execution of $\mathcal{A}$ on the graph $G$ associated with $(G',k)$. Our assumption on $\mathcal{A}$ rewrites as
		\begin{equation*}
			\Pr_{G' \sim \mathcal{G}', \, K \sim \unif([m/2])}\left[\text{$\mathcal{A}(G',K)$ fails}\right]  = \Pr_{(G',K) \sim \mathcal{G}}\left[\text{$\mathcal{A}(G',K)$ fails}\right] \le \frac{1}{10} \, .
		\end{equation*}
		Now recall \cref{red:recover}: for each $k \in [m/2]$, the result of $\mathcal{A}_k$ computed by Bob is exactly $\mathcal{A}(G',k)$, where $G'$ denotes Alice's input. Therefore, using that $G'$ and $K$ are independent and that $K$ is drawn uniformly, we get
		\begin{equation*}
			\ex_{G' \sim \mathcal{G}'} \left[ \left| \left\{k \in [m/2] : \text{ $\mathcal{A}_k$ fails }\right\}\right| \right] =\ex_{G' \sim \mathcal{G}'} \left[ \left| \left\{k \in [m/2] : \text{ $\mathcal{A}(G',k)$ fails }\right\}\right| \right] \le \frac{1}{10} \cdot \frac{m}{2} \, .
		\end{equation*}
		Hence, if we sample $G' \sim \mathcal{G}'$, only a small fraction of the values of $k$ will make $(G',k)$ a failing instance. Specifically, for any $\alpha \in [1/10, 1]$, Markov's inequality gives
		\begin{equation*}
			\Pr_{G' \sim \mathcal{G}'}\left[\left| \left\{k \in [m/2] : \text{ $\mathcal{A}(G',k)$ fails }\right\}\right| > \frac{1}{\alpha}  \frac{1}{10} \cdot\frac{m}{2}\right] \le \alpha\, .
		\end{equation*}
		In other words, at least a $1-\alpha$ fraction of the probability measure of $\mathcal{G}'$ yields a graph $G'$ such that $(G',k)$ is a succeeding instance for $\mathcal{A}$ with at least a $1-0.1/\alpha$ fraction of the possible values for $k \in [m/2]$. By \cref{lem:leftexpander} (which applies since we assume $m\ge 500$, $d \ge 600 \log n$), we also have that a $1-1/n$ fraction of the measure of $\mathcal{G}'$ yields a graph $G'$ where each component is a $(1 \pm 1/10)2d$-regular $\frac{1}{3}$-expander. This in particular implies that at least a $1-\alpha-1/n$ fraction of the probability measure of $\mathcal{G}'$ yields a graph $G'$ such that:
		\begin{itemize}
			\item for every $k \in [m/2]$, the graph $G_k=(V,E_k)$ given by $(G',k)$ has $|E_k| \le 1.3 dn$, by near-regularity and \cref{def:harddistr};
			\item for every $k \in [m/2]$, the graph $G_k=(V,E_k)$ given by $(G',k)$ has $|E^*_k| \ge \frac{9}{5} \frac{dn}{m}$, by near regularity and \cref{def:harddistr};
			\item with at least a $1-0.1/\alpha$ fraction of the possible values for $k \in [m/2]$ we have the following: $\mathcal{A}_k$ outputs an $(\epsilon,\phi,2)$-RED $\mathcal{U}^k_1,\mathcal{U}^k_2$ of the graph $G=(V,E_k)$ given by $(G',k)$. Moreover, by \cref{lem:special-edge} (which applies, since our assumptions on $n$, $m$, $\epsilon$, $\phi$ meet its preconditions) one has that $\mathcal{U}_1^k$ verifies the property
			\begin{equation*}
				\left| E^*_k \setminus \mathcal{U}^k_1\right| \ge \frac{4}{5} \cdot \left|E^*_k\right|  \, .
			\end{equation*}
		\end{itemize}
		We have then obtained that at least a $1-\alpha-1/n$ fraction of the probability measure of $\mathcal{G}'$ yields a graph $G'$ such that $\mathcal{P}_k$ holds for at least a $1-0.1/\alpha$ fraction of the possible values for $k \in [m/2]$. Setting $\alpha = 3/10$ gives the claim.
	\end{proof}
\end{proof}

\subsection{Hardness of the communication problem}
\label{subsec:hardcomm}
We show that $\recover_{\xi}$ requires linear communication in the number of edges when~$G' \sim \mathcal{G}'$ and~$\xi = \epsilon m$. Intuitively, this is because Bob is only allowed to output an $\approx \epsilon m$ factor more pairs than~$|E'|$, and for an appropriate parameter regime \smash{$\epsilon m | E'| < \frac{1}{1000} \binom{m/2}{2} \frac{n}{m} $} (where \smash{$ \binom{m/2}{2} \frac{n}{m} $} is the total number of vertex pairs in a graph consisting of $n/m$ disjoint subgraphs on $m/2$ vertices).

\begin{lemma} \label{lem:lb-recover}
	Let $\epsilon \in (0,1)$ such that $\epsilon \le {10^{-33}}/{d}$. If $m\ge 500$ and $d \ge 600 \log n$, then any deterministic protocol that solves $\recover_{\epsilon m}$ with probability at least $3/5$ over inputs $G' \sim \mathcal{G}'$ requires~$\Omega(dn)$ bits of communication.
\end{lemma}
\begin{proof}
	Let us fix hereafter $\mathcal{Q}$ to be a deterministic protocol that solves $\recover_{\epsilon m}$ with probability at least $3/5$ over inputs $G' \sim \mathcal{G}'$. We define $\I^{(1)} \subseteq \supp(\mathcal{G}')$  to be the subset of input graphs $G'=(S,E')$ in the support of $\mathcal{G}'$ such that $\mathcal{Q}$ outputs correctly on input $G'$. We also partition $\supp(\mathcal{G}')$ based on the number of edges, so for every integer $t$ we define $\I_t  \subseteq \supp(\mathcal{G}')$ to be the subset of graphs $G'=(S,E') \in \supp(\mathcal{G}')$ with~$|E'| = t$. We also let for convenience \smash{$\I^{(1)}_t =\I^{(1)}  \cap \I_t$} be the subset of correct graphs with $t$ edges. For any $G'=(S,E') \in \supp(\mathcal{G}')$, we denote by $M_{G'} \in \{0,1\}^*$ the message sent by Alice on input~\smash{$G'$} using the protocol $\mathcal{Q}$.
		
	\begin{claim}
			\label{claim:mutualunif}
			If $t$ is an integer such that $t \in [9/10 \cdot dn/2 , \, 11/10 \cdot dn/2]$ and  $|\mathcal{I}^{(1)}_{t}| \ge {1}/{2} \cdot |\mathcal{I}_t| $, then
			\begin{equation*}
				\entropy\left(M_{G'} \, \middle| \, G' \in \I_t^{(1)}\right)  \ge \frac{dn}{3} \, .
			\end{equation*}
	\end{claim}
	\begin{proof}
		For any $G' \in \mathcal{I}^{(1)}_{t}$, following the reception of  $M_{G'}$, Bob deterministically outputs at most  $6\cdot(\eps
		m)\cdot|E'| = 6 \cdot \epsilon m t$ pairs of vertices \smash{$F~\subseteq~\binom{S}{2}$}, such that at least $1/10 \cdot |E'| = t/10$ of them are actually edges of $G'$. Thus, for any~$t$ and any message \smash{$M \in \{0,1\}^{\le nm}$} (where we restrained messages to be at most $nm$-bit long without loss of generality), the number of \smash{$G'\in \mathcal{I}^{(1)}_{t}$} such that \smash{$M_{G'}=M$} is at most
		\begin{align*}
			\left| \left\{G' \in \mathcal{I}^{(1)}_{t}: \, M_{G'} = M\right\}\right| & \le \sum_{j=0.1 t}^{t} \binom{6 \cdot \epsilon m t}{j} \cdot \binom{\frac{n}{m} \binom{m/2}{2}}{t-j} \\ 
			&\le \sum_{j=0.1 t }^{t} \left(\frac{\e \cdot 6 \cdot \epsilon m t}{j} \right)^j \cdot \left(\frac{\e\cdot nm/8}{t-j}\right)^{t-j}\\
			&= \e^{t} \sum_{j=0.1t}^{t}  \left(\frac{60 \cdot \epsilon m t}{ t} \right)^j \cdot \left(\frac{ nm/8}{t}\right)^{t-j} \cdot \left(\frac{t}{t-j}\right)^{t-j} \, .
		\end{align*}
		We employ the following fact to bound the above, whose proof is deferred to \cref{apndx:ineq}.
		\begin{restatable}{fact}{math}\label{prop:math1}
			For any $0 < x \le y$, $(y/x)^x \le \e^{y/\e}$.
		\end{restatable}
		\noindent
		We apply \cref{prop:math1} to the factor $(t/t-j)^{(t-j)}$ with $x = t-j$ and $y=t$, and hence we get
		\begin{align}
			\left| \left\{G' \in\mathcal{I}^{(1)}_{t}: \, M_{G'} = M\right\}\right| & \le \e^{(1+1/\e)t} \sum_{j=0.1t}^{t}  \left(60 \epsilon m \right)^j \cdot \left(\frac{nm/8}{t}\right)^{t-j}\\
			& =  \e^{(1+1/\e)t} \left(\frac{nm}{8t}\right)^{t} \sum_{j=0.1t}^{t}  \left(60\cdot \frac{\epsilon 8t}{n} \right)^j\\
		&\le  \e^{(1+1/\e)t} \left(\frac{nm}{8t}\right)^{t}  \left( {500} \cdot \frac{\epsilon t}{n} \right)^{ t/10} t \, , \label{eq:ubcorrcond}
		\end{align}
		where in the last inequality we used the upper bounds on $\epsilon$ and $t$. On the other hand, by virtue of our assumption that $|\mathcal{I}_{t}^{(1)}| \ge {1}/{2} \cdot |\mathcal{I}_t| $, one has
		\begin{align}
			\left|\mathcal{I}^{(1)}_{t}\right| & \ge \frac{1}{2} \left|\left\{G' = (S,E') \in \supp(\mathcal{G}'): \, |E'|  = t \right\}\right|  \\
			&  \ge \frac{1}{2} \left|\left\{G' = (S,E') \in \supp(\mathcal{G}'): \, \forall \, i \in [n/m] ,\, \left|E' \cap \binom{S_i}{2}\right|  = \frac{t}{n/m} \right\}\right| \\
			& = \frac{1}{2}  \binom{\binom{m/2}{2}}{\frac{t}{n/m}}^{n/m} \ge \frac{1}{2}  \binom{\frac{m^2}{16}}{\frac{t}{n/m}}^{n/m} \ge \frac{1}{2}  \left(\frac{\frac{m^2}{16}}{\frac{t}{n/m}}\right)^{\frac{t}{n/m} \cdot n/m} =  \frac{1}{2} \left(\frac{ nm}{16t}\right)^{t}\, . \label{eq:lbcorr}
		\end{align}
		Next, we rewrite
		\begin{align*}
			\entropy\left(M_{G'} \, \middle| \,  G' \in \I_t^{(1)}\right)   & \ge \info\left(G' ; M_{G'} \, \middle| \, G' \in \I_t^{(1)}\right) \\
			&  =  \entropy\left(G' \, \middle| \, G' \in \I_t^{(1)}\right)  -  \entropy\left(G'  \, \middle| \, M_{G'},  G' \in \I_t^{(1)}\right) \, .
		\end{align*}
		Now, observing that all the graphs with same number of edges have the same probability of being sampled, one can employ the upper and lower bounds derived above:
		\begin{align*}
			 & \hphantom{=} \entropy\left(G' \, \middle| \, G' \in \I_t^{(1)}\right)  -  \entropy\left(G'  \, \middle| \,M_{G'},  { G' \in \I_t^{(1)}}\right)\\
			   & \ge  \log\left(	\left|\mathcal{I}^{(1)}_{t}\right|\right) - \log \left(\max_{M \in \{0,1\}^{\le nm}} \left| \left\{H \in \mathcal{I}^{(1)}_{t}: \, M_H = M_{G'} \right\}\right|\right) \\
			\text{by \eqref{eq:lbcorr} and \eqref{eq:ubcorrcond}} & \ge t \log\left(\frac{ \frac{1}{2} \frac{nm}{16t}}{\e^{(1+1/\e)}\cdot \frac{nm}{8t} \cdot  {(t)}^{1/t}\cdot \left( 500\cdot \frac{\epsilon t}{n} \right)^{1/10} }\right) \, .
		\end{align*}
		Then, we finish using $9/10 \cdot dn/2 \le t \le 11/10 \cdot dn/2$ and $\epsilon \le {10^{-33}}/{d}$:
		\begin{align*}
		\entropy\left(M_{G'} \, \middle| \, G' \in \I_t^{(1)}\right)   \ge t \log\left(\frac{ 1/120}{\left( 500 \cdot \frac{\epsilon 11 dn}{20n} \right)^{1/10} }\right)  \ge t \log 6 \ge \frac{9}{10} \frac{dn}{2}\ge \frac{dn}{3}\, .
		\end{align*}
	\end{proof}
	
	\begin{claim}
		\label{claim:conditionlargeprob}
		One has
		\begin{equation*}
			\Pr_{G'=(S,E') \sim \mathcal{G}'}\left[  |E'| \in \left[\frac{9}{20}dn , \, \frac{11}{20}dn\right]  \, \mathrm{ and } \, \frac{\left|\I_{|E'|}^{(1)}\right|}{\left|\I_{|E'| }\right|} \ge  \frac{1}{2}\right] \ge \frac{1}{10} \, .
		\end{equation*}
	\end{claim}
	\begin{proof}
		We begin by defining for convenience
		\begin{equation*}
			x  = \Pr_{G'=(S,E') \sim \mathcal{G}'}\left[  \frac{\left|\I_{|E'|}^{(1)}\right|}{\left|\I_{|E'| }\right|} < \frac{1}{2}\right]  \, .
		\end{equation*}
		Then, by \cref{lem:leftexpander} (which applies by our assumption on $m$ and $d$), one has
		 \begin{equation}
		 	\label{eq:unionbound}
		 		\Pr_{G'=(S,E') \sim \mathcal{G}'}\left[  |E'| \in \left[\frac{9}{20}dn , \, \frac{11}{20}dn\right]  \, \mathrm{ and } \, \frac{\left|\I_{|E'|}^{(1)}\right|}{\left|\I_{|E'| }\right|} \ge  \frac{1}{2}\right]  \ge 1-\frac{1}{n} - x\, .
		 \end{equation}
		 We are left with the task of upper-bounding the probability $x$. Letting $\T = \{0, \dots, \binom{m/2}{2} \frac{n}{m}\}$ be  the set of possible numbers of edges in $G' \sim \mathcal{G}'$, we have
		 \begin{align*}
		 	& \hphantom{=} \Pr_{G'=(S,E') \sim \mathcal{G}'}\left[\mathcal{Q} \text{ outputs correctly on input } G'\right]\\
		 	& = \Pr_{G'=(S,E') \sim \mathcal{G}'}\left[G' \in \I^{(1)}\right] \\
		 	& =\sum_{t \in \T} \Pr_{G'=(S,E') \sim \mathcal{G}'}[|E'|=t] \Pr_{G'=(S,E') \sim \mathcal{G}'}[G' \in \I^{(1)} \, | \, |E'|=t] \\
		 	& = \sum_{t \in \T}\Pr_{G'=(S,E') \sim \mathcal{G}'}[|E'|=t]  \cdot \frac{\left|\I_{t}^{(1)}\right|}{\left|\I_{t}\right|} \, ,
		\end{align*}
		where the last equality holds because all the graphs with $t$  edges have the same probability of being sampled. Hence, splitting the sum one has
		\begin{align*}
			& \hphantom{=} \Pr_{G'=(S,E') \sim \mathcal{G}'}\left[\mathcal{Q} \text{ outputs correctly on input } G'\right]\\
		 	& \le \sum_{\substack{t \in \T : \\ |\I_t^(1)| < \frac{1}{2}|\I_t^(1)| }}  \Pr_{G'=(S,E') \sim \mathcal{G}'}[|E'|=t] \cdot \frac{1}{2} + \sum_{\substack{t \in \T : \\ |\I_t^(1)| \ge \frac{1}{2}|\I_t^(1)| }} \Pr_{G'=(S,E') \sim \mathcal{G}'}[|E'|=t]\cdot 1\\
		 	& = \Pr_{G'=(S,E') \sim \mathcal{G}'}\left[  \frac{\left|\I_{|E'|}^{(1)}\right|}{\left|\I_{|E'| }\right|} < \frac{1}{2} \right] \cdot \frac{1}{2} + \Pr_{G'=(S,E') \sim \mathcal{G}'}\left[ \frac{\left|\I_{|E'|}^{(1)}\right|}{\left|\I_{|E'| }\right|} \ge \frac{1}{2}\right] \cdot 1\\
		 	& = x \cdot \frac{1}{2} + (1-x) \cdot 1 = 1-\frac{x}{2}\, .
		 \end{align*}
		 If we suppose for the sake of a contradiction that $x \ge 449/500$, then
		 \begin{equation*}
		 	\Pr_{G'=(S,E') \sim \mathcal{G}'}\left[\mathcal{Q} \text{ outputs correctly on input } G'\right] \le \frac{551}{1000} < \frac{3}{5} \, ,
		 \end{equation*}
		 which contradicts our assumption on $\mathcal{Q}$. Therefore, we plug $x < 449/500$ into~\eqref{eq:unionbound} together with the assumption $n > m \ge 500$, thus  getting the claim.
	\end{proof}
	\noindent
	Let now $\mathcal{C}$ be the event that $|E'| \in [9/10 \cdot dn/2 , \, 11/10 \cdot dn/2]$ and  \smash{$|\mathcal{I}^{(1)}_{|E'|}| \ge {1}/{2} \cdot |\mathcal{I}_{|E'|}|$}.
	The size of Alice's message on input $G'=(S,E') \sim \mathcal{G}'$ is lower bounded by
	\begin{align*}
			\entropy\left(M_{G'}\right)
			& \ge \entropy\left(M_{G'} \, \middle| \, \ind_{ \mathcal{C}}, |E'|\right)  \\
			& \ge\Pr\left[  \mathcal{C} \right]    \sum_{\substack{t \in  \left[\frac{9}{20}dn , \, \frac{11}{20}dn\right] \cap \mathbb{N}:\\ \left|\mathcal{I}^{(1)}_{t}\right| \ge {1}/{2} \cdot \left|\mathcal{I}_t\right| } } \Pr\left[|E'|=t \, \middle| \, \mathcal{C} \right] \cdot \entropy\left(M_{G'} \, \middle| \, {\mathcal{C}}, |E'|=t\right) \,,
		\end{align*}
		and because conditioning does not increase entropy we rewrite
		\begin{align*}
				\entropy\left(M_{G'}\right) & \ge\Pr\left[  \mathcal{C} \right]     \sum_{\substack{\frac{9}{20}dn \le t \le  \frac{11}{20}dn :\\ |\mathcal{I}^{(1)}_{t}| \ge {1}/{2} \cdot |\mathcal{I}_t| } }\Pr\left[|E'|=t \, \middle| \, \mathcal{C} \right] \cdot \entropy\left(M_{G'} \, \middle| \, \ind_{G' \in  \I^{(1)}} ,{\mathcal{C}}, |E'|=t\right)  \\
			& \ge\Pr\left[  \mathcal{C} \right]     \sum_{\substack{\frac{9}{20}dn \le t \le  \frac{11}{20}dn:\\ |\mathcal{I}^{(1)}_{t}| \ge {1}/{2} \cdot |\mathcal{I}_t| } }\Pr\left[|E'|=t \, \middle| \, \mathcal{C} \right] \cdot \Pr\left[G' \in  \I^{(1)} \, \middle| \,  |E'|=t\right] \cdot\entropy\left(M_{G'} \, \middle| \, { G' \in \I_t^{(1)}}\right) \, .
		\end{align*}
		Finally, we use \cref{claim:mutualunif} and  \cref{claim:conditionlargeprob} and get
		\begin{align*}
			 	\entropy\left(M_{G'}\right) & \ge \Pr\left[  \mathcal{C} \right]     \sum_{\substack{\frac{9}{20}dn \le t \le  \frac{11}{20}dn:\\ |\mathcal{I}^{(1)}_{t}| \ge {1}/{2} \cdot |\mathcal{I}_t| } } \Pr\left[|E'|=t \, \middle| \, \mathcal{C} \right] \cdot \frac{1}{2} \cdot  \frac{dn}{3}  \\
			 	& =  \frac{dn}{6} \Pr_{G'=(S,E') \sim \mathcal{G}'}\left[  |E'| \in \left[\frac{9}{20}dn , \, \frac{11}{20}dn\right]  \, \mathrm{ and } \, \frac{\left|\I_{|E'|}^{(1)}\right|}{\left|\I_{|E'| }\right|} \ge  \frac{1}{2}\right] \\
			 	& \ge  \frac{dn}{60} \, .
	\end{align*}
\end{proof}

\subsection{Proving  the lower bound}
\label{subsec:lbproof}
Finally, we set $d$ and $m$ depending on $\epsilon,\phi,n$ so as to combine \cref{lem:ind-recover} and \cref{lem:lb-recover} and conclude our streaming lower bound.
\lb*

\begin{proof}
	We recall that for any $\epsilon = 1-\Omega(1)$, computing an $(\epsilon,\phi)$-ED with constant probability requires $\Omega(n \log n)$ bits for all $\phi \in (0,\epsilon]$. To see why, consider an input graph that is a matching of size $n/10$, and notice that an ED is a $1-\epsilon$ fraction the edges. Therefore, we now consider the case $\epsilon =  O(1/\log n)$ and prove an $\Omega(n/\epsilon)$ lower bound.
	
	Let $\epsilon,\phi \in (0,1)$ be any ED parameters and define $c = \psi \cdot 10^{-40}$. Note that $c$ is a constant since $\psi = \Omega(1)$ is the fixed sparsity from \cref{def:harddistr}. In order to apply  \cref{lem:ind-recover} and \cref{lem:lb-recover}, we use distributions  $\mathcal{G}$ and $\mathcal{G}'$ from \cref{def:harddistr} with $d$ and $m$ set as follows:
	\begin{equation}
		\label{set:dm}
		d = \frac{\epsilon}{2c^3 \phi } + \frac{c}{2 \epsilon} \quad \text{and} \quad m = \frac{11}{c \phi} \, .
	\end{equation}
	We verify that these values meet the preconditions set by \cref{lem:ind-recover} and \cref{lem:lb-recover}, i.e. we show that $d$ and $m$ satisfy
	\begin{equation}
		\label{condition:dm}
		\epsilon \le \min \left\{10^{-5} \psi \frac{d}{m} , \, \frac{1}{30\sqrt{m}}, \, \frac{10^{-33}}{d}\right\}\, , \quad \phi \ge \frac{11}{m}\, , \quad 500 \le m \le \frac{n}{100}\, , \quad 600\log n \le d < m \, ,
	\end{equation}
	for all $\epsilon,\phi$ such that
	\begin{equation}
		\label{setting:epsphi}
		\epsilon \le \frac{c^2}{\log n} \, , \quad \phi \ge \frac{1}{c^2n} \, , \quad \text{and} \quad \frac{\epsilon^2}{c^4} \le \phi \le \epsilon \, .
	\end{equation}
	Substituting $d$ and $m$ as defined in~\eqref{set:dm} into~\eqref{condition:dm} we get
	\begin{equation*}
		10^{-5} \psi \frac{d}{m} = 10^{-5} \psi \left(\frac{\epsilon}{2c^3 \phi } + \frac{c}{2 \epsilon} \right) \cdot \frac{c\phi}{11} \ge 10^{-5} \psi \frac{\epsilon}{2c^2 11 } = 10^{75} \frac{\epsilon}{2\psi 11 } \ge \epsilon \, ,
	\end{equation*}
\begin{equation*}
	\frac{1}{30\sqrt{m}} = \frac{\sqrt{c}\sqrt{\phi}}{30\sqrt{11}} \ge \frac{\sqrt{c}\epsilon}{30\sqrt{11}c^2} \ge \epsilon\, ,
\end{equation*}
\begin{equation*}
	\frac{10^{-33}}{d} = 10^{-33}\frac{4c^3\epsilon\phi}{2\epsilon^2 + 2c^4\phi}=10^{-33}\frac{4c^3\epsilon}{2\epsilon^2/\phi + 2c^4} \ge 10^{-33}\frac{4c^3\epsilon}{2c^4 + 2c^4} = 10^{-33} \frac{\epsilon}{c} \ge \epsilon \, ,
\end{equation*}
\begin{equation*}
	\frac{11}{m} = c{\phi} \le \phi \, ,
\end{equation*}
\begin{equation*}
	m = \frac{11}{c \phi} \ge \frac{11}{c \epsilon}  \ge \log n \ge 500 \, ,
\end{equation*}
\begin{equation*}
	m = \frac{11}{c \phi} \le 11 c n \le \frac{n}{100}\, ,
\end{equation*}
\begin{equation*}
	d = \frac{\epsilon}{2c^3 \phi } + \frac{c}{2 \epsilon} \ge \frac{c}{2 \epsilon}  \ge \frac{\log n}{2c} \ge 600 \log n \, ,
\end{equation*}
\begin{equation*}
	d = \frac{\epsilon}{2c^3 \phi } + \frac{c}{2 \epsilon} \le \frac{\epsilon}{2c^3 \phi } + \frac{c}{2 \phi}  \le \frac{1}{2c \log n \phi } + \frac{c}{2 \phi}  < \frac{11}{c \phi} = m \, ,
\end{equation*}
so all the conditions~\eqref{condition:dm} are verified.
	
	With any $\epsilon,\phi$ as in~\eqref{setting:epsphi} and for our choice of $d,m$ as in~\eqref{condition:dm}, \cref{lem:ind-recover} gives that if there is an $L$-bit space deterministic streaming algorithm $\mathcal{A}$ that outputs an $(\epsilon,\phi,2)$-RED of $G \sim \mathcal{G}_{n,d,m}$ with probability at least $9/10$, then there is a deterministic protocol $\mathcal{Q}$ that solves $\recover_{\epsilon m}$ with probability at least $3/5$ over inputs $G' \sim \mathcal{G}'$ with communication complexity $L$. On the other hand, by \cref{lem:lb-recover} we know that, in the parameter regime of~\eqref{condition:dm}, the complexity of $\mathcal{Q}$ over $ \mathcal{G}'$ is at least $\Omega(dn)$. Hence,
	\begin{equation*}
		L = \Omega(dn)  = \Omega\left(\left(\frac{\epsilon}{\phi}+\frac{1}{\epsilon}\right) \cdot n \right) = \Omega\left(\frac{n}{\epsilon}\right)\, .
	\end{equation*}
 	By Yao's minimax principle, the theorem statement is proved with $C = c^{-4}$.
	
\end{proof}

\newpage

\addcontentsline{toc}{section}{References}
\bibliographystyle{alpha}
\bibliography{refs}

\newpage

\appendix

\section{Technical facts}
\subsection{Polynomial time balanced sparse cut algorithm}
\label{apndx:selfloops}
In this section, we prove the following fast BSCA, used for our polynomial time BLD construction.

\fastbalsparse*

\begin{proof}
	First we show how to reduce the BSCA task on graphs with self-loops to graphs without self-loops. Let $G=(V,E,w)$ be a weighted undirected graph with one weighted self-loop per vertex, where $w: E \cup V \rightarrow \mathbb{R}_{\ge 0}$ and $E \subseteq \binom{V}{2}$. The idea is to construct a graph with two copies of every vertex, where the surplus vertices are solely connected to their corresponding vertex with weight equal to half the weight of the self-loop. If a vertex has no self-loops then we make only one copy for it. Formally, let $A=\{u \in V:w(u) > 0\}$ be the set of vertices with self-loops. Then let $B$ be a copy of $A$ representing the self-loops of vertices in $A$, and let $s: A \rightarrow B$ be a bijection mapping $u \in A$ to its corresponding self-loop vertex $s(u) \in B$. We define the self-loop free graph to be $\hat{G}=(\hat{V},\hat{E},\hat{w})$, where $\hat{V} = V \cup s(A)$, $\hat{E}=E \cup (\cup_{u \in A} \{\{u, s(u)\}\})$, and $\hat{w}(e)=w(e)$ for every $e \in E$ and $\hat{w}(\{u, s(u)\})=\frac{1}{2}w(u)$ for all $u \in A$. This is illustrated in \Cref{fig:selfloopsredcution}. We highlight that $\hat{G}$ has the same number of vertices and edges as $G$, up to constants.

\begin{figure}[h]
	\begin{minipage}[center]{\textwidth}
		\centering
		\includegraphics[scale=0.9]{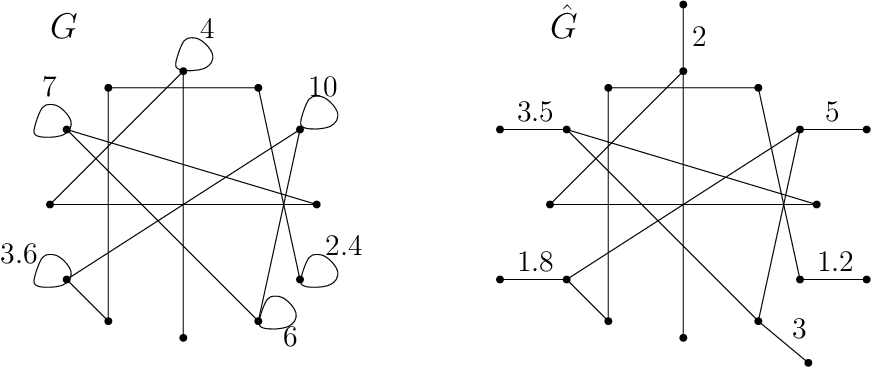}
	\end{minipage}
	\caption{On the left: the original graph with self-loops $G$. On the right: its self-loop free version~$\hat{G}$. The weights of the internal edges of $G$ are copied verbatim (not shown in the figure), while the weights of self-loops are halved.}
	\label{fig:selfloopsredcution}
\end{figure}

\noindent
We run a BSCA for self-loop free graphs on $\hat{G}$, and if it asserts $\hat{G}$ to be an expander we conclude the same for $G$, whereas if it returns a cut then we tweak it to get a solution for $G$. This is summarised in Algorithm~\ref{alg:selfloop}. The following claim gives the correctness of the reduction. The idea is that the subgraphs given by $\{u,s(u)\}$ for $u \in A$ are expanders, so there is no good reason for a sparse cut to cross them. Hence, $u$ and $s(u)$ should be on the same side of the sparse cuts.

\begin{algorithm}[h]
	\caption{\textsc{SelfLoopBalSparseCut}: reduction to self-loop free BSCA}
	\label{alg:selfloop}
	\begin{algorithmic}[1]
		\LeftComment{$G=(V,E,w)$ is a weighted undirected graph}
		\LeftComment{$\phi \in (0,1)$ is the sparsity parameter}
		\Procedure{\textsc{SelfLoopBalSparseCut}$(G,\phi)$}{}
		\State $\hat{\omega} \gets \bsc(\hat{G},\phi)$ \Comment{$\hat{G}$ is the self-loop free version of $G$ as described above}
		\If{$\hat{\omega} = \bot$}
			\State $\omega \gets \bot$
		\Else
		\State $(R,\nu) \gets \hat{\omega}$
		\State $X \gets R \cap V$
		\State $X^* \gets \argmin\{\vol_G(X),\vol_G(V \setminus X)\}$
		\State $\omega \gets (X^*, \vol_G(X^*))$
		\EndIf
		\State \Return $\omega$
		\EndProcedure
	\end{algorithmic}
\end{algorithm}

\begin{claim}
	\label{lem:selfloopfree}
	Let $\bsc$ be an $(\alpha,\lambda)$-BSCA. Then \textsc{SelfLoopBalSparseCut} is a $(2\alpha,4\lambda)$-BSCA, provided that the input sparsity parameter $\phi \in (0,1)$ satisfies $\phi \le \frac{1}{10\alpha}$.
\end{claim}
\begin{proof}
	For any cut $\emptyset \neq X \subsetneq V$, let us define for convenience $\partial_G X=w(X,V \setminus X)$ to be the total weight of the edges crossing the cut $X$ in $G$, and for any cut \smash{$\emptyset \neq R \subsetneq \hat{V}$} define \smash{$\partial_{\hat{G}} R=\hat{w}(R,\hat{V}\setminus R)$} to be the total weight of edges crossing $R$ in $\hat{G}$. A preliminary observation is that for any cut $\emptyset \neq X \subsetneq V$ we have that \smash{$\vol_{\hat{G}}(X \cup s(X \cap A))=\vol_G(X)$}, \smash{$\partial_{\hat{G}} (X \cup s(X \cap A)) = \partial_G X$}, and consequently \smash{$\Phi_{\hat{G}}(X \cup s(X \cap A)) = \Phi_G(X)$}. By definition of $\bsc$, we know that $\hat{\omega}$ is a $(\phi, \alpha,\lambda,0)$-BSCW of~$\hat{G}$, i.e.
	\begin{enumerate}
		\item \label{itdef:expapp} if $\hat{\omega}=\bot$, then $\hat{G}$ is a $\phi$-expander;
		\item if $\hat{\omega}=(R,\nu)$, then:
		\begin{enumerate}
			\item \label{itdef:sparseapp} $\Phi_{\hat{G}}(R) < \alpha \cdot \phi $,
			\item \label{itdef:smallsideapp} $\vol_{\hat{G}}(R) \le \vol_{\hat{G}}(\hat{V} \setminus R) $,
			\item \label{itdef:balapp} for every other cut $\emptyset \neq T \subsetneq \hat{V}$ with $\Phi_{\hat{G}}(T) < \phi$ and $\vol_{\hat{G}}(T) \le \vol_{\hat{G}}(\hat{V} \setminus T)$ one has $ \vol_{\hat{G}}(T) \le \lambda \vol_{\hat{G}}(R)$:
			\item \label{itdef:goodestimapp} $\nu =\vol_{\hat{G}}(R)$.
		\end{enumerate}
	\end{enumerate}
	Our goal is to translate these properties to $G$ by showing that $\omega$, as constructed in \textsc{SelfLoopBalSparseCut}, is a  $(\phi,2\alpha,4\lambda,0)$-BSCW of~$G$, i.e.
	\begin{enumerate}
		\item \label{itdef:expapp2} if $\omega=\bot$, then ${G}$ is a $\phi$-expander;
		\item if $\omega=(X^*,\nu)$, then:
		\begin{enumerate}
			\item \label{itdef:sparseapp2} $\Phi_{{G}}(X^*) < \alpha \cdot \phi $,
			\item \label{itdef:smallsideapp2} $\vol_{{G}}(X^*) \le \vol_{{G}}({V} \setminus X^*) $,
			\item \label{itdef:balapp2} for every other cut $\emptyset \neq T \subsetneq {V}$ with $\Phi_{{G}}(T) < \phi$ and $\vol_{{G}}(T) \le \vol_{{G}}({V} \setminus T)$ one has $ \vol_{{G}}(T) \le \lambda \vol_{{G}}(X^*)$:
			\item \label{itdef:goodestimapp2} $\nu =\vol_{{G}}(X^*)$.
		\end{enumerate}
	\end{enumerate}
	\paragraph{Correctness of the case $\hat{\omega}=\bot$.} This  means that for every cut $\emptyset \neq R \subsetneq \hat{V}$ one has $\Phi_{\hat{G}}(R) \ge \phi$. This in particular holds for cuts of the form $X \cup s(X \cap A)$ with $\emptyset \neq X \subsetneq V$. As we observed that $\Phi_{\hat{G}}(X \cup s(X \cap A)) = \Phi_G(X)$, we conclude that $G$ is also a $\phi$-expander.
	
	\paragraph{Correctness of the case $\hat{\omega}=(R,\nu)$.}
	The idea is that adding missing self-loops and removing dangling self-loops makes $R$ a solution for $G$. This is illustrated in \Cref{fig:selfloopscuts}.
	
	\begin{figure}[h]
		\begin{minipage}[center]{\textwidth}
			\centering
			\includegraphics[scale=0.9]{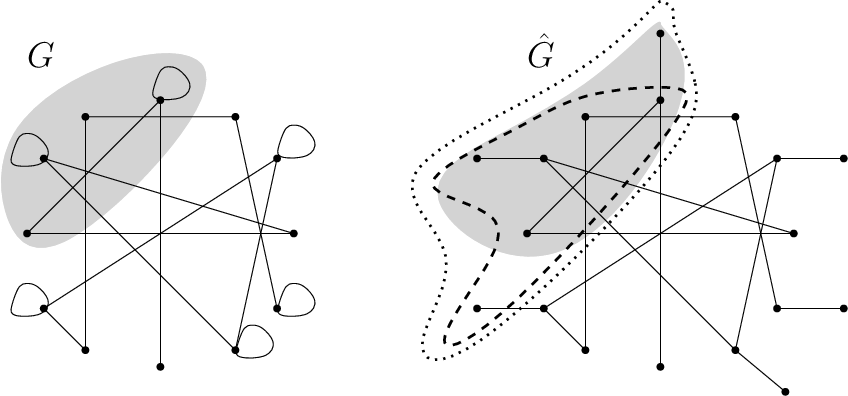}
		\end{minipage}
		\caption{There is a straightforward correspondence between the shaded cuts in $G$ and $\hat{G}$, since they have same volume and same number of crossing edges. The dashed cut has missing self-loops and dangling self-loops (think of it as the cut $R$). The dotted cut adds the missing self-loops (think of it as the cut $Q$). The shaded cut in $\hat{G}$ removes the dangling self-loops (think of it as the cut $P$), and corresponds to the shaded cut in $G$ (think of it as the cut $X$).}
		\label{fig:selfloopscuts}
	\end{figure}

	\noindent
	We know from~\eqref{itdef:sparseapp} that
	\begin{equation*}
		\frac{\partial_{\hat{G}} R}{\vol_{\hat{G}} (R)} < \alpha \phi \, .
	\end{equation*}
	Consider the cut $Q=R \cup s(R \cap A)$, i.e. a version of $R$ where we have added the self-loops of all the vertices in $R$ that have one. Then $Q \subsetneq \hat{V}$, as otherwise $V \subseteq R$, which would imply that $R$ contains more than half the volume of $\vol_{\hat{G}}(\hat{V})$, contradicting~\eqref{itdef:smallsideapp}. Next, we remark that the volume of $Q$ cannot decrease compared to $R$, since we added vertices, i.e.
	\begin{equation}
		\label{eq:volrq}
		\vol_{\hat{G}}(Q) = \vol_{\hat{G}}(R) + \partial_{\hat{G}}(s(R \cap A) \setminus (R \cap B)) \ge \vol_{\hat{G}}(R)\, .
	\end{equation}
	At the same time, we do not increase the number of crossing edges since we move the self-loop edges inside the cut, so
	\begin{equation*}
		\frac{\partial_{\hat{G}} Q}{\vol_{\hat{G}}(Q)} = \frac{\partial_{\hat{G}} R - \partial_{\hat{G}}(s(R \cap A) \setminus (R \cap B))}{\vol_{\hat{G}}(R) + \partial_{\hat{G}}(s(R \cap A) \setminus (R \cap B))} \le \frac{\partial_{\hat{G}} R}{\vol_{\hat{G}} (R)} \, .
	\end{equation*}
	Thus, we have added self-loops to $R$ without increasing its sparsity nor decreasing its volume. Now we remove the dangling self-loops from $Q$, i.e. we remove the set $D=(Q \cap B) \setminus s(Q \cap A)$, and consider the cut $P=Q \setminus D$. We remark that $P \neq \emptyset$, as otherwise $R$ is solely made of vertices from~$s(A)$, which contradicts $\Phi_H(R) < \alpha \phi \le 1/10 < 1$. Since we removed an equal amount of mass from both the crossing edges and the volume, we get
	\begin{equation*}
		\frac{\partial_{\hat{G}} P}{\vol_{\hat{G}}(P)} \le \min\left\{\frac{\partial_{\hat{G}} P}{\vol_{\hat{G}}(P)}, 1\right\} \le \frac{\partial_{\hat{G}} P + \partial_{\hat{G}} D}{\vol_{\hat{G}}(P)+\partial_{\hat{G}} D}  = \frac{\partial_{\hat{G}} Q}{\vol_{\hat{G}} (Q)} \, .
	\end{equation*}
	Thus, removing the dangling self-loops does not increase sparsity. However we may have decreased the volume, which threatens the balancedness property of the cut. We then bound the decrease: because we proved before ${\partial_{\hat{G}} Q}/{\vol_{\hat{G}} (Q)} < \alpha \phi$, we know
	\begin{equation*}
		\partial_{\hat{G}} D < \frac{\alpha \phi}{1-\alpha\phi}\vol_{\hat{G}}(P) \, ,
	\end{equation*}
	so
	\begin{equation}
		\label{eq:volpq}
		\vol_{\hat{G}}(Q) = \vol_{\hat{G}}(P) + \partial_{\hat{G}} D < 2 \vol_{\hat{G}}(P) \, .
	\end{equation}
	Now one can see that $P$ is identical to $(R \cap V) \cup s(R \cap A)$, and $X$ as defined in Algorithm~\ref{alg:selfloop} equals $R \cap V$. Hence $\partial_G X = \partial_{\hat{G}} P$ and $\vol_G(X)=\vol_{\hat{G}}(P)$. Therefore,
	\begin{equation}
		\label{eq:smallratio}
		\frac{\partial_G X}{\vol_G (X)} < \alpha \phi \, .
	\end{equation}
	To handle the case where $X$ is the larger volume side of the cut, \cref{alg:selfloop} takes $X^*$ to be the side of $(X,V\setminus X)$ with smaller volume. We show that the volume of $V\setminus X$ is within a small constant factor of the volume of $X$. To do so, recall that $\Phi_{\hat{G}}(R)<\alpha \phi$, so
	\begin{equation*}
		\frac{\partial_{\hat{G}} Q + \partial_{\hat{G}}(s(R \cap A) \setminus (R \cap B))}{\vol_{\hat{G}}(Q) - \partial_{\hat{G}}(s(R \cap A) \setminus (R \cap B))} < \alpha \phi \,,
	\end{equation*}
	which implies
	\begin{equation*}
		\partial_{\hat{G}}(s(R \cap A) \setminus (R \cap B)) < \alpha \phi \vol_{\hat{G}}(Q) \, .
	\end{equation*}
	Hence, by property~\eqref{itdef:smallsideapp} of $R$, we have
	\begin{equation*}
		\vol_G(X) \le \vol_{\hat{G}}(Q) < \frac{1}{1-\alpha\phi}\vol_{\hat{G}}(R) \le \frac{1}{2} \cdot  \frac{1}{1-\alpha\phi} \vol_{\hat{G}}(\hat{V}) = \frac{1}{2} \cdot  \frac{1}{1-\alpha\phi}\vol_G(V) \, .
	\end{equation*}
	Therefore, we have that $X^* = \argmin\{\vol_G(X), \vol_G(V\setminus X)\}$ satisfies
	\begin{equation}
			\label{eq:nearlysamevol}
		(1-2\alpha\phi)\vol_G(X) \le \vol_G(X^*) \le \vol_G(V\setminus X^*) \, ,
	\end{equation}
	and
	\begin{equation}
		\label{eq:largevol}
		\vol_{\hat{G}}(R) < \frac{2}{1-2\alpha \phi}\vol_G(X^*) \, ,
	\end{equation}
by~\eqref{eq:volrq},~\eqref{eq:volpq}. We now conclude by verifying properties~\eqref{itdef:sparseapp2},~\eqref{itdef:smallsideapp2},~\eqref{itdef:balapp2},~\eqref{itdef:goodestimapp2}.

\begin{itemize}
	\item \textit{Property~\eqref{itdef:sparseapp2}.} From~\eqref{eq:smallratio} and~\eqref{eq:nearlysamevol}, we get $\Phi_G(X^*) < 	\alpha\phi /({1-2\alpha \phi}) \le 2 \alpha \phi$.
	\item \textit{Property~\eqref{itdef:smallsideapp2}.} This property follows by definition of $X^*$ in \cref{alg:selfloop}.
	\item \textit{Property~\eqref{itdef:balapp2}.} Consider any cut $\emptyset \neq T \subsetneq V$ such that $\Phi_G(T) < \phi$ and $\vol_G(T) \le \vol_G(V \setminus T)$. Then we know that
	\begin{equation*}
		\Phi_G(T) = \Phi_{\hat{G}}(T \cup s(T \cap A)) < \phi \, ,
	\end{equation*}
	and
	\begin{equation*}
		\vol_{\hat{G}}(T \cup s(T \cap A)) = \vol_G(T) \le \vol_G(V \setminus T) = \vol_{\hat{G}}(\hat{V} \setminus (T \cup s(T \cap A))\, ,
	\end{equation*}
	as well as
	\begin{equation*}
		\vol_G(T) = \vol_{\hat{G}}(T \cup s(T \cap A)) \le \lambda \cdot \vol_{\hat{G}}(R) < \frac{2\lambda}{1-2\alpha \phi}\vol_G(X^*) \le 4\lambda \vol_G(X^*)\, ,
	\end{equation*}
	by an application of~\eqref{itdef:balapp} to the cut $T \cup s(T \cap A)$ and~\eqref{eq:largevol}.
	\item \textit{Property~\eqref{itdef:goodestimapp2}.} Follows directly by definition of $\omega$ in \cref{alg:selfloop}.
\end{itemize}
\end{proof}
\noindent
Using \cref{lem:selfloopfree}, we now conclude a proof of \cref{th:fastbalsparse}. To do so, we run two different algorithms, depending on how large $p$ and $q$ are.

\paragraph{The large case.}
An algorithm was recently given that, for $p$-vertex $q$-edge input graphs with edge weights bounded by $\poly(p)$ produces a $(O(\phi \log^3 p), \phi)$-ED\footnote{In this case, the definition of ED uses the weight of edges across cuts, and weighted volumes, including self-loops.} with probability $1-1/\poly(p)$ in $(p+q)\cdot \pylog(p)$ time, for any $\phi \in (0,1)$~\cite{liexpdec}. One can then use this ED algorithm to efficiently compute an approximate balanced $\psi$-sparse cut for all $\psi \in (0,1)$: as it was observed in~\cite{offlineexpdec}, given access to a graph and its $(O(\psi \log^3p), \psi)$-ED, one can easily obtain a $(\psi,c' \cdot \log^3 p, C',0)$-BSCW for some constants $c',C'\ge 1$. This gives a $(c' \cdot \log^3 p, C')$-BSCA that runs in $(p+q)\cdot \pylog(p)$ time for all input sparsity parameters~$\psi \in (0,1)$.

Whenever $p \ge n^{1/c_{\text{time}}}$, for a constant $c_{\text{time}} \ge 1$ to be determined later, we then apply the algorithm of \cite{offlineexpdec,liexpdec} together with \cref{lem:selfloopfree} setting $\alpha=c' \log^3 n$ and $\lambda = C'$ (where $c',C'\ge 1$ are the constants one gets from the self-loop free BSCA of~\cite{offlineexpdec,liexpdec}). Then, we get a $(c_{\sw} \log^3 n,C_{\sw})$-BSCA where $c_{\sw}=2c'$ and $C_{\sw} = 4C'$, provided that \smash{$\phi \le \frac{1}{10\alpha}$}. Using $\phi = \psi$ and since we impose \smash{$\psi \le  \frac{1}{10 c_{\sw} \log^3 n}$}, we get \smash{$\phi \le \frac{1}{10\alpha}$}.

Using our assumption on $p,q,W$, the space and random bit complexity trivially follow from the bound on the running time, which is in turn at most $\poly(n)$. The success probability is $1-1/\poly(p)$, which gives the claimed success probability of $1-1/\poly(n)$.

\paragraph{The small case.}
When $p < n^{1/c_{\text{time}}}$, we use an algorithm running in $T(p,q,W)$ time. The algorithm consists in using an algorithm that $(c'' \log p)$-approximates the sparsity of the input graph (such as the deterministic polynomial time algorithm of \cite{leightonrao,llr}, which allows to have weighted self-loops by incorporating them in the demand graph) as follows (which is essentially the idea of \cite{offlineexpdec}): maintain a cut $R$ initially empty and repeatedly compute the approximate sparsest cut $S$ of $G[V\setminus R]^1$; if $S$ has sparsity at least $2\psi c'' \log p$ then terminate, else update $R$ to be $R \cup S$; if  $R$ accounts for at least a $1/5$ fraction of the total volume of $G$ then terminate, else repeat.

First note that this algorithm terminates (e.g. when $V \setminus R$ is a single vertex). Also observe that if $R = \emptyset$ then we can conclude that $G$ is a $2\psi$-expander. Otherwise, $R$ has always sparsity at most $4\psi c'' \log p$ in $G$ and one of $R,V\setminus R$ hs volume at most half of the total. Moreover, if the smaller of $R,V\setminus R$ accounts for at least a $1/5$ fraction of the total volume of $G$ then every $\psi$-sparse cut with volume at most half the total is within a $5/3$ fraction of the volume of the smaller of $R,V\setminus R$. Finally, if $R$ accounts for less than a $1/5$ fraction of the total volume of $G$, then $G[V\setminus R]^1$ is a $2\psi$-expander; consider a cut $T$ with sparsity at most $\psi$ and volume at most half the total; then if $T$ had more than half its volume in $V \setminus R$, the weight of edges crossing the cut $T$ in $G$ would be more than a $\psi$ fraction of the volume of $T$, a contradiction; hence, $T$ is within a factor of $2$ of the volume of $R$. This gives a $(c'' \log n,2)$-BSCA.

We pick $c_{\text{time}}$ so that $T(p,q,W)$ is at most \smash{$\otil(n/b^2)\cdot \log^{O(\frac{\log n}{\log 1/b})}n$}, thus giving the claimed bound on the space and random bit complexity. In our case, $T(p,q,W)=\poly(p,q,\log W)$ (which can be achieved by solving the linear programming relaxation of \cite{leightonrao} with the ellipsoid method and rounding it deterministically with the algorithm of \cite{llr}) and recall that $W=\poly(n)$, so there exists a constant $c_{\text{time}} \ge 1$ such that $T(p,q,W)$ is at most \smash{$n/b^2 \cdot \log^{O(\frac{\log n}{\log 1/b})}n$}. The running time also follows, and the success probability is $1$.
\\~\\
 \cref{th:fastbalsparse} is concluded by combining the large and small case.

\end{proof}

\subsection{Expansion and regularity of random graphs}
\label{apndx:randomgraphs}
\erdos*
\begin{proof}
	Let $H=(U,F)$ be a sample from $\er(N,p)$, so that $|U|= N$. Also let $L=D-A$ be the Laplacian matrix of $H$, where $D$ denotes its degree diagonal matrix and $A$ denotes its adjacency matrix. We use the following matrix Bernstein concentration bound on $L$ to lower bound its second eigenvalue, and then lower bound the expansion of $H$ using Cheeger's inequality.
	\begin{theorem}[Theorem 1.6.2 in~\cite{matcon}]
		\label{th:bern} Let $S_1,\dots, S_\kappa \in \mathbb{R}^{N \times N}$ be independent symmetric random matrices such that $\ex[S_\iota]=0$ and $\nm{S_\iota} \le \lambda$ for every $\iota \in [\kappa]$, and let $S = \sum_{\iota=1}^{\kappa} S_\iota$. Then, for all~$\theta \ge 0$
		\begin{equation*}
			\Pr[\nm{S} \ge \theta] \le 2N \cdot \exp\left(-\frac{\theta^2/2}{\nm{ \sum_{\iota=1}^{\kappa} \ex[S_\iota^2]}+\lambda\theta/3}\right) \, .
		\end{equation*}
	\end{theorem}
	\noindent
	We start by rewriting $L$ so that we can apply the above result. For each $e=\{u,v\} \in \binom{V}{2}$ let $z_e \sim \ber(p)$ be an independent Bernoulli random variable with bias $p$, and also let $L_e $ be the Laplacian matrix of the graph $(U,\{e\})$, i.e. $(L_e)_{ab} = 1$ if $a=b=u$ or $a=b=v$, $(L_e)_{ab} = -1$ if $\{a,b\}=e$, and $(L_e)_{ab}=0$ otherwise. Define $S_e = z_e L_e - \ex[z_eL_e]$ and \smash{$S = \sum_{e \in \binom{V}{2}} S_e$}. One can now note that
	\begin{equation*}
		\ex[S_e]=0 \quad \text{and} \quad \nm{S_e} \le \max\left\{\nm{(1-p)L_e},\nm{-pL_e}\right\} \le 2\quad \text{for all $e\in \binom{V}{2}$.}
	\end{equation*}
	Letting $J$ be the all-ones $N \times N$ matrix we  also have
	\begin{equation*}
		S = \sum_{e \in \binom{V}{2}} z_e L_e - \sum_{e \in \binom{V}{2}} \ex[z_eL_e] = L- p \cdot (N \cdot I-J)\, .
	\end{equation*}
	Hence, \cref{th:bern} with $\lambda=2$ gives
	\begin{equation}
		\label{bern:appl}
		\Pr[\nm{L-p \cdot (N \cdot I-J)} \ge \theta] \le 2N \cdot \exp\left(-\frac{\theta^2/2}{\nm{ \sum_{e \in \binom{V}{2}} \ex[(z_e L_e - \ex[z_eL_e])^2]}+2\theta/3}\right) \,.
	\end{equation}
	Now we observe that the first term in the denominator can be rewritten as
	\begin{align*}
		\nm{ \sum_{e \in \binom{V}{2}} \ex\left[(z_e L_e - \ex[z_eL_e])^2\right]} & = \nm{ \sum_{e \in \binom{V}{2}} 2L_e\ex\left[(z_e-z_e\cdot p+p^2)\right]} \\
		& = \nm{ \sum_{e \in \binom{V}{2}} 2pL_e} \\
		& = 2p \nm{N \cdot I-J} \\
		& = 2pN	\, ,
	\end{align*}
	where the last equality holds because $K \coloneqq N \cdot I-J$ is the Laplacian matrix of the complete graph on $N$ vertices. If we now set $\theta = pN/10$ in~\eqref{bern:appl}, we get
	\begin{equation}
		\label{eq:specapprox}
		\Pr\left[\nm{L-p \cdot K} \ge \frac{1}{10}pN\right] \le  2N \cdot \exp\left(-\frac{pN}{600}\right) \, .
	\end{equation}
	Now, in order to use a Cheeger type of inequality, it is convenient to view $H$ as a regular graph, so that its normalized Laplacian matrix is simply a scalar multiple of $L$. If we denote with $\deg(u)$ the degree of $u \in U$ in $H$, we get by a Chernoff bound that
	\begin{equation}
		\label{eq:degrees}
		\Pr\left[\text{$\exists u \in U$, } \left|\deg(u) - p\cdot (N-1)\right| > \frac{1}{14}p \cdot (N-1)\right] \le 2N \cdot \exp\left(-\frac{p\cdot (N-1)}{600}\right) \,.
	\end{equation}
	From~\eqref{eq:specapprox} and~\eqref{eq:degrees}, we then have
	\begin{equation}
		\label{eq:sectolast}
			\Pr\left[\left(\nm{L-p K} \ge \frac{1}{10}pN\right)\text{ or } \left(\exists u \in U : \left|\deg_H(u)-\bar{d}\right|>\frac{1}{14} \bar{d}\right)\right] \le 4N \cdot \exp\left(-\frac{pN}{600}\right) \, .
	\end{equation}
	We conclude by showing that if the event in the above probability does not occur, then the event argument of the probability in the statement of \cref{prop:er} does not occur either. To see this, let $\lambda_2$ be the second smallest eigenvalue of $1/\bar{d}  \cdot L$, and $\mu_2$ be the second smallest eigenvalue of $1/(N-1) \cdot K$. Hence, $\nm{L-p K} < \frac{1}{10}pN$ implies $\lambda_2 \ge \mu_2-1/9$ by Weyl's inequality. Moreover, we know that $\mu_2 = N/(N-1)$, so $\lambda_2 \ge 9/10$. Finally, let $\emptyset \neq S \subsetneq U$ be a cut with $\vol(S) \le \vol(U\setminus S)$ such that $\Phi_H(S) = \Phi_H$. Assuming $\left|\deg_H(u)-\bar{d}\right|\le \frac{1}{14} \bar{d}$ for all $u \in U$, we have that $\vol(S) \in [(1-1/14)\bar{d} |S|, (1+1/14)\bar{d} |S|]$. Let $\vt{y} \in \mathbb{R}^U$ be defined as $y_u = 1/|S|$ if $u \in S$ and $y_u=-1/|U \setminus S|$ otherwise. By Courant-Fischer we then have
	\begin{align}
		\label{eq:courfisch}
		\lambda_2 = \min_{\substack{\vt{x} \in \mathbb{R}^U\setminus \{\vt{0}\}:\\ \langle \vt{x}, \ind \rangle= 0}} \, \frac{1}{\bar{d}} \cdot \frac{\vt{x}^\top L\vt{x}}{ \vt{x}^\top \vt{x}} \le \frac{1}{\bar{d}} \cdot  \frac{\vt{y}^\top L\vt{y}}{\vt{y}^\top \vt{y}}  = \frac{\uncut{S}{U}}{\bar{d}|S|} \cdot \frac{|U|}{|U\setminus S|}  \le \left(1+\frac{1}{14}\right)\Phi_H \cdot \frac{|U|}{|U\setminus S|} \, .
	\end{align}
	Again assuming $\left|\deg_H(u)-\bar{d}\right|\le \frac{1}{14} \bar{d}$ for all $u \in U$, one can check that $\vol(S) \le \vol(U\setminus S)$ implies $|U\setminus S| \ge 6/7 \cdot |S|$. Thus,~\eqref{eq:courfisch} gives
	\begin{equation*}
		\lambda_2 \le \left(1+\frac{1}{14}\right)\cdot \left(1+\frac{7}{6}\right) \Phi_H \le \frac{27}{10} \Phi_H \, ,
	\end{equation*}
	so $\Phi_H \ge 1/3$.
\end{proof}

\subsection{A useful inequality}
\label{apndx:ineq}
\math*
\begin{proof}
	Write $x=y\cdot \e^{\rho-1}$ for some $\rho \le 1$. Then $(y/x)^x \le \e^{y/\e}$ if and only if $(1-\rho)e^{\rho-1} \le 1/\e$. Define $f(\rho)=(1-\rho)\e^{\rho}$. We prove the claim by showing that $f(\rho) \le 1$ for all $ \rho \in (-\infty, 1]$.
	
	We begin by noting that $f$ is continuous over $(-\infty, 1]$. Then, we take the derivative of $f$, and we get $f'(\rho)=-\rho \e^\rho$. We observe that $f'(\rho) > 0$ for $\rho <0$, and $f'(\rho) > 0$ for $\rho >0$. We further observe that $f(0) = 1$, $f(1)=0$, and $\lim_{\rho \rightarrow -\infty} f(\rho) = 0$. Hence, $f(\rho)$ is upper-bounded by $1$ for all~$\rho \in (-\infty, 1]$.
\end{proof}

\end{document}